\documentclass[12pt,a4wide]{book}
\usepackage{graphicx}
\usepackage{hyperref}
\usepackage{amsmath}
\usepackage{makecell}
\usepackage{parskip}
\usepackage{mathtools}
\usepackage{float}
\usepackage{algorithm}
\usepackage{algpseudocode}
\usepackage{enumerate}
\usepackage{adjustbox}
\usepackage{longtable}
\usepackage{enumitem}
\usepackage{booktabs}
\usepackage{amssymb}
\usepackage{pdfpages}
\usepackage{adjustbox}
\usepackage{rotating}
\usepackage{amssymb}
\usepackage{mathrsfs}
\usepackage{multirow}
\usepackage[font=small,labelfont=bf]{caption}
\usepackage[font=footnotesize,caption=false]{subfig}
\usepackage{ lscape}
\usepackage{etoolbox}
\usepackage{array, booktabs, caption}
\usepackage[font=small,labelfont=bf]{caption}
\usepackage[font=footnotesize,caption=false]{subfig}
\usepackage{makecell}

\usepackage[T1]{fontenc}
\usepackage{algpseudocode}
\usepackage[symbol]{footmisc}

\usepackage{makecell}
\usepackage{mathtools}
\usepackage[english]{babel}
\usepackage[T1]{fontenc}
\usepackage{enumerate}
\usepackage{longtable}
\usepackage{enumitem}
\usepackage{booktabs}
\usepackage{amssymb}
\usepackage{makecell}

\usepackage[T1]{fontenc}
\usepackage{color}
\usepackage{xcolor}

\usepackage{graphicx}
\usepackage{mathtools}
\usepackage{amsfonts}
\usepackage{amssymb}
\usepackage{graphics}
\usepackage{color}
\usepackage{pgfplots}
\usepackage{hyperref}
\usepackage{listings}

\usepackage{booktabs}
\usepackage{pgfplots}
\usepackage{amssymb}
\usepackage{mathrsfs}
\usepackage{adjustbox}
\usepackage{amsmath}
\usepackage{hyperref}
\usepackage{longtable}
\usepackage{multirow}
\usetikzlibrary{arrows,automata}

\usepackage{float}
\usepackage[caption = false]{subfig}

\usepackage{listings}

\usepackage{url}
\usepackage{multirow}
\usepackage{multicol}
\usepackage{amssymb}
\usepackage{mathrsfs}
\usepackage{enumerate}
\usepackage{enumitem}
\newsavebox{\measurebox}

\usepackage{cite}
\usepackage{epsf} 
\usepackage{latexsym} 
\usepackage{graphicx}
\usepackage{amssymb}
\usepackage{amsmath,amssymb}
\usepackage{amsfonts}
\usepackage{fancyhdr}
\usepackage{enumitem}
\usepackage{chngcntr}
\usepackage{floatpag}
\usepackage{latexsym}
\usepackage{tabularx, booktabs} 
\usepackage[title,titletoc,toc]{appendix}
\usepackage{subfig}
\usepackage{amsmath}
\usepackage{mathrsfs}
\usepackage{afterpage}
\newcommand\blankpage{%
    \null
    \thispagestyle{empty}%
    \addtocounter{page}{-1}%
    \newpage}
    
\makeatletter
\newcommand\fs@norules{\def\@fs@cfont{\bfseries}\let\@fs@capt\floatc@ruled
  \def\@fs@pre{}%
  \def\@fs@post{}%
  \def\@fs@mid{\kern3pt}%
  \let\@fs@iftopcapt\iftrue}      
 
\usepackage{longtable}  
  
\usepackage{pgfplots}  
\usepackage{multirow}
  
\usepackage{colortbl}

\definecolor{webgreen}{rgb}{0, 0.5, 0} 
\definecolor{webblue}{rgb}{0, 0, 0.5} 
\definecolor{webred}{rgb}{0.0, 0, 0} 
\definecolor{webred}{rgb}{0.5, 0, 0} 
\definecolor{webblack}{rgb}{0, 0, 0} 

\usepackage{booktabs}
\usepackage{makecell}

\usepackage{longtable}

\fancyhfoffset[L]{0.5cm}

\makeatother

\usepackage[toc,nonumberlist]{glossaries}
\makeglossaries
\glsdisablehyper

\usepackage{datetime2}

\makeatletter
\newcommand{\monthyeardate}{%
  \DTMenglishmonthname{\@dtm@month}, \@dtm@year
}
\makeatother

\newcolumntype{Y}{>{\centering\arraybackslash}X}
\newtheorem{theoremm}{Theorem}[chapter]
\newtheorem{eqed}{Example}[chapter]
\newtheorem {lemmaa}{Lemma}[chapter]

\newtheorem{defnn}{Definition}[chapter]
\newtheorem {corollaryy}{Corollary}[chapter]
\newtheorem {axiomm}{Axiom}[chapter]
\newtheorem {inferencee}{Inference}[chapter]

\newtheorem {hypothesiss}{Hypothesis}[chapter]
\newtheorem {conjecturee}{Conjecture}[chapter]

\newtheorem {prp}{Property}[chapter]

\newenvironment{proof}{\noindent {\bf Proof :\ } }{\hfill$\Box$ }
\newenvironment{infprf}{\noindent {\bf Informal Proof :\ } }{$\Box$ }

\newenvironment{lemma}{\begin{lemmaa} \sl}{\end{lemmaa}}

\newenvironment{Informal Proof}{\begin{prp} \sl}{\end{prp}}

\newenvironment{proof of correctness}{\noindent {\bf Proof of Correctness :\ } }{\hfill$\Box$ }

\newenvironment{sketch of proof}{\noindent {\bf Sketch of proof :\ } }{\hfill$\Box$ }

\newtheorem {definition}{Definition}

\newtheorem {procd}{Procedure}

\newcommand{\eat}[1]{}

\definecolor{mygreen}{rgb}{0,0.6,0}
\definecolor{mygray}{rgb}{0.5,0.5,0.5}
\definecolor{mymauve}{rgb}{0.58,0,0.82}
\begin{document}
\sloppy

\pagestyle{empty}
\pagestyle{empty} 
\begin{center} 
\Large{\textbf{\textsf{{Asynchronism in Cellular Automata}}}}\\ 
	\end{center}
	\begin{center} 
	\vspace{4.0cm}
	 \textbf{\normalsize{VIRENDRA KUMAR GAUTAM}} \\
	\vspace{4.8cm}
        \begin{center}
 		 \includegraphics[height=4cm, width=4cm]{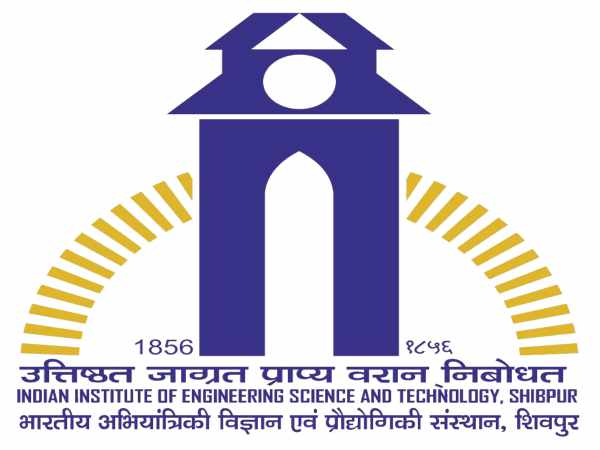}
        \end{center} 
        \vspace{1.2cm}     
      \scriptsize{\textbf{DEPARTMENT OF INFORMATION TECHNOLOGY}}\\ 
      \scriptsize\textbf{{INDIAN INSTITUTE OF ENGINEERING SCIENCE AND TECHNOLOGY, SHIBPUR}} \\
     \scriptsize{\textbf{HOWRAH, WEST BENGAL, INDIA-711103}}\\
               \vspace{0.5cm} 
   \large{\textbf{June, 2024}} \\  
  	\end{center} 

\newpage


\pagestyle{empty} 
\begin{center} 
\Huge{\textbf{\textsf{{Asynchronism in Cellular Automata }}}}\\
        \vspace{1.5cm} 
        \Large{\textbf{Virendra Kumar Gautam}}\\   
        	 \vspace{0.3cm} 
     \large{{Registration No. 2022ITM008}} \\
	     \vspace{1.2cm}
        \small{\em{A report submitted in partial fulfillment for the degree of}}\\     
        \vspace{0.12in} 
        \small{\textbf{Masters of Technology}}\\ 
        \small{\textbf{in}}\\ 
        \small{\textbf{Information Technology}}\\ 
        
     \vspace{1.0cm} 
	\large{\em{Under the supervision of}}\\ 
	\vspace{0.10cm} 
	\large{\textbf{Dr. Sukanta Das}}\\ 
	\small{Associate Professor}\\
	\small{Department of Information Technology}\\
	{Indian Institute of Engineering Science and Technology, Shibpur}\\       
        
        \vspace{0.45in} 
        \begin{center}
		\includegraphics[height=4cm, width=4cm]{logo.jpg}
        \end{center}
		\vspace{0.3cm}

	\normalsize{\textbf{Department of Information Technology}}\\ 
\normalsize{\textbf{Indian Institute of Engineering Science and Technology, Shibpur}}\\
	\normalsize{\textbf{Howrah, West Bengal, India -- 711103}}\\ 
    \normalsize{\textbf{June, 2024}}\\ 
	\end{center} 

\newpage

\pagestyle{empty} 
\begin{center} 
\begin{figure}[h]
\vspace{-1.5cm} 
\centering
	\includegraphics[height=2.5cm, width=2.5cm]{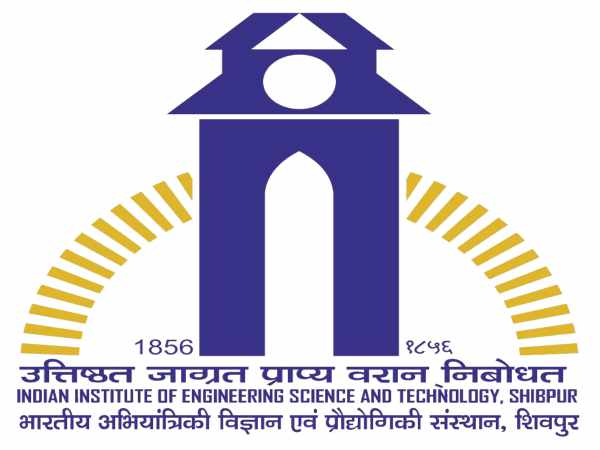}
\end{figure}
\small{\textbf{Department of Information Technology}}\\ 
	\small{\textbf{Indian Institute of Engineering Science and Technology, Shibpur}}\\
	\small{\textbf{Howrah, West Bengal, India -- 711103}}\\  
\vspace{0.40in} 

{\Large \bf CERTIFICATE OF APPROVAL}\\ 
\end{center} 
\vspace{0.22in} 
\normalsize{\par It is certified that the project report entitled \ 
\emph{\bf Asynchronism in Cellular Automata} is a record of bonafide work carried out by \emph{\bf Virendra Kumar Gautam} under my supervision and guidance in the Department of Information Technology of Indian Institute of Engineering Science and Technology, Shibpur.

In my opinion, the work is satisfactory and has reached the standard necessary for the submission in the final semester of
\textsl{Master of Technology} \ in  
	\ \textsl{Information Technology} \  of  
      the \ Indian Institute of Engineering Science and Technology, Shibpur, and to the best of my knowledge, the results in any form has not been submitted to any other University or Institute for any purpose.
\vspace{0.2in} 
\begin{flushright}
	\includegraphics[height=3\baselineskip]{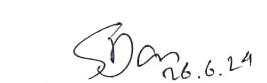}
	 \vspace{-0.05in}
	\par\textbf{(Dr. Sukanta Das)}  \\
		\small{Associate~ Professor} \\
	\small{Dept. of Information Technology} \\
		\small{Indian Institute of Engineering Science and Technology,} \\
		\small{Shibpur, Howrah, West Bengal, India --711103} \\
\end{flushright}
\vspace{0.1in}
\hspace{-0.2in}{Counter signed by:}
\begin{flushright}
	\includegraphics[height=4\baselineskip]{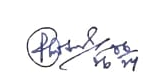} 
	\vspace{-0.05in}
	\par\textbf{(Dr. Prasun Ghosal)}  \\
	\small{Associate~ Professor \& Head} \\
	\small{Dept. of Information Technology} \\
	\small{Indian Institute of Engineering Science and Technology,} \\
	\small{Shibpur, Howrah, West Bengal, India --711103} \\
\end{flushright}

	


\vspace{1mm}
} 
\newpage 

\newpage

\vspace*{\fill}  

\begin{quote}

\centering

\begin{large}


\textbf{\Large{\textit{Dedicated}}}\\

\vspace{0.25cm}

\textbf{\Large{\textit{to}}}\\

\vspace{0.25cm}
\textbf{\Large{\textit{My Family, Supervisor and Mentors, and to all struggling People}}}\\
\vspace{0.25cm}

\vspace{0.25cm}


\vspace{0.25cm}


\end{large}

\end{quote}

\vspace*{\fill} 

\clearpage

\pagestyle{plain}
\pagenumbering{roman}
\afterpage{\blankpage}

\cleardoublepage
\phantomsection
\addcontentsline{toc}{chapter}{Acknowledgement}
\begin{center}
\vspace{-5.5cm}
 \textbf{\Large ACKNOWLEDGEMENT}
\end{center}

First and foremost, I would like to express my heartfelt gratitude to my advisor, Dr. Sukanta Das, Associate Professor, Department of Information Technology, Indian Institute of Engineering Science and Technology (IIEST), Shibpur, for his unwavering and continuous support and assistance throughout the preparation of this dissertation. I deeply appreciate his patience, motivation, enthusiasm, and extensive knowledge. His guidance and supervision has been invaluable in every stage of writing this thesis, and I could not have imagined having a better mentor and advisor. During my M.Tech journey, I have learned lot of things from him, especially how to be disciplined in research and in life.


I would also like to extend my profound respect and appreciation to Dr. Souvik Roy, Assistant Professor, Ahmedabad University, Gujarat, for his insightful suggestions and advice, which have greatly enhanced my analytical and methodological rigor. Additionally, I am immensely grateful to Dr. Sukanya Mukharjee, Assistant Professor, IEM Kolkata, whose intellectual engagement has been tremendously beneficial for me.


The work reported in this dissertation was the result of collaborative efforts. In this project on "Asynchronism in Cellular Automata," I collaborated with Dr. Souvik Roy, who developed the theorem on Skewed Fully Asynchronous Cellular Automata, and I worked with him to explore its dynamics and potential applications.


I am thankful for the financial support provided by the Indian Institute of Engineering Science and Technology, Shibpur, during my M.Tech tenure.


I am also grateful to the Head of the Department, Prof. Prasun Ghosal, Associate Professor, and all the esteemed professors of the Department of Information Technology at IIEST, Shibpur, for their kindness and support at various stages of my research. In addition to my advisor, I would like to express my deepest appreciation to each member of the M.Tech committee for their insights and technical suggestions. My thanks also go to the department's technical and non-technical staff, including Malay-sir, Suman-da, and Dinu-da, for their dedication and support.

I am sincerely thankful to friends for their continuous support and encouragement. In particular, I want to thank my friends Mrityunjay Chauhan, Ranga Prasad Vikruti, and Vinod Jha, whose support was invaluable, especially during challenging times. His encouragement and inspiration were constant through both good and difficult periods. Special thanks to my all of the labmates and batchmates for their friendship, love, and emotional support. I also appreciate Pritam, Subrata Paul, who are research scholars of the department of Information Technology, for their support over the past one and half years.

Most importantly, I express my heartfelt respect and gratitude to my parents, Late Mr. Ramesh Chandra Gautam and Mrs. Madhuri Gautam, for their unwavering support, sacrifice, and inspiration throughout my academic journey. Additionally, I acknowledge the invaluable guidance of my primary and secondary education teachers and mentors, who introduced me to the world of education and laid the foundation for my academic pursuits.


\vspace{35mm}

\begin{flushleft}
\includegraphics[width = 30\baselineskip, height=6\baselineskip]{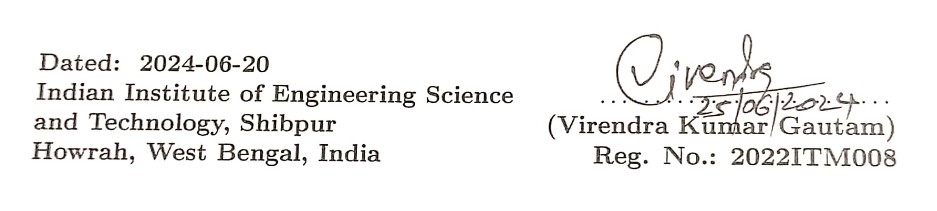}
\end{flushleft}



\newpage
\cleardoublepage
\phantomsection
\addcontentsline{toc}{chapter}{Abstract}
\afterpage{\blankpage}
\chapter*{Abstract}
\label{abstr}
\textbf{T}his study presents the concept of Skewed Fully Asynchronous Update in cellular automata and examines the resulting behaviour. It also investigates the dynamics of elementary cellular automata under this update scheme, comparing it with other update methods such as synchronous and fully asynchronous updates. Furthermore, the research introduces the application of fully asynchronous cellular automata for solving clustering problems and explores the dynamics of elementary cellular automata within the framework of $\alpha$-asynchronous cellular automata.

In this work, we introduce a new type of asynchronous cellular automata called Skewed Fully Asynchronous Cellular Automata  (SACA) as a source of noise in cellular auotmata. The $Skewed$ $Fully$ asynchronous cellular automata allows to  update the states of only two consecutive and adjacent cells, say $c_i$ and $c_{i+1}$, simultaneously at each and every time step. Under this proposed scheme, we analyzed dynamics and behaviour of elementary cellular automata (ECA). The dynamical behaviour are compared with the $fully$ asynchronous cellular automata, and synchronous cellular automata. This comparative study points out varieties of rich phenomenon that some of the elementary cellular automata shift from convergent nature ( or non reversible non convergent dynamics) to reversible nature (or divergence). We also identify the cases where the divisibility of the lattice size by 2 or 4 introduces massive repercussion in the system following presence or absence of the atomicity property. Lastly, we theorize the reason behind convergence towards all 0 and all 1 point attractors under the proposed skewed environment which partially validates our experimental observations.

\textbf{T}he work also presents the study and outcomes of the application of fully asynchronous cellular automata for solving clustering problem. By clustering, we mean a group of objects having similar characteristics (or properties). This work uses reversible asynchronous cellular automata. Here, a scheme is used to merge a group of clusters at every level based on the closeness of the clusters and it repeats until we get the desired number of clusters. 
The use of fully asynchronous reversible cellular automata for clustering enables us to work with only a very few rules and results in faster convergence. 

\cleardoublepage
\tableofcontents

\cleardoublepage
\phantomsection
\addcontentsline{toc}{chapter}{\listfigurename}
\listoffigures

\cleardoublepage
\phantomsection
\addcontentsline{toc}{chapter}{\listtablename}
\listoftables


%
\clearpage

\pagestyle{fancy}
\pagenumbering{arabic} 
\chapter{Introduction}
\label{chap1}
\hrule
\textit{If people do not believe that mathematics is simple, it is only because they do not realize how complicated life is.}
\begin{flushright}
\textbf{John von Neumann}
\end{flushright}
\vspace{5pt}
\hrule
\vspace{6pt}
\section{Background}
The advancement of science has always been inspired by the marvels of nature. Nature operates in a remarkably unpredictable manner, with each organism playing a unique role in influencing overall behaviour. For example, Bee Colony: while the queen bee plays a central role, the activities and roles of worker bees and drones are largely decentralized and self-regulated; Human Brain: the brain acts as a central controller, it relies on a network of neurons that function in a decentralized manner to process information and control bodily functions. Decentralization can be seen in each and every system, be it natural or man-made. 
Therefore, numerous existing centralized and decentralized systems inspired human to invent and develop scientific, engineered and artificial systems. These inventions and development, happening from a long back, are not only for making life easy but also to understand the structure, and complexity of the nature.

Following the above words, we see that Turing Machine, and Distributed Computing Systems are two of the pioneer work and inventions in the history of the decentralized systems, till the present date. The Turing Machine~\cite{turing1937computable, Turing}, introduced by Alan Turing, whose computation relied on a centralized control tape head, laid the groundwork for the mathematical model that has been used since the inception of the modern computer era. Another computationally powerful model were introduced by Stanislaw Ulam and Jonh von Neumann in 1940s. This laid down the invention and the development of cellular automata (CAs)~\cite{Neuma66} which stands as a remarkable milestone in computational theory, drawing upon a rich tapestry of mathematical models and frameworks, each contributing uniquely to its conception. Among these foundational elements, the von Neumann architecture, Markov chains, lambda calculus, and Petri nets play pivotal roles, offering diverse perspectives and mechanisms that underpin the invention and evolution of cellular automata~\cite{moore1962machine, Toffo87}.

Since inception of CAs, a large population of think tanks have been attracted by this computational model for study and modeling real world systems related to diverse field of studies. It is because of CAs's ability of modeling physical systems, computational systems such as Turing Machine, and biological systems such as self - reproduction~\cite{Sethi2016}. A parallel computing systems have been modeled by utilizing these CAs~\cite{Sethi901, Anindita12, Sarkar12}. The von Neumann's CAs has so high complexity that these CAs have not been considered for implementation in the computers. With the passage of time, many evolutionary modifications and simplifications were brought into the CAs just to bring the complexity as low as it could be possible. In this scenario and demand, a great development occurred after 20 years from the inception of the CAs. This simplification of structure of CAs was done by S. Wolfram~\cite{Sethi2016, wolfram2002new}. He proposed a new type of cellular automata (CAs), which they called elementary cellular automata (ECAs). This proposal for simplification leads to a one dimensional 
regular grid of point (cells) in euclidean space, each cell having two states, and three neighbors~\cite{wolfram2002new}. This structure of CAs attracted researchers to undertake it for modeling and study of some natural phenomenon~\cite{Sethi13}.
The cells in the CAs interact into a distributed manner. All cells are updated int temporal-spatial by interacting with each of their neighboring cells synchronously. The aggregate of the state of all cells at a time step represents the CAs's global state. This philosophy of synchronization assumes the existence of a global which ensures the simultaneous update of all cells~\cite{ref_r1, Sethi2016, Sethi901, Sarkar12}.

However, nature is full of randomness, uniformity and non-uniformity. It inspired many researchers to invent such a computational system that could be able to simulate and model the uniform, non-uniform and random processes occurring in our surroundings. It leads to the introduction of the asynchronism into the classical CAs, which restrict simultaneous update of states of each and every cell at the same time step ~\cite{Nakamura}. Asynchronism means independency into the functionalities of a system or absence of global clock. This philosophy of asynchronism in CAs leads to development and invention of asynchronous cellular automata (ACAs). In ACAs, one or more than one cells can be updated at a same time step corresponding to a set of update pattern associated to a given CA rule $R$~\cite{ref_r1,roydistributed, Das95}. With emergence of ACAs, many more ideas were raised related to the update scheme. However, two of the other asynchronous updating schemes have been considered as the main ACAs schemes~\cite{fates13, Fates17, Sethi901}:
\begin{itemize}
	\item[•] \textit{Fully} Asynchronous Cellular Automata (ACAs): Only one cell is updated corresponding to an update pattern associated to a CA $R$ at a time step chosen at random and uniformly.
	\item[•] $\alpha$ - Asynchronous Cellular Automata: One or more cells are updated corresponding to an update pattern associated to a CA $R$ at a time step chosen at random and uniformly. Only those cells are updated those have probability as least $\alpha$
\end{itemize}

\section{Motivation}

The idea of cellular auotmata (CAs) was, firstly, given by Jon von Neumann. He was the first man who introduced it to the world in early 1950's. He invented the CAs just to study the computation power and impact of implementability of self-reproducing systems and machines. As per~\cite{turing1937computable, Turing}, a computing systems is said to be computationally universal if it is capable of simulating other computing models. The same idea stroked into the Jon von Neumann to think about a model which could be a nature inspired, capable to simulate real world phenomenon. But, he was well focused on the constructive development to which they tried to achieve with the utilization of the cellular automata (CAs). He tried to show that how the two concepts (i) computational universality and (ii) computational constructive universality are related to each other in his demonstration. He demonstrated that it is computationally possible to implement a computationally universal machine such as Turing Machine with the help of the CA.

In recent years, numerous scientists, engineers and research scholars from various field of studies have focused on the cellular automata. Christopher Langton utilized CA to study the ``Artificial Life'' model. It is due to the natural characteristic of CA: \emph{self-replicating, self-organization, self-healing} e.g. Conway's ``Game of Life''~\cite{Gardner70}. We have discussed about Wolframe's classification,and Neumann's universal constructor in Chapter~\ref{chap2}.

The first motivation to work on this thesis is to study the behaviour of reversible ECAs under the fully ACAs when utilized for solving clustering problem. The dynamical behaviour of these ECAs have already been studied by many researchers around the globe~\cite{Sethi2016, Sethi901, roy:hal-04456320}.

Another motivation to explore the effect of allowing update of consecutive cells corresponding to a given ECA $R$. Many researchers have focused on study of single cell update, while update of consecutive cell have never been undertaken to study perturbation into the system. It motivated to explore the dynamics of the ECAs in comparison with the dynamics of the ECAs under synchronous and fully asynchronous cellular auotmata. 

Another motivation leads to explore the dynamical behaviour of the ECAs when one or more cells can be updated randomly and uniformly at a time step.


\section{Objective of the thesis}
This thesis mainly aims on qualitative and quantitative analysis of the computational capabilities of decentralized computational models, specifically focusing on decentralized computation with CAs. In these cellular automata (CAs), each cell comprises of a finite automatan. Each cell communicates with its neighboring cells to determine its next state~\cite{Sethi2016,Neuma66,BhattacharjeeNR16}. These cells are arranged in a regular euclidean grid, and the intriguing aspect of these CAs is their ability to generate complicate and complex global behaviour through simple local interactions of the cells~\cite{PhysRevLett.88.237901,Cor,Smith71},BhattacharjeeNR16. Numerous research work have been published centered around CA, some of them are considered pioneering work such as~\cite{Zuse1982,wolfram2002new, Smith71}. 
In this dissertation, we study, investigate, and theorize that how these automata can be utilized to resolve the following issues and real world problem:
\begin{itemize}
	\item Application of \textbf{fully} ACAs to solve clustering problem.
	\item Dynamical behaviour of ECA $R$ under \textbf{skew} asynchronous update scheme.
	\item Dynamical behaviour of ECA $R$ under $\alpha$ asynchronous update scheme.
\end{itemize}

The surrounding is not free from ambiguities. All the systems (living and non-living) have some type of ambiguities, and can never be free from such perturbation at any stage of their temporal evolution. These evolutions may cause behavioural changes into these systems. In this thesis, we have proposed skew asynchronous CAs. While the researchers have proposed and published quality of work related to the fully asynchronous and $\alpha$ asynchronous in ~\cite{ref_r4,Sethi2016,fates13, Bersini94}. In recent years, researchers also have focused on utilization of CAs to solve data mining, and machine learning related problems such as classification and clustering problems ~\cite{ref_r1,ref_r7,Sethi2016}.

\section{Contribution of the thesis}
To achieve our objectives, we have carried out this study under the aforementioned environment and surroundings. Here is a summary of the key findings from our study and research efforts:
\begin{itemize}
	\item We have explored ECAs those are reversible under fully asynchronous CAs. Since, we utilized this update scheme of CAs to resolve clustering problem; we explored communication classes associated to these reversible ECAs. 	
	\item These communication classes are group of finite configurations, where a configuration can never be a member element of any two classes.
	\item We found that only 12 reversible ECAs under fully ACAs are useful for the clustering purpose. These ECAs have more than three communication classes. 
	\item The choice of a list of these reversible ECAs depends on many factors to be utilized for clustering of a dataset. These factors are size of dataset, type of attributes, and encoding technique to encode the given dataset.
	\item We have found that the cellular automata based clustering algorithm are as googd as the other benchmark clustering algorithms such as K-means, PAM, and Hierarchical clustering. 
	\item  We proposed a new type of asynchronous CAs which we call skew asynchronous cellular automata, where two consecutive and adjacent cells get updated simultaneously corresponding to a given ECA $R$.
	\item The skew asynchronous cellular automata is an asynchronous update scheme which breaks the atomicity propoerty.
	\item We have also presented the comparative research findings of dynamical behaviour of ECAs in Chapter~\ref{chap4}.
	\item Lastly, we have tried to explore the dynamical behaviour of minimal representative ECAs under the $\alpha$ - asynchronous update scheme. We have reported the findings in Chapter~\ref{chap5}. We have also theorized the logic working behind the reversible and convergent ECAs (convergent to either all - \textbf{0} or all - \textbf{1}).
\end{itemize}

\section{Organization of the thesis}
\begin{itemize}
	\item \textbf{Chapter~\ref{chap2}.} This chapter is devoted to survey on the asynchronous cellular automata. The chapter describes the classical form of the cellular automata and its meaning, types of non - uniformity in cellular automata, and important types of non -classical CA such as automata networks, non - uniform CAs, and hybrid CA. The chapter mainly surveys the types of asynchronism on cellular automata (CAs).
	\item \textbf{Chapter~\ref{chap3}.} This chapter describes the application of fully asynchronous cellular automata for solving clustering problem. It also address different types of cluster validation indices such as dunn index, calinski harabasz index, and silhouette score. Finally,  the chapter address the results of proposed clustering algorithm in comparison with different types of benchmark clustering algorithms such as K-Means, Hierarchical clustering.
	\item \textbf{Chapter~\ref{chap4}.} This chapter describes the skew asynchronous cellular automata in detail. The chapter focuses on the convergent dynamical behaviour of 88 mimimal representative CAs rules under skew asynchronous update scheme.
Finally, it address the comparative result of skew asynchronous cellular automata, fully asynchronous cellular automata, and classical cellular automata.
	\item \textbf{Chapter~\ref{chap5}.} This chapter address the dynamical behaviour of 88 minimal representative CAs rules under $\alpha$ - asynchronous update scheme. The addresses both reversible and convergent dynamics of these rules under the $\alpha$ - asynchronous update scheme.
	\item \textbf{Chapter~\ref{chap7}.} This chapter address the conclusion of the work presented in this thesis. It sum up with the overall outcome of the work and future directions of the proposed work.
\end{itemize}

\chapter{Survey on Cellular Automata}
\label{chap2}
Cellular Automaton (CAs) is a discrete, abstract computational system that consists of a regular network of finite state automata, formally known as cells. A cell changes its state depending on the state of its neighbors using a local update rule and all the cells change their states simultaneously using the same update rule. Because of their simple structures, CAs have been considered a hot field of study in terms of research and application in diverse fields such as biology, physics, chemistry, mathematics,VLSI design, and computer science~\cite{langton90,DCTMitchell93,Supreeti_2018_chaos,KAMILYA2019116,Sethi2016}.

In contrast, ACAs allow cells to update independently, potentially leading to more realistic and varied dynamic behaviour. Asynchronous Cellular Automata (ACAs) represent a fascinating and versatile area within computational theory and mathematical modeling.  This chapter provides a comprehensive literature survey on ACAs, including fully asynchronous cellular automata (FACAs) and alpha asynchronous cellular automata ($\alpha$ - ACAs), and explores their definitions, properties, and applications.
\section{Cellular Automata}
\label{sec:ca_}
Informally, a cellular automatan is made up of cells arranged in a regular grid (euclidean plane), where cells are finite automaton with finite state - set $S$. States of Each cell of the CA get updated with the passage of time and associated locations, which results into the transition of CA from one global state to another. This alteration of state of a cell takes place based on the state(s) of its neighboring cells, which takes place through a next state function called local transition function. The current state's of the neighboring cells together with the particular cell to be updated are utilized as input to the local transition function. The output of this function throughout the time serves as the next state of the cell. And, the aggregation of these states of cells with respect to the time is termed as the configuration of the CA.

A formal definition of the CA is given as below:
\begin{definition}~\cite{Sethi2016,fates13}
	A cellular automaton is a quadruple ($\mathcal{L}$, $\mathcal{S}$, $\mathcal{M}$, $\mathcal{R}$) where,
	\begin{itemize}
		\item $\mathcal{L} \subseteq \mathbb{Z}^\mathcal{D}$ is the $\mathcal{D}-$dimensional cellular space. A set of cells are placed in the locations of $\mathcal{L}$ to form a regular network.
		\item  $\mathcal{S}$ is the finite set of states; e.g. $\mathcal{S} = \{0, 1,\cdots, d-1\}$.
		\item $\mathcal{M}=(\vec{v_1},\vec{v_2},\cdots,\vec{v_m})$ is the neighborhood vector of $m$ distinct elements of $\mathcal{L}$ which associates one cell to it's neighbors.
		\item The local rule of the CA is $\mathcal{R}:\mathcal{S}^m\rightarrow\mathcal{S}$. A cell's subsequent state is determined by the expression $f(s_1,s_2,s_3,\cdots,s_m)$ where $s_1,s_2,s_3,\cdots,s_m$ denotes the states of it's $m$ neighbors.
		\end{itemize}   
\end{definition}
\label{def1_ca}
The cellular automata (CAs) can be categorized on many basis such as kind of grid, uniformity in update of cells, uniformity in update rules. The classical CA has three fundamental characteristics those are \textbf{l}ocality, \textbf{s}ynchronicity, and \textbf{u}niformity, where:
\begin{itemize}
	\item \textbf{Locality} stands for the computation of CAs using an ECA $R$, termed as a local computation. The response of this local computation (interaction of neighboring cells) brings variation into current state of a given cell.
	\item \textbf{Synchronicity} refers to the concurrent and simultaneous update of all cells under the consideration of local computation.
	\item \textbf{Uniformity} refers to the application of the same local rule by all the cells of CA.	 
\end{itemize}

Out of these three characteristics, the locality has been seen more important among the three. It is because of its employment by all cells to carry out the actual local computation. Apart from these many characteristics, another dimension, called radius of CA, can also be studied to understand cellular automata and impact of local transition rule. The radius of a CA refers to the quantity of succeeding cells in a direction that a cell depends on. For instance, a cell in elementary cellular automata, depends on its immediate left and immediate right cell. Therefore, the radius of this type of CA is $r = 1$. Hence, this CA is a ($1$ + $1$ + $1$) = $3$ - neighborhood CA. 

In this thesis, we have utilized one dimensional grid. This type of cellular automata was extensively explored by Stephen Wolfram
~\cite{ref_r8,wolfram2002new}. He termed it as an elementary cellular automata. 

\begin{figure}[hbt!]\centering 
	\includegraphics[width=0.50\textwidth]{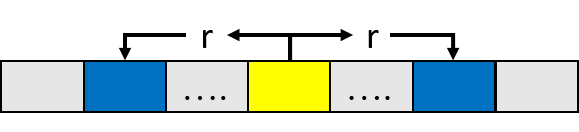} 
	\caption{neighborhood dependence in a one-dimensional cellular automaton with radious $r$.}
	\label{r}
\end{figure}

In this Fig.~\ref{r}, a cell in yellow uses $r$ left and $r$ right cells, in its neighbor, in order to evolve to the next state. However, John von Neumann and Moore used 2D CA in their neighborhood dependency. Each cell in both types of CA are squares. John von Neumann's CA cell has one of $29$ possible states~\cite{langton1984self,Smith71,ollinger2002quest,ollinger2003intrinsic,cook2004universality}. In following section, we have discussed elementary cellular automata in detail for illustration.

\begin{figure}[hbt!]
	\subfloat[]{
		\begin{minipage}[c][1\width]{
				0.30\textwidth}
			\label{von}
			\centering
			\includegraphics[width=1\textwidth]{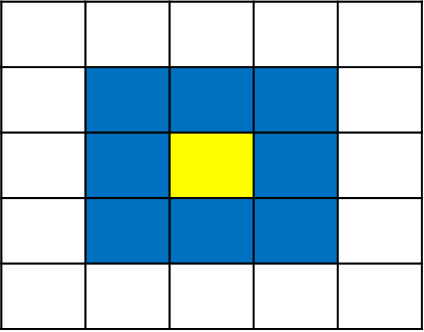}
	\end{minipage}}
	\hfill 	
	\subfloat[]{
		\begin{minipage}[c][1\width]{
				0.30\textwidth}
			\label{moor}
			\centering
			\includegraphics[width=1\textwidth]{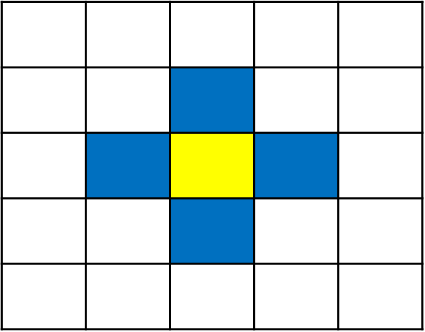}
	\end{minipage}}
	\caption{(a) Moore neighborhood dependencies for two-dimensional cellular automata;  (b) John Von Numann dependencies for two-dimensional cellular automata.}
\label{fig:neighbors}
\end{figure}

\subsection{Universal constructor: John von Neumann}
\label{subsec:vonNeumann}
It was the time of second world war, many top scientists were addressed and invited to join a project called Manhattan project at Los Alamos~\cite{banda2,ref_r19}. Amongs other Jon von Neumann and Stanisław Marcin Ulam were also the part of the society of the scientists joined the secret project. This project was about ``to build bomb'' that could be ``utilized against the Nazis''. It was not a big challenge for the scientists. There, it was the time some of the scientists were curious about the research on the computing complexity of a system. Much of theory and experimental works were done. During the same time Jon von Neumann also get motivated and the research potential of the computing complexity attracted him to take it as a project. So, the great scientists ``Stanislaw Marcin Ulam and Jon Von Neumann'' became interested in the abilities of Cellular Automata (CA) and self-reproduction of a machine.
\begin{figure}[htbp]

\begin{center}    
    \includegraphics[width=0.8\textwidth]{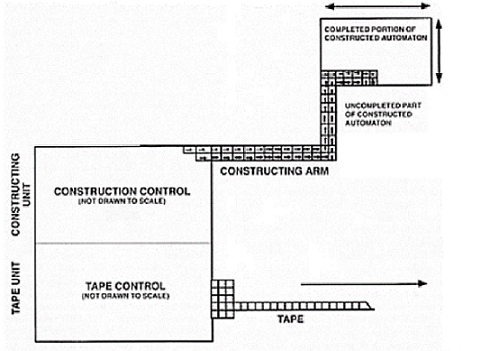}
    \captionof{figure}{John Von Neumann's Self-reproducting automata~\cite{banda2,Neuma66}}
    \label{fig:somelabe}
\end{center} 
\end{figure}
Stanislaw M. Ulam was more interested to develop an intelligent system which could behave like self-reproducing system and able to simulate the artifial growth and evolution. He got success and design the first system of this kind in the Log Alamos laboratory. This and other work by the Ulam inspired Jon von Neumann. Jon von Neumann (1950s) started working with this computing model to design the first computing model he termed cellular 
 automata~\cite{Neuma66,CARONLORMIER2008522}. His work mainly concentrated around the one and two dimensional geometry of the regular network of states of cellular automata (CA). Further work of Neumann was concentrated on formulation of sufficient conditions to design a non - trivial, self - reproducing system based on cellular automata(CA), and he imagined to formulate first kinematic model consisting of a robot (automata) floating in a lake, it will reproduce itself. But his this research work was hampered by the issue of motion in the lake. Thus, he showed the straight and capacity of the cellular automata. Later, Arthur Burks (a computer designer) attempted to create a two - dimensional realization of von Neumann's three - dimensional system of computation~\cite{Neuma66}.
\section{Wolfram's Cellular Automata and Classification}
We have presented a formal definition of cellular automata in~\ref{sec:ca_}. Many researchers have presented and proposed various type of cellular automata with some modifications in the basic philosophy of the cellular automata (CA). Later, after around 20 years since the inception of the CA, Wolfram were the first who systematically studied the cellular automata\cite{Wolfram94}. He even proposed a well structured type of cellular automata that could be as powerful as a natural one, which they called an elementary cellular automata (ECA)~\cite{Wolfram94}. The Section~\ref{subsec:eca} is dedicated to the concept and philosophy of the elementary cellular automata. 

\subsection{Elementary Cellular Automata}
\label{subsec:eca}
The cellular automata (CA)~\ref{def1_ca} with one-dimensional lattice, three neighborhood interaction, and two state finite set were proposed by~\cite{Wolfram94}. It is one of the simplest discrete, dynamical and complex system. Due to the simple structure of the ECA and capability to model natural phenomena (and computational systems such as Turing machine), the ECA has attracted think-tanks from diverse disciplines to study and utilize it more closely. In an elementary cellular automata (ECA), each cell of one-dimensional lattice in euclidean space evolves to next state based on the local interaction with its immediate left and immediate right neighbor cell. The lattice can either be infinite or finite under the point of consideration and objective of study. Therefore, the ECA can formally be defined as~\cite{Sethi2016,fates13,CPLX_CPLX21495}
\begin{definition}~\cite{Sethi2016,wolfram84b}
	An elementary cellular automaton is a quadruple ($\mathcal{L}$, $\mathcal{S}$, $\mathcal{M}$, $\mathcal{R}$) where,
	\begin{itemize}
		\item $\mathcal{L} \subseteq \mathbb{Z}^\textbf{1}$ is the $\textbf{1}-$dimensional regular grid in euclidean space. These cells are placed in the locations of $\mathcal{L}$ to form a regular grid of cells.
		\item  $\mathcal{S}$ is the finite set of two states $\mathcal{S} = \{0, 1\}$.
		\item $\mathcal{M}=(\vec{v_{i-1}},\vec{v_i},\vec{v_{i+1}})$ is the neighborhood vector of $3$ distinct elements of $\mathcal{L}$, that associates one cell to it's neighbors.
		\item $\mathcal{R}:\mathcal{S}^3\rightarrow\mathcal{S}$, where $\mathcal{R}$ is the local transition rule of the ECA. A cell's subsequent state is determined by the expression $R(s_1,s_2,s_3)$ where $s_1,s_2,s_3$ denotes the states of it's $3$ neighbors.
		\end{itemize}   
\end{definition}
\label{def_eca}
Table~\ref{Table:ECA} depicts next state of the cell under the basic principles of ECA with $S = \{0, 1\}$, and $r = 1$. 
\begin{table}[h!bb]
\centering
\begin{tabular}
{ p{4.5 mm} p{4.5 mm} p{4.5 mm} p{4.5 mm} p{4.5 mm} p{4.5 mm} p{4.5 mm} p{4.5 mm} p{4.5 mm} p{4.5 mm} } 
 \hline
PS & 111 &110& 101 & 100 &011 &010 &001 &000 & Rule\\\hline
NS & 0 & 0 & 0 & 1 & 1 & 1 & 1 & 0 & 30\\ 
NS & 0 & 1 & 1 & 0 & 1 & 1 & 0 & 0 & 108\\
NS & 1 & 1 & 0 & 0 & 1 & 0 & 1 & 1 & 201\\\hline
\end{tabular}
\caption{Present State (PS), and Next State (NS). This table shows an evolution of few rules from their present state to next state.}	
\label{Table:ECA}
\end{table}

Under the consideration of ECA, the total possible functioning CAs are $2^{2^3}$, which results into a total of $2^8 = 256$. In the Table~\ref{Table:ECA}, we mentioned a term \emph{RMT} of rules $30, 108, 201$. The RMT is a decimal number which is defined in the Definition~\ref{def_rmt}.

\begin{definition}
(Rule Min Term(\textbf{RMT}))~\cite{Sethi2016,Bhattacharjee16b}: Let $f : S^{2r+1} \rightarrow S$ be the local transition rule of an ECA. An input $(x_{-r},\cdots, x_0,\cdots, x_r) \in s^{2r+1}$ to $f$ is called a Rule Min Term (RMT). The output of the $f(r)$ is a non-negative integer number: $r = x_{-r}.d^{2r} + x_{-r+1}.d^{2r-1} + \cdots + x_r$ and $f(x_{-r}, \cdots, x_0, \cdots, x_r)$. RMT $r$ is called passive if $f(x_{-r}, \cdots, x_0, \cdots, x_r)=x_0$. Otherwise the RMT is called active RMT. For a $d$-state CA, total number of RMTs is $d^{2r+1}$.
	\label{def_rmt}
\end{definition} 
In the ECA there are $8$ possible $RMTs$ represented by $0 (000),~1 (001), ~2 (010), ~3 (011), ~4 (100), ~5 (101), ~6 (110), ~and 7 (111)$. The $RMTs$ $0 (000), ~2 (010), ~3 (011)$, and $5 (100)$
are passive $RMTs$, while $1 (001), ~4 (100), ~6 (110)$, and $7 (111)$ are active $RMTs$ which is displayed in Table~\ref{Table:ECA}) corresponding to $ECA$ $rule$ - 30.
\subsection{Classification of Cellular Automata}
\label{subsec:wolf}
We have discussed about the $ECAs$ in Subsection~\ref{subsec:aca} which is the most basic type of CAs. Utilizing this ECA, Wolfram studied the dynamical behaviour of all $256$ CAs and categorized into four basic classes periodic, homogeneous, chaotic and complex localized~\cite{Wolfr83,Packa85b}. Classification of the CA were further have been done by many other researchers too~\cite{FSSPCA3,mitch98a,Sarkar00Survey}. Wolfram's classification~\cite{wolfram84b} is given as below, which was done by considering ECAs with finite size lattice, and classical definition of CAs for ECAs:
\begin{enumerate}
	\item[] \textbf{Class I.} Evolution leading to a homogeneous state. Rules such as 0, 8, 32, 40, 168,
	\item[] \textbf{Class II.} Evolution leading to a periodic state. Rules such as 4, 12, 20, 164, 172,
	\item[] \textbf{Class III.} Evolution leading to a chaotic state. Rules such as 18, 106, 158; and
	\item[] \textbf{Class IV.} Evolution leading to a development of complex structure. Rules such as 54 and 110.
\end{enumerate}	
\begin{figure}[hbt!]
\centering 
	\includegraphics[width=0.60\textwidth]{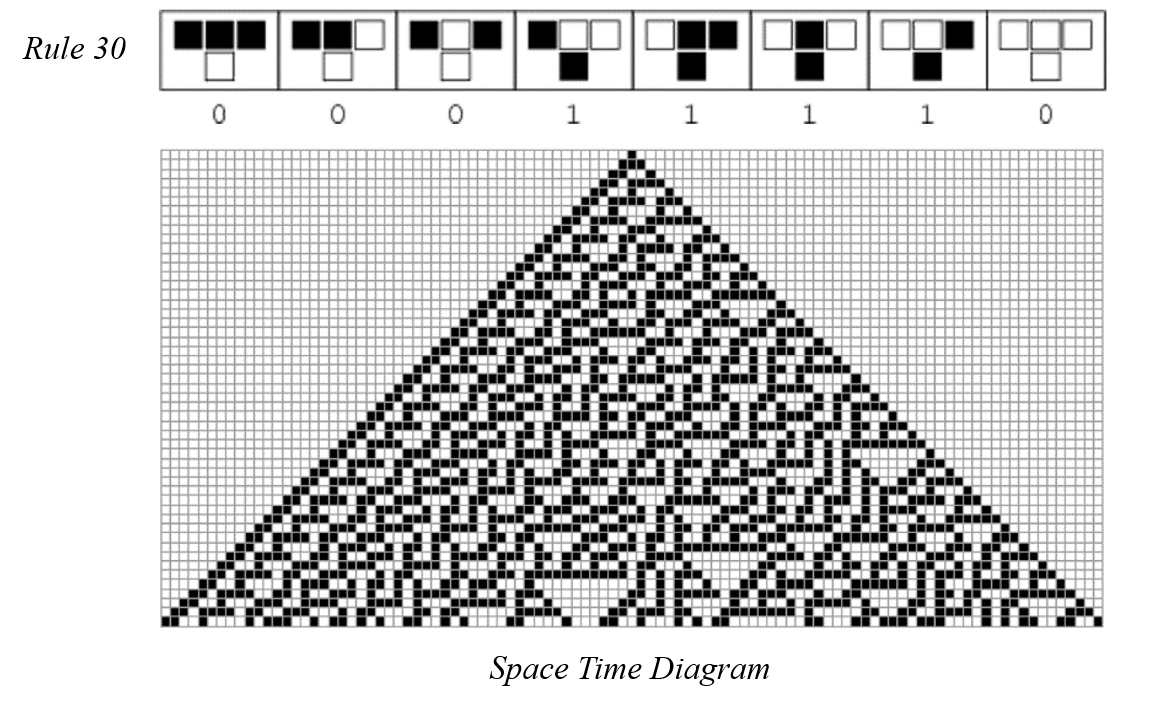} 
	\caption{Space-time diagram of a $1-$D Cellular Automaton of (Wolfram rule $30$). White cells denote state $0$, whereas black cells denote state $1$.}
	\label{rule30}
\end{figure}

He also extended this classification to two dimensional cellular automata in~\cite{Packa85b}. Fig~\ref{rule30} depicts the space-time diagram of an elementary cellular automata (ECA) rule - $30$.

\section{Types of Cellular Automata}
\label{sec:typesca}

In section-~\ref{sec:ca_}, we presented a classical cellular automata (CAs), elementary cellular automata (ECAs), von Neumann's constructor etc., where the essential three characteristics of CAs were never disturbed.

However, in this section we discuss about the asynchronous cellular automata (ACAs), types of ACAs, and applications in detail. 

A classical cellular automata posses three basic properties: (i) locality, (ii) synchronicity and (iii) uniformity. Disobeying any of the these three results into a variant of cellular automata, which definately can not be a synchronous cellular automata. Uniformity means that each CAs cell is updated by the same local rule. Synchronicity implies that all cells are updated simultaneously. Locality indicates that the rules operate locally, with each cell's neighborhood dependencies being consistent. Cells perform computations locally, and the global behaviour of CAs emerges from this local computation. Synchronicity is a specific type of uniformity, where all cells are updated simultaneously and uniformly. In fact, uniformity pervades CAs, evident in the local rule, cell updates, and lattice structure~\cite{BhattacharjeeNR16}.

Since the inception of CAs, it have been widely utilized by researchers as a modeling tool. However, it has become evident that many natural phenomena, such as chemical reactions within a living cell, do not exhibit uniformity. This realization led to the development of a new variant of CAs incorporating non-uniformity. Consequently, the introduction of non-uniformity in CAs has resulted in three main variants, arising from the relaxation of the constraints associated with uniformity. Following are the types of non-uniformity:
\begin{itemize}
	\item[1]. Automata Network: In this model, the CA is structured on a network, and the states of the nodes evolve according to neighborhood relationships defined by the network, thereby breaking the uniform neighborhood constraint.
	\item[2]. Asynchronous Automata: In ACAs, cells are not updated at the same discrete time step and can be updated independently, thus violating the uniform update constraint.
	\item[3]. Hybrid Automata: These CAs allow cells to have different local transition functions, thus violating the uniform local rule constraint.
\end{itemize}


\subsection{Automata Network}
\label{subsec:autnet}
Traditionally, cellular automata consist of a regular network with uniform local neighborhood dependencies. However, in an automata network, this uniform neighborhood dependency is relaxed. Here, cellular automata rules allow a cell to have an arbitrary number of neighbors, enabling them to operate on any given network topology, as demonstrated by Marr and Hütt ~\cite{Marr20}. The rules of automata networks are not always local, and since non-local and local rules differ, it is expected that non-local rules may lead to different behaviours compared to conventional local rule-based CAs. Examples of work in this area include Boccara~\cite{Boccara94}, Newman~\cite{Newman99}, and Yang~\cite{Yang07}.

The earliest instances of non-uniform neighborhood interactions in CAs are found in~\cite{Smith76,Jump74}. From the 1990s onward, networks have become a crucial model for addressing various complex problems, as exemplified by Adami, and Watts and Strogatz~\cite{Adami95, Watts98}. These wokrs significantly increased researchers' interest in automata networks. Extending standard lattice cellular automata and random Boolean networks results in a broader class of generalized automata networks~\cite{Tomasssini06}.
Researchers have explored automata networks as algebraic structures, developing their theory in parallel with other algebraic frameworks such as semi-groups, groups, rings, and fields~\cite{Domosi}. They have also introduced a method to emulate the behaviour of any synchronous automata network using a corresponding asynchronous one. ~\cite{Kayama} Presented a network representation of the Game of Life, exhibiting characteristics similar to Wolfram’s class IV rules. However, studies in this area are still in their early stages.
\subsection{Asynchronous Automata}
\label{subsec:asynch}
Concept of asynchronism in CAs were added by many researchers. It was Nakamura, who is considered the developer of the asynchronous cellular automata (ACAs). Later, follwed by Golze~\cite{Golze78}, Ingersan and Buval~\cite{Ingerson84}, and Cori et al~\cite{robco}. A proper survey on the CAs and ACAs can be found in~\cite{BhattacharjeeNR16,Fate2014}. 

Addition of the asynchronism into the classical cellular automata brought an abrupt and drastic changes into the dynamical behaviour of the cellulara automata. He also have properly investigated its computational complexity ~\cite{Nakamura}. He proposed technique by which ACAs could be used to implement universal Turing machine. He also had simulated working process of the synchronous cellular automata (ACAs) using the asynchronous cellular automata(CAs)~\cite{Nakamura}. ACAs have diverse applications and with the passage of time new chapters are getting added into it. The asynchronous cellular automata are a type of cellular automata which have no global clock. Instead of, each cell is updated independently corresponding to a given CAs rule $R$. The asynchronous cellular automata are basically categorized into two types:
\begin{itemize}
\item[1]. Fully Asynchronous Cellular Automata~[\ref{subsec:fullyACA}].
\item[2]. $\alpha$ Asynchronous Cellular Automata~[\ref{subsec:alphaACA}].
\end{itemize}

The $\alpha$ is known as synchrony rate where $\alpha \in [0, 1]$. For $\alpha = 1$, the CAs behaves like classical CAs. This tends to a sense of special case of $\alpha$ - ACAs. Latter, many asynchronous ACAs scheme have been proposed by researchers over globe such as $\beta - $ asynchronism and $\gamma - $ asynchronism~\cite{BoureFC12}, $m - $asynchronism~\cite{Dennunzio13a}.

One of the excellent work on the computation equivalence of CAs and ACAs  and have been shown by Hemmerling~\cite{Hemmerling82}. Further, he demonstrated that any $d-$state deterministic rule can be modeled by a 3$d^2$ state asynchronous rule with same neighborhood interactions. With different perspective, Golze~\cite{Golze78} demonstrated that a (D+1)-dimensional asynchronous rule is capable of simulating a D-dimensional synchronous rule. Latter, Dennunzio~\cite{Dennunzio12a} illustrated the ability of fully ACAs to simulate a universal Turing machine. Petri net is a directed bipartite graph that cab be implemented by ACAs~\cite{Golze82}. ACAs can also be utilized to model concurrency and distributed computing systems which has been illustrated in~\cite{robco} further investigated by Pighizzini and Droste~\cite{Pighizzini94,Droste2000}.

Letter, number conservation~\cite{Suzudo04}, reversibility~\cite{Sarkar12, Sethi2016} and density classification problems~\cite{fates13} have also been investigated by researchers. The possibility of returning back to the initial configuration in a finite number of time steps motivated the researchers to investigate the reversibility problem under ACAs~\cite{Sarkar12}. Because of the capability of the ACAs to generate pseudo random numbers and behaving as to model tranpositions of binary bits, ACAs have further been explored and utilized for designing symmetric-key cryptography algorithms~\cite{Sethi2016, Wolfram86}.

In following subsections~\ref{subsec:fullyACA}, and~\ref{subsec:alphaACA}, we give a special attention to understand the working of the main asynchronous CAs.
\subsection{Fully Asynchronous Cellular Automata}
\label{subsec:fullyACA}
One of the extensively studied type of asynchonism in the cellular automata is \textbf{f}ully asynchronous cellular automata. State of exactly one cell is updated while keeping states of remaining other cells unchanged~\cite{Nakamura}. This is the central/basic philosophy of the fully cellular automata. The fully asynchronous cellular autotmata together with the Wolfram's cellular automata have been extensively utilized by many researchers around the globe to model the real world problems. Be it number conservation property, classification or clustering problem, land use and city planning problem etc~\cite{Sethi2016, Bandini12, arrighi2013stochastic, jca/Fates09}. The Fig - ~\ref{faca} displays working of fully asynchronous cellular automata. 
\begin{figure}[h]
	\centering
	\includegraphics[scale = 0.8]{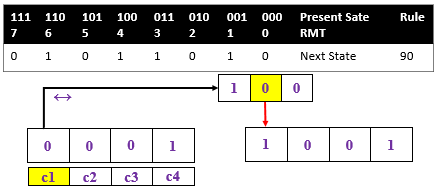}
	\caption{Working of {\em fully} Asynchronous Cellular Automata under periodic boundary condition.}
	\label{faca}
\end{figure}

Chapter-~\ref{chap3} is centered on the fully asynchronous cellular automata, where we have utilized it to solve clustering problem. We have utilized a total of seven real world datasets for experiment and exploration of the capability and applicability of the fully asynchronous cellular automata for clustering purpose.

\subsection{$\alpha$ - Asynchronous Cellular Automata}
\label{subsec:alphaACA}
$\alpha-$ asynchronism is highly sensitive and complex type of asynchronism in the cellular automata. In this, only those cells are updated uniformly and randomly those have probability at least $\alpha$, and remaining get never updated throughout the time step.
\emph{Alpha Asynchronous Cellular Automata} ($\alpha$ -ACA) is a type of asynchronous CA in which states of only $\alpha$ probability cells are updated simultaneusly, while states of (1-$\alpha$) probability cells are remain unchanged.

We try to know the working of the $\alpha$ - ACA by the example - 2.2 in which the elementary cellular automata has been used to demonstrate the algorithm.

\textbf{Example 2.1:} Suppose, we are given an $ECA - 90$, ring size $n = 5$, initial configuration $(1, 1, 0, 0, 1)$ and $\alpha$ = 0.5. Then, we are interested to see the next four configuration after evolved from  the given and subsequent evolved configurations.

Under the consideration of $\alpha -$ asynchronous CA update scheme, suppose that cell number $0$ and $1$ (i.e. two immediate neighbors) cells have $\alpha$ probability. So, therefore, these cells are select uniformly and randomly for update corresponding to the update rules based on the given $ECA - 90$. Then, using periodic boundary condition, the input patter based on the select cell number $0$ is $(1, 1, 1)$. Corresponding to this input pattern, the output pattern is $0$ for this cell. While, input pattern for cell number $1$ is $(1, 1, 0)$ and output pattern corresponding to this input pattern is $1$. Therefore, the first next configuration becomes: $(0, 1, 0, 0, 1)$. The other three next configurations have been summarized into the Table 2.2.

\begin{table}[h]
	\centering	
	\resizebox{0.8\textwidth}{!}{
		\begin{tabular}{|c|c|c|c|c|c|c|c|c|c|}
			\hline
{\bfseries Cell Number } & {\bfseries $0$} & {\bfseries $1$}& {\bfseries $2$} & {\bfseries $3$} & {\bfseries $4$} & {\bfseries Pr($\alpha$) - cells to be updated}\\\hline
Present State & 1 & 1 & 0 & 0 & 1 & \\\hline
Next State(i) & 0 & 1 & 0 & 0 & 1 & 0, 1\\
Next State(ii)& 0 & 1 & 1 & 1 & 1 & 0, 2, 3\\

Next State(iii)& 0 & 1 & 0 & 1 & 1 & 1\\

Next State(iv) & 0 & 0 & 0 & 1 & 1 & 0, 1, 2, 3, 4\\\hline
	\end{tabular}}
	\caption{The next four configurations after evolution of initial configuration (1, 1, 0, 0, 1) under the $\alpha$ fully asynchronous cellular automata and ECA - 90}\label{tab1}
\end{table}

\subsection{Hybrid Automata}
\label{hybrca}
When cells use different local rules, it breaks uniform local interaction with neighbors. It is the most popular type of non - uniformity in CAs. This model has been firstly studied and developed by Pries~\cite{Pries86}. Pries explored the group properties of one dimensional CAs with finite lattice size under both the null and periodic boundary conditions. Since then, it become a hot area of study among the researchers. For example, ~\cite{Hortensius89a, Das92, cattell1996synthesis, Supreeti_2018_chaos}.

In recent times, a more generalized definition for hybrid automata has been given by Cattaneo et al.~\cite{Cattaneo09}. Basic global characteristics such as surjectivity, injectivity, equicontinuity,
decidability, structural stability etc related to the hybrid automata have also been explored by him~\cite{Cattaneo09, Dennunzio12b, Dennunzio14a, salo2014}.

\section{Some Definitions and Terms}
\begin{definition}~\cite{Sethi2016,Sarkar11}
    A configuration $ \mathcal{X} \in \mathcal{E}_{n}$ is a recurrent configuration, if for every reachable configuration $ \mathcal{Y}$, the configuration $ \mathcal{X}$ is also reachable from the configuration $ \mathcal{Y}$.
\end{definition}
\begin{definition}~\cite{Sethi2016,Sarkar11}
    An ECA R under the $\alpha$ - asynchronous update scheme is reversible if each and every configuration $ \mathcal{X} \in \mathcal{E}_{n}$ is a recurrent configuration, otherwise irreversible.
\end{definition}
where $\mathcal{E}_{n}$ represents configuration space associated to a give CAs size of $n$.

\subsection{Boundary Conditions}
By virtue of its nature, the cellular space ($\mathcal{L}$) has no boundary and notion of boundary condition. It is  an infinite space in an euclidean space. However, some of the works have been done in the framework of cellular automata where $\mathcal{L}$ has been used with assumption of finite dimension. Therefore, the CAs those utilize such a cellular space ($\mathcal{L}$), are called finite CAs. Here, we discuss two main boundary conditions as follw:

\begin{definition}
\textit{A CA is called as a Finite Cellular Automaton if the cellular space $\mathcal{L} \in \mathbb{Z}^{D}$ is finite. Otherwise, it is an Infinite Cellular Automaton.}
\end{definition}

\begin{itemize}
\item[1]. \textbf{Open Boundary Condition:} In cellular automata (CAs) with open or fixed boundaries, the cells at the edges (the leftmost and rightmost cells in one-dimensional CAs) lack some neighbors, which are usually assigned predetermined states. One common type of open boundary condition is the \emph{null boundary}, where these edge cells' missing neighbors are consistently set to state 0 (Null). Null boundary conditions in CAs have been widely studied, especially in the context of Very Large Scale Integration (VLSI) design and testing. 
\begin{figure}[h!bb]
	\centering	
	\includegraphics[scale = 0.8]{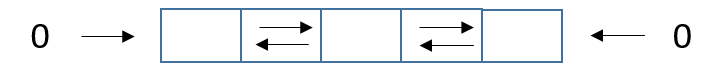}
	\caption{{Null boundary condition.}}
	\label{nb}
\end{figure}

\item[2]. \textbf{Periodic Boundary Condition:} the boundary cells are neighbors of some other boundary cells. For instance, for 1-D CAs, the rightmost and leftmost cells are neighbors of each other (see Fig. 3b).
\begin{figure}[h!bb]
	\centering	
	\includegraphics[scale = 0.8]{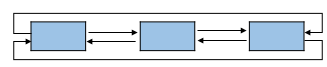}
	\caption{{Periodic boundary condition.}}
	\label{nb}
\end{figure}
\end{itemize}

We have utilized \textit{periodic boundary condition} throughout this dissertation.

\section{Applications}
This spreads light on the applications of cellular automata in diverse fields those have been published in the past.

ACA have been utilized to design the essential components of Cryptography algorithm such as permutation and combination of bits, transformation of bits, bit shifts stc.~\cite{Wolfram86}, parallel computing ~\cite{Das95}, Market dynamics, weather patterns, Earthquake and Astronomical simulations, neural networks~\cite{Hopfield}, modeling of spread of diseases through population~\cite{NakataY} etc. are some of the major applications where the asynchronous CA have been used for modeling purpose. We give some detailed discussion related to the application as follow:

\begin{itemize}
	\item[1]. \textbf{Cellular Automata for Electronic Circuit Design}: Because of the simplicity, cascading property and modularity, Non - Uniform CAs have captured attention of the researchers to utilize it for electronic circuit design in area of VLSI field of study. For example CAs Machine with high degree of parallelism have been designed and developed by Toffoli and Morgulus~\cite{Toffo87}, and studied the complex systems. Atrubin utilized it and developed a parallel multiplier~\cite{Atrubin65}. A quality of works in the literature deals with the application of CAs in cryptography~\cite{Wolfram86,Nandi94}.
	\item[2]. \textbf{Computer Vision}: Researchers have also explored the applications of CAs in computer vision, machine learning, pattern recognition, and in allied field of studied, where especially two - dimensional CAs have been utilized. For example, image processing\cite{Rosin2006, Kazar11}, CA base GF(2) coding transformer scheme for gray level~\cite{Paul1999}.
	
	\item[3]. \textbf{Medical and Bio-science}: Since CAs posses (i)self - replication property, (ii) simple and fast to implement a model, (iii) CAs based simulation provides remarkable resemblance with original experiments, therefore they have been utilized in the medical and bio-sciences for modeling purposes. For example, evolution of DNA within the CAs based framework~\cite{Burks94}, C${_p}$G island detection in DNA sequence~\cite{Ghosh07}.
	
\end{itemize}

In next Chapter~\ref{chap3}, we have presented the experimental set up and brief study of fully ACA with its application to solve the clustering problems. 

\chapter{Fully Asynchronous Cellular Automata for Clustering}
\label{chap3}

\section{Introduction}
\label{sec:S1}
Traditionally, a cluster is a collection of two or more objects which have similar characteristics (or features)~\cite{ref_r9}. A clustering technique aims to assign the target objects among
the clusters {\em effectively}. This metric also exploits the inherent anatomy of data objects for partitioning the dissimilar objects into separate clusters. Formally, an effective cluster can be defined as a set of objects those have minimum {\em intra-cluster} distance while the objects from different clusters have maximum {\em inter-cluster} distance. Thus a cluster establishes an intrinsic connection among the objects such that the objects are {\em reachable} to each other. So, clustering can be thought of as a bijective mapping $\mathcal{M}: {\mathcal{K}}\rightarrow {\mathcal{K}}$ where $\mathcal{K}$ refers to the set of target objects. Let us consider two target objects $x$ and $y$ of $\mathcal{K}$; they are {\em reachable} to each other if $x=\mathcal{M}^{l_1}(y)$ and $y=\mathcal{M}^{l_2}(x)$ for some finite $l_1$ and $l_2$. Therefore, we can get $|\mathcal{K}|!$ possible bijective mappings where every pair of target objects maintains reachability. As our target is to perform clustering effectively, therefore, to maintain less intra-cluster distance, the objects of a cluster need to maintain {\em closeness} and the number of clusters should be {\em optimal} also. The inherent hardness of this problem
motivates us to take finite sized reversible cellular automata (CAs) as natural clustering tool. Here, the {\em cyclic spaces} of reversible CAs can be used as cluster spaces where in every cycle, each pair of configurations (target objects in terms of clustering) maintains `reachability'. Besides, to maintain optimal number of clusters we select such reversible cellular automaton (CA) which produces $m$ number of cycles which is linear in the respective CA size. To maintain less intra-cluster distance, we select those CA rules which are capable of connecting the configurations which maintain minimum possible change in their cells’ state values (as the CA configurations are used as the data objects).

Reversible cellular automata have already proven its efficacy in data clustering~\cite{ref_r7,ref_r13,ref_r14}. In the prior art, the working principle of the clustering is to merge the {\em closely reachable} clusters at every level until it converges. Next, an efficient clustering scheme is reported in~\cite{ref_r11} which deals with high-dimensional data. The work mentioned in~\cite{ref_r12} introduces a novel encoding technique and this research is mainly used for characterization of CA rules for effective clustering. The state-of-art methodologies use the cellular automata with `synchronous' update of cells i.e., all the CA cells update their states at a particular time step and CA hops from the current configuration to the next configuration. If we add impurity in cell update, the cells are updated in non-deterministic fashion; such CA is called as {\em asynchronous} cellular automaton. Asynchronous CAs are flexible to model natural phenomena due to its independence in cells update. A tour on asynchronous cellular automata and some insightful thoughts of such automata in real world problem solving are described in~\cite{ref_r15}. Like classical CAs, the property of reversibility has also been studied in the CAs with non-determinism in evolution~\cite{ref_r3,ref_r4,ref_r16}. Traditionally, ``backward determinism"~\cite{ref_r17,ref_r18,ref_r19} is the key concept in reversibility as every configuration has exactly one predecessor and one successor in reversible cellular automata; but employing non-determinism in cell update introduces the possibility of multiple successors and predecessors in evolution. Thus to find reversibility in the domain of asynchronous cellular automata, instead of finding one-to-one correspondence in configuration space, we look into the fact whether every configuration can be reached by itself. The prior art~\cite{ref_r3,ref_r4,ref_r16} focuses only on finite {\em fully} asynchronous cellular automata under periodic boundary condition. In an $n$-cell fully asynchronous CA, at each time step, only one cell updates its state and such cell is selected uniformly at random. To establish the notion of reversibility in fully asynchronous CA, the concept of {\em communication class} has been introduced~\cite{ref_r3,ref_r4}. Formally, a communication class is a collection of configurations such that any two configurations always {\em communicate} to each other as they are {\em reachable} to each other; but there are multiple paths to establish the communication between any two configurations of that communication class. Therefore, a communication class can be thought of as a collection of simple cycles and for every two such cycles there exists an unique connection. This leads us to use asynchronous cellular automata for solving the problem of clustering for the following case.

Let us consider $|\mathcal{K}|$ number of objects have $p$ features $f_1$, $f_2$, $\cdots$, and $f_p$ and our target is to assign them in different clusters effectively. Let us consider $x$ be an object which is close to object $y$ for $p_1$ features and that $x$ is also close to object $z$ for $p_2$ features (here, $p_1, p_2 < p$). Therefore, $x$ and $y$ can be in same cluster; similarly, $x$ can also form cluster with object $z$. Now, there is an intrinsic connection between this two clusters as they have common objects. Thus communication classes of fully asynchronous CAs plays the appropriate role to deal with such situation.

In asynchronous CA, the cyclic space consists of a collection of disjoint communication classes. This work uses the communication classes of finite sized fully asynchronous reversible cellular automata under periodic boundary condition. Section~\ref{sec:basic} introduces the terminologies and notations of cellular automata related to this work and a brief description on asynchronous cellular automata has also been introduced. It has already been proven that out of 256 possible CA rules, only {\em forty six} rules are called as {\em recurrent} rules~\cite{ref_r4}; but not all these CAs can work for efficient clustering. The CAs which produce the extremities in communication classes can be debarred for clustering. Such extremities are $n$-cell CA with {\em one} communication class or $2^n$ communication classes or {\em two} communication classes with one configuration in one class and the remaining $2^n -1$ are in other class. Thus we get only {\em twelve} rules which can play the key role in data clustering by maintaining optimal number of clusters as well as less intra-cluster distance (Section~\ref{sec:S3}). This work focuses on the real world datasets with {\em quantitative} attributes. Here, we use the an encoding scheme referred in~\cite{ref_r13,ref_r14}. By this encoding scheme, the real valued objects will be converted to a binary strings of finite length which are called {\em useful} configurations (Section~\ref{subsec:encoding}). Our clustering algorithm is explained in Section~\ref{subsec:algo} in-terms of communication classes. The devised CA based algorithm has been tested and examined with the benchmark clustering algorithms such as K-means and hierarchical clustering techniques. Section~\ref{sec:S4} presents an experimental analysis on different datasets retrieved from  archive.ics.uci.edu/ml/index.php. For doing comparison of the proposed clustering scheme with other clustering algorithms, we need to validate the clusters quality. For measure the quality of the clusters, here, we use three well known metrics like {\em dunn index}, {\em silhouette score} and {\em DB score}. The experimental result establishes that fully asynchronous CA based clustering can also maintain quality clusters like K-means and hierarchical algorithms.

In the above background, this chapter is organized as follows. The fully asynchronous cellular automata is introduced in Section~\ref{sec:basic},establishment of relationship between CA and clustering in Section~\ref{sec:S3}. The main result of the work in Section~\ref{sec:S4}. Finally, Section~\ref{sec:S5} summarize this chapter and discuss the application of the proposed CA based clustering algorithm for solving clustering problems.

\section{Basics of Fully Asynchronous Cellular Automata and Cyclic Space}
\label{sec:basic}
For the presented work, we utilized the fully asynchronous cellular automata which is a kind of asynchronous cellular automata \cite{ref_r6,ref_r7, Sethi2016}.Under the scheme of the fully asynchronous cellular automata, update of state of only one cell takes place uniformly and randomly at next time step. In this work, we utilized one dimensional two state $\{ 0, 1\}$ and three - neighborhood $n$ - cell fully asynchronous cellular automata under the periodic boundary condition.\\ In a CA, collection of present state of all the cells of CA at a time step t is called the $\mathit{configuration}$ of the CA. These configurations can transit to another configuration under the applied kind of CA. Which can be defined in terms of the global transition function $\mathit{G}$:C$\rightarrow$C, C = $\{ 0, 1\}^n$. The global transition function $\mathit{G}$:C$\rightarrow$C, C = $\{ 0, 1\}^n$ describes the evolution of a configuration to another configuration, and therefore it represents the configuration space~\cite{ref_r7}. Suppose that we have a global configuration x and next configuration of it is  y = (x${_i}$) $\forall i \in$ n. Therefore, the both x, and y configuration belong to C. This can be shown using global transition function G as y = G(x). More precisely, let y${_i}$ be the next state of a cell x${_i}$ of configuration x, which produced the next configuration y. So, y${_i}$ = $\mathit{R}$(x$_{(i-1)}$, x$_i$, x$_{(i+1)}$), where y${_i} \in$ y, where $\mathit{R}$ is the local rule corresponding to cell i, and x$_{(i-1)}$, x$_i$, x$_{(i+1)}$ is the collection of present states of neighbors of cell i. This collection of the neighbors of a cell i is termed as rule min term (RMT), in which x$_{(i-1)}$ is the present state of left neighbor, x$_{i}$ is present state of cell i and x$_{(i+1)}$ is present state of right neighbor of the cell i. And if x${_i}$ = $\mathit{R}$(x$_{(i-1)}$, x$_i$, x$_{(i+1)}$), then this is called the $\mathit{self-replication}$. All the RMT which shows such type behaviour, are called $\mathit{self-replicating}$ RMTs. This property plays a vital role in formation of cyclic spaces. Suppose we have a configuration $\mathit{C}$, corresponding to a CA which has all the self-replicating RMTs. The, this configuration will reach its initial configuration in each iteration. It's mean it creates a cycle. Therefore, collection of these type of the cyclic configurations in CA, forms cyclic space. In Figure 1, for a CA of size four, we can see the configuration 0000, 1111, 1010, and 0101 are self replicating configurations. While, 0011, 0111, 0110, 1110, 1100, 1101, 0010, 1011, 1000, 1001, 0100, and 0001 are reachable to each other, are the part of the same cycle. It is clearly shown that configuration 0000 cannot be reached from 1100 or 0011. Therefore, it is obvious that a configuration of communication cycle c$_{i}$ and cannot be reached from any configuration of cycle c$_{j}$, and vice - versa.
The CA rules are represented into a binary string of 0 and 1 as shown in table 1. Since, we have used fully ACA under periodic boundary condition. So, the extreme right cell of the CA would become the leftmost neighbor of the extreme left cell of the lattice. In same fashion, the extreme left cell of the lattice would become the right heighbor of the extreme rightmost cell of the CA. Some of the rules and with their present and text have been presented below in the Tabel~\ref{Table1}.\\
For the study of the work presented in this report, the fully asynchronous CA was used. A CA in which only one cell is updated randomly and uniformly is called a fully asynchronous CA. This is one after the classical CAs which have been explored widely and utilized in many applications.
\begin{center}
\begin{figure}[h]
\includegraphics[width= 12 cm, height = 6 cm]{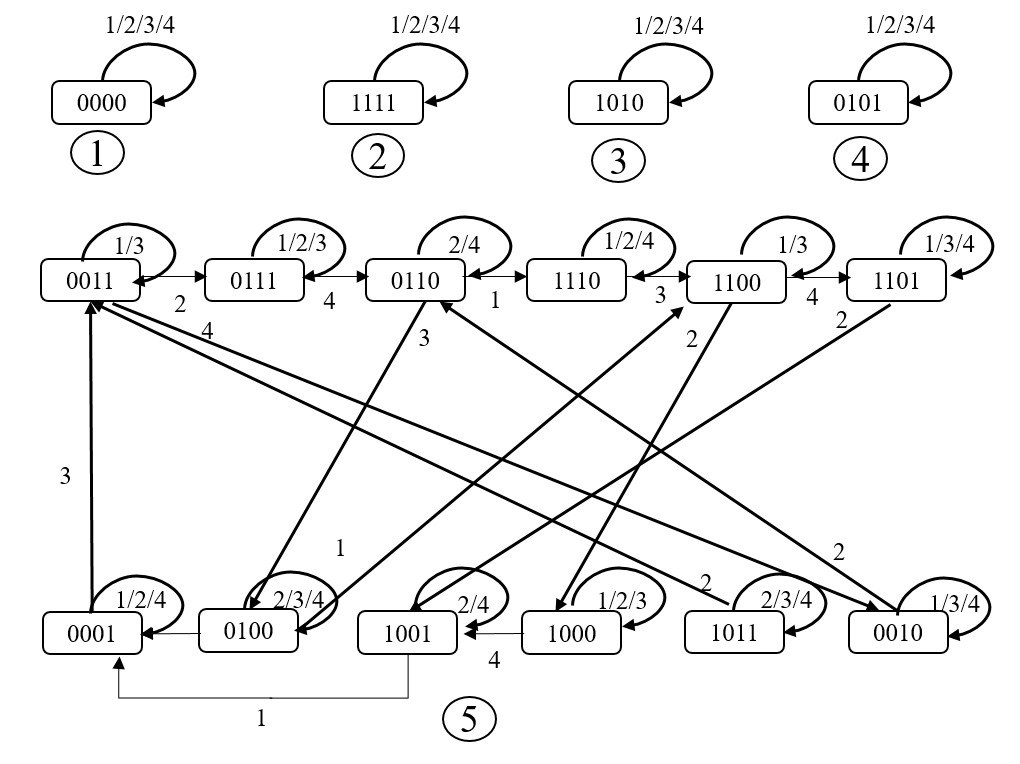}
    \caption{Cyclic space and communication classes of reversible ECA - 142 under fully asynchronous CA.}
\end{figure}  
\end{center}

\begin{table}
\centering
\begin{tabular}{ p{4.5 mm} p{4.5 mm} p{4.5 mm} p{4.5 mm} p{4.5 mm} p{4.5 mm} p{4.5 mm} p{4.5 mm} p{4.5 mm} p{4.5 mm} } 
 \hline
PS & 111 &110& 101 & 100 &011 &010 &001 &000 & Rule\\\hline
NS & 1 & 0 & 0 & 0 & 1 & 1 & 1 & 0 & 142\\ 
NS & 0 & 1 & 1 & 0 & 1 & 1 & 0 & 0 & 108\\
NS & 1 & 1 & 0 & 0 & 1 & 0 & 1 & 1 & 201\\\hline
\end{tabular}
\caption{Present State (PS), and Next State (NS). This table shows an evolution of few rules from their present state to next state.}
\label{Table1}
\end{table}

\section{Reversible ECA for Clustering problem}
\label{sec:S3}
We already have discussed about the reversibility and reachability in the section of basics of cellular automata. In this section we present those reversible CA rules which are used to study the clustering problem under the application of the fully asynchronous CA.

It has already been reported that there are only fourty six CA rules which shows reversible characteristic under the fully asynchronous system of CAs. These are 33, 35, 38, 41, 43, 46, 49, 51, 52, 54, 57, 59, 60, 62, 97, 99, 102, 105, 107, 108, 113, 115, 116, 118, 121, 123, 131, 134, 139, 142, 145, 147, 148, 150, 153, 155, 156, 158, 195, 198, 201, 204, 209, 211, 212, and 214 \cite{ref_r3,ref_r4,Sethi2016}. But, the all these rules have their different size of communication classes (cyclic space). A detailed description about the communication classes of these rules is represented in Table 2. The other remaining reversible ECAs will not be used for clustering under the fully asynchronous CAs, because they have either only one communication class of size 2$^n$ or two communication classes such that one communication class of size one and other of size 2$^n$-1, where CA size is $n$. It is also because of our motive to optimize the clustering of the objects\cite{ref_r4,Sethi2016}.\\
\begin{table}
\centering
\begin{tabular}{ |p{1cm}| p{0.6cm}| p{0.6cm}| p{0.6cm}| p{0.6cm}|p{0.6cm}|p{0.6cm}| p{0.6cm}| p{0.6cm}| p{0.8cm}| p{0.8cm}|p{0.8cm}|p{0.8cm}|}\hline
CA Size(n) & 105 & 134 & 142 & 148 & 150 & 158 & 212 & 214 & 108 & 201 & 156 & 198 \\\hline\hline
8 & 7 & 7 & 7 & 7 & 7 & 7 & 7 & 7 & 90  & 90 & 48 & 48\\ 
9 & 2 & 7 & 7 & 7 & 7 & 7 & 7 & 7 & 154 & 154 & 76 & 76\\
10 & 5 &8 & 8 & 8 & 8 & 8 & 8 & 8 & 270 & 270 & 123 & 123\\
11 & 2 &8 & 8 & 8 & 8 & 8 & 8 & 8 & 471 & 471 & 200 & 200\\ 
12 & 9 &9 & 9 & 9 & 9 & 9 & 9 & 9 & 825 & 825 & 322 & 322\\\hline
\end{tabular}
\caption{Number of communication classes against CA size n of reversible ECA mentioned in this table.}
\label{Table2}
\end{table}
The communication classes of ECAs 134, 142, 148, 150, 158, 212 and 214 are exactly same with respect to their size and member elements. But ECAs 108 and 201 have different characteristic, as sometimes they are found to have exactly same communication classes and sometimes one or two communication classes rule out this behaviour. Similarly, ECAs 156 and 198 also behaves like the ECAs 108 and 201. Therefore, to design and test fully asynchronous CAs as a clustering tool, we can choose one from ECAs 134, 142, 148, 150, 158, 212, and 214. But we cannot choose only one from 108 and 201, and 156 and 198. Another ECA to be used for designing is 105. So, we are left with only six choices. These choices are as follows:\\
\begin{itemize}
\item Select any ECA from 134, 142, 148, 150, 158, 212 and 214
\item Select 105, and (or) select 108, or(and) 201, and (or) select 156, or(and) 198 
\end{itemize}

If we restrict to use at least two ECAs to design fully ACAs as a clustering tool, then we found 15 different combinations with 2 ECAs, 20 different combinations with 3 ECA, 15 different combinations with 4 ECAs, 6 different combinations with 5 ECAs and one single combinations with six ECAs. Here, we have considered 134, 142, 148, 150, 158, 212 and 214 as a single ECAs. Because they behave exactly same. In the following section,we will explain the encoding technique to encode the dataset into binary strings. We will also explain the strategy to construct relation between CAs and clustering.
\section{Establishment of Relation between CA and Clustering }
As we have discussed that we have used one dimensional, three neighborhood, and two state fully asynchronous cellular automata under periodic boundary condition. So, each cell of lattice of CA will have either 0 or 1 state. Therefore, a configuration can be viewed as a binary string. Therefore, we require to encode the given objects into the binary strings. These strings were utilized by the CA based clustering algorithm for mapping to configurations under the CAs~\cite{ref_r7}.

\subsection{Techniques for Encoding Datasets}
\label{subsec:encoding}
Suppose that we have dataset (set of target objects) X = $\{$X$_1$, X$_2$, ..., X$_u\}$, where m is the size of dataset. Each target object has attribute A = $\{$A$_1$, A$_2$, ..., A$_v\}$, each A$_i$ with a set of finite values . It is possible that all the attributes either be qualitative or quantitative, or combination of both types. Since the CAs are in binary form, therefore, all the objects were essentially needed to be encoded into binary form. If it would not be done correctly, we would not get optimal and correct clustering. Therefore, to achieve intracluster distance as minimum as possible, we considered Hamming distance as prime factor to convert the target objects into binary strings. The method (frequency based encoding) already have been utilized and described in ~\cite{ref_r7}.\\ A quantitative attribute should be converted into binary string in such way that the all the objects could be arranged in either ascending or descending order together with the minimum Hamming distance between them. Now, We represent a hypothetical example to explain the encoding technique.\\\\
\textbf{Example 3.1:} Take a hypothetical set of books, where each book has three attributes: (1) number of pages, (2) ratings by reviewers and (3) type of binding. The first two attributes are quantitative and continuous in nature, whereas the last one is qualitative and categorical in nature. We have summarized the encoding details in Table 3 as shows below. In this scheme, qualitative attribute values are encoded as 01 (hard) or 10 (soft). The quantitative attribute values are encoded into three sub intervals to be represented by 00, 01 and 11, respectively. For example, the number of pages values are divided into sub intervals $[40, 110]$, $[120, 300]$
and $[325, 400]$, represented by 00, 01 and 11, respectively. So, the all the target objects are of six bit binary string. These strings can be utilized to show configuration of CAs of size six cells.
\begin{table}[h]
\centering
\begin{tabular}{ |p{0.5cm}| p{1.0cm}| p{1.5cm}| p{1.2cm}| p{1.5cm}|p{1.2cm}|p{1.5cm}|p{2.0cm}|}\hline
IDs & No of Pages &  Col1 - Binary String & Ratings & col2 - Binary String & Binding Type & col3 - Binary String & Encoded CAs\\\hline\hline
1  & 40  & 00 & 9.5 & 11 & soft & 01 & 001101\\ 
2  & 200 & 01 & 4   & 00 & hard & 10 & 010010\\
3  & 300 & 01 & 9   & 11 & hard & 10 & 011110\\
4  & 325 & 11 & 8   & 11 & soft & 01 & 111101\\ 
5  & 129 & 01 & 4.5 & 01 & soft & 01 & 010101\\
6  & 65  & 00 & 7   & 01 & hard & 10 & 000110\\
7  & 319 & 11 & 6.8 & 01 & soft & 01 & 110101\\ 
8  & 110 & 00 & 3   & 00 & soft & 01 & 000001\\
9  & 400 & 11 & 2.6 & 00 & soft & 01 & 110001\\
10 & 350 & 11 & 9.3 & 11 & soft & 01 & 111101\\\hline
\end{tabular}
\caption{Encoding of quantitative and qualitative attributes in binary form.}
\label{Table3}
\end{table}

In Table 3, the respective attributes of the objects have been encoded. In Table 4, the objects have been represented in the form of respective binary string, which produced by utilizing the binary form of their attribute value presented in Table 3.\\

Next, we have presented fully asynchronous CAs as a natural cluster, which connects these objects together to form clusters. So as to show that each object can be connected to a unique cyclic space in CAs.\\Under introduction and basics of cellular automata, we discussed about two important properties: (1) reversibility, and (2) reachability. We also had discussed about the bijection that every reversible CAs can be represented as a bijective mapping. Therefore, the reachable objects are said closed objects and belong to same group (cluster). Hence, the reachability property can be used to treat the CAs as the natural clustering tool.\\

From the above discussions, we can conclude that a CAs can be found effective for clustering if it maintains the following two conditions: (i) all the cells of CA must follow CA rules that contribute a reasonable number of self-replicating RMTs and (ii) their number of cycles should be limited. We already have represented these reversible CA rules in section 2.1.
In next subsection, we represent our cycle based clustering for fully asynchronous CAs that was designed by~\cite{ref_r7} for non uniform cellular automata.

\subsection{Algorithm for Cycle Based Clustering}
After detailed discussion on essentially and primarily required materials, we are presenting first algorithm for clustering using fully asynchronous cellular automata. The algorithm was given by~\cite{ref_r7}. It is named cycle based clustering algorithm, for clustering those objects into same cluster which are reachable to each other and forms a cycle. Let we have x, and y be two objects such that x, y $\in$ c$_i \cap$  c$_j$ . Then, any object z $\in c_j$ will not be reachable from any object of c$_i$.  Since x and y belongs to same cluster, therefore, they form a cycle. In forming these cycles and clusters, hamming distance plays an important role. The distance ensures that CA does not grows exponentially but maintains the minimum intractuster distance.
The steps of the algorithm are as below:

\begin{enumerate}
\item Based on the types of attributes (quantitative or qualitative), the target objects of a given dataset are encoded into the set of n-bit useful configurations that is denoted by $c$.
\item Randomly choose an arbitrary CA $\mathit{R}$ of size n that maintains rules at all cells. But choose anyone from the CAs 134, 142, 148, 150, 158, 212, and 214. Since they have same property and behaviour. With that CA, choose 105; 156, 198; 108 and 201. Discussed in Section~\ref{sec:S3}.

\item Let C' be the set of remaining objects to be clustered. So, initially C' = C and Set k = 1.

\item Let the configurations those are close to x$_i$ be C$^k$ such that C$^k \subset$ C'. Set C$^k$ = C$^k \cup $  x$_i$, C' =  C'$\backslash$ C$^k$. Increment k by 1.

\item And, iterate step 4 until all the configurations are clustered, and return total number of clusters $`m'$.
\end{enumerate}

\subsection{Algorithm to produce k clusters}
\label{subsec:algo}
The algorithm was invented by~\cite{ref_r7} where they used it to study non uniform cellular automata for clustering problem. The algorithm can be found in~\cite{ref_r7}. Though we have presented it as below algorithm 1:\\
\textbf{Input:} A dataset (set of target objects), required clusters (m), and an auxiliary CA space.\\
\textbf{Output:} The clusters $\{ C_1, C_2,..., C_m\}$\\\\
\textbf{Example 3.1:} We explain our algorithm by this example using real world dataset (Engineering Colleges - collected from Kegle plateform).\\
The datset has total 5 quantitative attributes those are  Teaching, Fees, Placements, Internship, and Infrastructure. The encoding of the attributes were done as per the procedure explained in the sub section 3.1 and discussed in example 3.1.\\\
The size of each encoded object become n = 15. So, we used reversible CA of size 15 X = $\{X_1, X_2, ... X_{25}\}$. We found 14 useful (unique) configurations out of the total 26 configurations. At level 0, these 14 configurations become the primary clusters C and those are:\\
X$_1$ = 001001001001001, X$_2$ =  001010001001001, X$_3$ = 001010001010010, X$_4$ =  010010001010010, X$_5$ = 010010001010011, X$_6$ =  010011010010011, X$_7$ = 011011010010011, X$_8$ = 011100010010011, X$_9$ =  011100011010011,\\ 
X$_{10}$ = 011100011011011,
X$_{11}$ = 011100011011100, X$_{12}$ = 011100100011100, X$_{13}$ = 011100100100100, X$_{14}$ = 100100100100100. The algorithm generates total 14 primary clusters at level 0. These clusters are C$^{0}_1$ = $\{$ X$_1$ $\}$, C$^{0}_2$ = $\{$ X$_2$ $\}$,..., C$^{0}_{14}$ = $\{$ X$_{14}$ $\}$
The algorithm selects a CA $\mathit{R}$ uniformly and random from the list of candidate CA and applies the clustering procedure. It generates 5 auxiliary CA such that aux clusters are auxC$_{1}$ = 
$\{$  C$^{0}_1$, C$^{0}_2$, C$^{0}_3$, C$^{0}_4$, C$^{0}_9$, C$^{0}_{10}$ , C$^{0}_{14}$ $\}$
auxC$_{2}$ = $\{$  C$^{0}_5$, C$^{0}_8$, C$^{0}_{11}$ $\}$
auxC$_{3}$ = $\{$  C$^{0}_6$ $\}$
auxC$_{4}$ = $\{$  C$^{0}_{12}$, C$^{0}_{13}$ $\}$
auxC$_{5}$ = $\{$  C$^{0}_7$ $\}$, where each C$^{0}_i$, i $\in \{1, 14\}$  has degree of participation of 100\%. It forms a total of 5 primary clusters for next level. On next second level an another CA were selected and used to clusters these 5 primary clusters to produce three desired clusters. This level produces the following auxiliary clusters
auxC$_{1}$ = $\{$  C$^{0}_1$ $\}$, degree of participation is 100\%.
auxC$_{2}$ = $\{$  C$^{0}_2$, C$^{0}_4$ $\}$ degree of participation of C$^{0}_2$ is 100\%, and of C$^{0}_4$ is 100\%
auxC$_{3}$ = $\{$ C$^{0}_3$, C$^{0}_5$ $\}$, degree of participation of C$^{0}_3$, and C$^{0}_5$ is 100\%, and 100\%. Since it clusters those primary clusters which have maximum degree of participation that is first maximum and second maximum. Here each participating primary cluster has maximum degree of participation 100\%. Hence it procudes the three clusters as below:\\
C$_{1}$ = $\{$  C$^{0}_1$, C$^{0}_2$, C$^{0}_3$, C$^{0}_4$, C$^{0}_9$, C$^{0}_{10}$ , C$^{0}_{14}$ $\}$, 
C$_{2}$ = $\{$  C$^{0}_5$, C$^{0}_8$, C$^{0}_{11}$, C$^{0}_{12}$, C$^{0}_{13}$ $\}$, and
C$_{3}$ = $\{$ C$^{0}_6$, C$^{0}_7$ $\}$.
\begin{algorithm}[h]
\caption{Algorithm for clustering given target objects}
\textbf{Step 1:} Read Encoded Target Objects (M), where $|M| = k'$\\
\textbf{Step 2:} Read n-cell auxiliary CA space\\
\textbf{Step 3:} \textbf{Input:} required number of clusters `m', and m$_0$ = k'\\
\textbf{for} j=1 to m$_0$

$\space \hspace{1cm}$ set c${^0}_{j} \leftarrow \{x_j\}$;\\
\textbf{Step 4:} while(True):

$\space \hspace{1cm}$ select auxiliary CA space randomly and uniformly

$\space \hspace{1cm}$ generate auxiliary clusters C$^{i}{_1}$, C$^{i}{_2}$,..C$^{i}{_m}$ for selected 

$\space \hspace{1cm}$ CA space.

$\space \hspace{1cm}$ create and initialize a matrix to store degree of 

$\space \hspace{1cm}$ participation A[a$_{ij}$]$_{m'*m_{i-1}}$ to 0

$\space \hspace{1cm}$ while t < m':

$\space \hspace{2cm}$  for j = 1 to m$_{i-1}$

$\space \hspace{2.5cm}$ a$_ij \leftarrow $ degPart(C$^{i}{_t}$, C$^{i-1}{_j}$)

$\space \hspace{1cm}$  for j = 1 to m$_{i-1}$

$\space \hspace{1cm}//$ find first maximum degree of participation in previous 

$\space \hspace{1cm}$ primary clusters

$\space \hspace{1cm}//$ suppose a$_{t'j}$ = max(a${_tj}$), 1$<=t<=$m

$\space \hspace{1cm}//$ find auxiliary cluster with max deg of part of the c$^{i-1}_{j}$

$\space \hspace{2cm}$ for j${_1}$ = 1 to m$_{i-1}$ and (j1 $!=$ j)

$\space \hspace{2.5cm}$ a$_{t'j'}$ = second maximum of a$_{tj1}$

$\space \hspace{2.5cm}$ if  a$_{t'j'}$ = 0 then $\mathit{Continue}$

$\space \hspace{2.5cm}$ else
 
$\space \hspace{3cm}$ select the cluster and update auxiliary

$\space \hspace{3cm}$ cluster set, increment z by 1

$\space \hspace{3cm}$ mark c$^{i-1}_j$ and c$^{i-1}{_j'}$ selected

$\space \hspace{3cm}$ delete the corresponding row t' from

$\space \hspace{3cm}$ matrix A

$\space \hspace{2cm}$ for all non selected clusters

$\space \hspace{2.5cm}$ select c$^i_{z}$ and increment z by 1

$\space \hspace{2cm}$ now set m$_i$ = z and increment i by 1

\textbf{Step 5:} \textbf{Output:}Print all the clusters produced by algorithm\\
\textbf{Exit}

\end{algorithm}

\subsection{Association of Reversible CA with k Cluster}
This section presents the discussion on the reversible CA associated to produce $k$ cluster corresponding to a given dataset. When the aove the experiment were performed, it was observed that the type of dataset, and association of reversible CA plays an important role in clustering under the cycle based clustering. An incorrect association of CAs and dataset may result into undesired clustering. Therefore, the Figure -3.2 and 3.3 depicts the association of CAs corresponding to a given dataset and desired number of clusters (k). 

\begin{figure}[h!]
\centering
		\includegraphics[scale = 0.55]{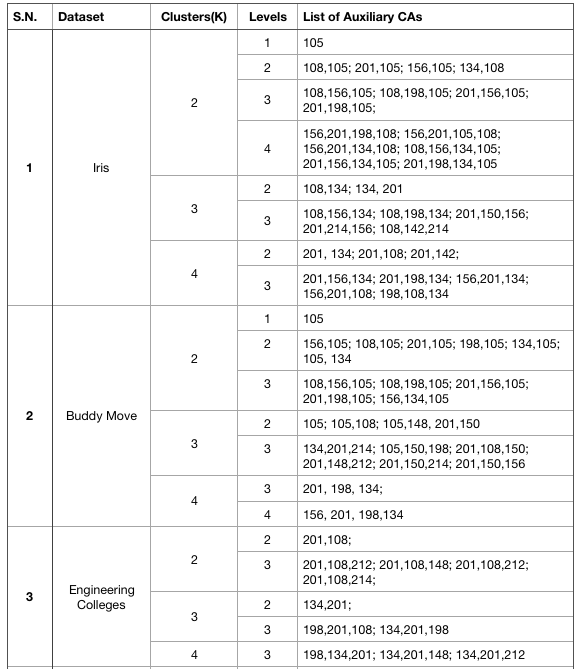}
	\caption{user Knowledge, Seeds, Stoneflakes and Segmentation datasets with associated clusters (k), level by which these clusters can be obtained corresponding to a given list of auxiliary reversible CA.}
\end{figure}
\begin{figure}[h!]
\centering
		\includegraphics[scale = 0.55]{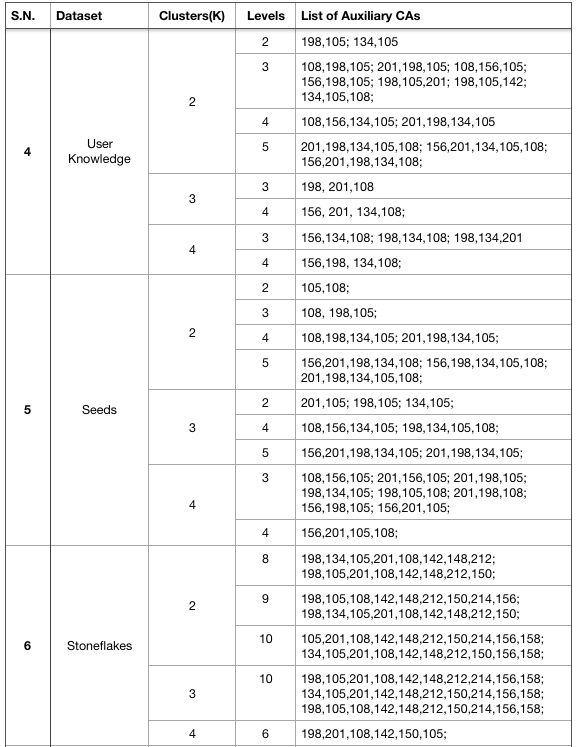}
            \includegraphics[scale = 0.481]{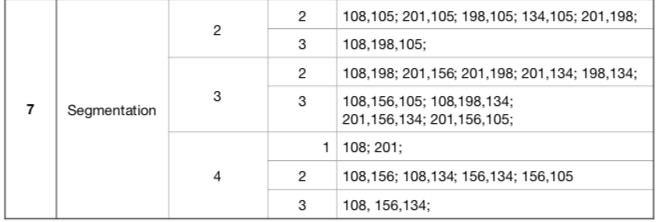}
	\caption{User Knowledge, Seeds, Stoneflakes and Segmentation datasets with associated clusters (k), level by which these clusters can be obtained corresponding to a given list of auxiliary reversible CA.}
\end{figure}

\section{Results and Discussion}
\label{sec:S4}
This section reports the performance and results of our iterative cycle-based clustering algorithm on some real datasets (Iris, Buddy-Move, User Knowledge, Seeds, Engineering Colleges, Segmentation, and Stonflake). These datasets can be collected from Kaggle and archive.ics.uci.edu/ml/in/dex.php.
The accuracy of the proposed CA based cyclic clustering algorithm and of the other three benchmark algorithms K-means, Hierarchical, and PAM (Partitioned Around Medoids) clustering algorithms were evaluated using the clustering validation indices and scores those are silhouette score, dunn index and db score provided by python package. The performance of each score corresponding to an algorithm and dataset for a given cluster size have been presented in tables Table 5 and Table 6. Table 5 and Table 6 contains the results of benchmark algorithm  clustering respectively for cluster k = 2, and k = 3.\\ For the performance analysis and their validation, we used benchmark cluster validation scores which are Dunn score, Silhouette score and DB score. The set of auxiliary CA with the smallest Davies-Bouldin score and highest silhouette score and Dunn index are considered to generate best quality of clustering results. In this work, we used seven datasets: Iris, Buddy Move, User Knowledge, Seed, Stoneflakes, Segmentation, and Engineering Colleges. Each of them has mostly quantitative attributes. The details of these datasets is reported into Table~\ref{Table4}.\\
In Table~\ref{Table4}, column 3 shows the number of quantitative attributes (Qnt attributes), and column 5 shows the target objects that are to be distributed among m clusters using an n-cell CA, where column 6 shows the size of CA for respective dataset.\\
In Table 5, we can observe that the proposed CA based algorithm, almost works well for clustering problems. For K=2 required clusters, it performs well for target objects of Iris, Buddy Move, User Knowledge, seeds. While it gives some comparative results for iris and seeds target objects. We can see their comparative performance in the Figure~\ref{fig:1a},\ref{fig:2a}, \ref{fig:3a}, \ref{fig:4a}, \ref{fig:5a}, \ref{fig:6a}, \ref{fig:7a} (for required clusters of k = 2). We observed that the Silhouette and Dunn score shows comparative performance to the benchmark clustering algorithms. The CA based clustering algorithm performed well for buddymove, user knowledge datasets while it shows comparative performance for Iris dataset. We observed that the CA based algorithm performs well for buddymove, and user knowledge (k=2). \\
If we talk about the three required clusters (K=3), Table~\ref{Table6} presents the statistical performance of the algorithm applied on the datasets presented in Table~\ref{Table4}. The proposed algorithm gives a comparative scenario of results with respect to smallest Davies-Bouldin score and highest silhouette score and Dunn index scores. The proposed CA based cyclic clustering algorithm performs as better as the PAM algorithm. This was observed from the Fig~\ref{fig:2b}, \ref{fig:5b}, \ref{fig:6b}, \ref{fig:7b} (for k=3). The observations made on the segmentation, stoneflakes, user knowledge, and iris supports our CA based algorithm, for better silhouette, and dunn scores which is depicted in the Figures 2 to 8 and same in Table~\ref{Table5} and Table~\ref{Table6}. \\ 
The reversible CAs 105, 108 214 or 108, 201, 214 can be used to  get two cluster of iris dataset. The required clusters would be processed at level 3 and 4. Similarly, two clusters of user knowledge would also be obtained with auxiliary CAs 158, 201 or 158 108, or 201 and 212 at level 4. The result of three clusters (k=3) presented in this article were obtained with the following auxiliary CA with respect to the dataset as mentioned: Iris - 108, 201, 212, 108, 105, 212. Buddy Move - 156, 198, 108, 201 and any CA from 134, 142, 148, 150, 158, 212 and 214. User Knowledge - 105, 156, 201, and 108,201. Seeds - 108, 214, 201 (for k=2, and k=3).

\begin{figure}[h]
	\vspace{1pt}
	\subfloat[\label{fig:1a}]{%
		\includegraphics[scale = 0.55]{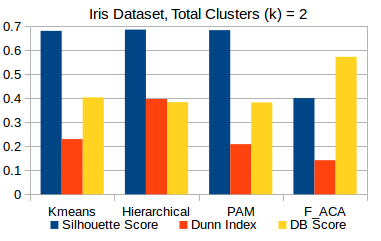}
	}
	\hfill
	\subfloat[\label{fig:1b}]{%
		\includegraphics[scale = 0.55]{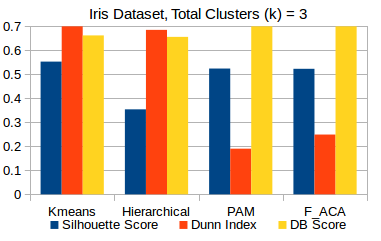}
	}    
	\vspace{-2pt}
	\caption{Iris Dataset, Performance of Silhouette, Dunn and DB Score: (a) Cluster (K = 2), (b) Cluster (K=3).}
\end{figure}

\begin{figure}[h]
	\vspace{-1pt}
	\subfloat[\label{fig:2a}]{%
		\includegraphics[scale = 0.55]{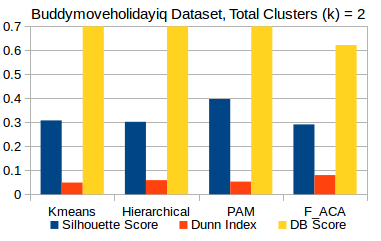}
	}
	\hfill
	\subfloat[\label{fig:2b}]{%
		\includegraphics[scale = 0.55]{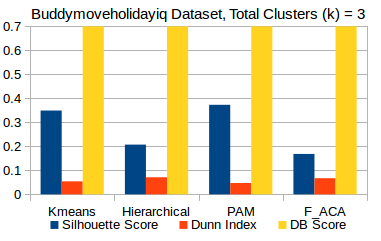}
	}    
	\vspace{-2pt}
	\caption{Buddymove Holidayiq Dataset, Performance of Silhouette, Dunn and DB Score: (a) Cluster (K = 2), (b) Cluster (K=3).}
\end{figure}

\begin{figure}[h]
	\subfloat[\label{fig:3a}]{%
		\includegraphics[scale = 0.55]{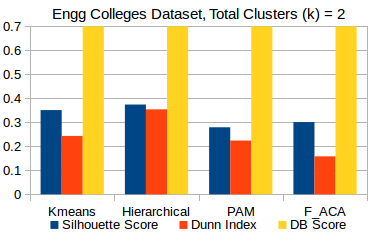}
	}
	\hfill
	\subfloat[\label{fig:3b}]{%
		\includegraphics[scale = 0.55]{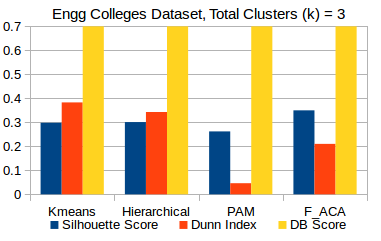}
	}    

	\caption{Engineering College Dataset, Performance of Silhouette, Dunn and DB Score: (a) Cluster (K = 2), (b) Cluster (K=3).}
\end{figure}

\begin{figure}[h]
	\vspace{-5pt}
	\subfloat[\label{fig:4a}]{%
		\includegraphics[scale = 0.55]{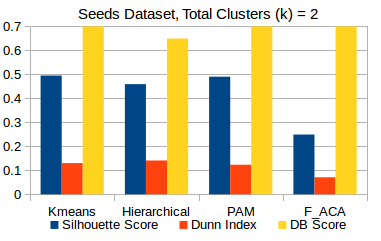}
	}
	\hfill
	\subfloat[\label{fig:4b}]{%
		\includegraphics[scale = 0.55]{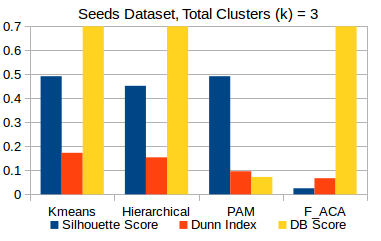}
	}    
	\caption{Seeds Dataset, Performance of Silhouette, Dunn and DB Score: (a) Cluster (K = 2), (b) Cluster (K=3).}
\end{figure}

\begin{figure}[h]
	\vspace{-8pt}
	\subfloat[\label{fig:5a}]{%
		\includegraphics[scale = 0.55]{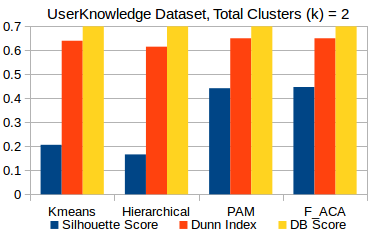}
	}
	\hfill
	\subfloat[\label{fig:5b}]{%
		\includegraphics[scale = 0.55]{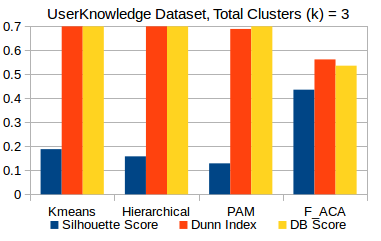}
	}    
	\caption{User Knowledge Dataset, Performance of Silhouette, Dunn and DB Score: (a) Cluster (K = 2), (b) Cluster (K=3).}
\end{figure}

\begin{figure}[h]
	\vspace{-8pt}      
	\subfloat[\label{fig:6a}]{%
		\includegraphics[scale = 0.55]{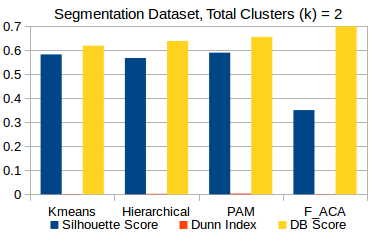}
	}
	\hfill
	\subfloat[\label{fig:6b}]{%
		\includegraphics[scale = 0.55]{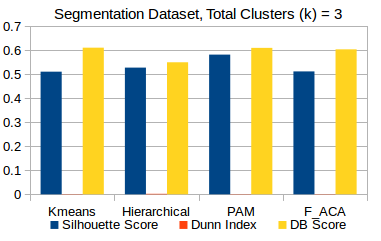}
	}    
	\caption{Segmentation Dataset, Performance of Silhouette, Dunn and DB Score: (a) Cluster (K = 2), (b) Cluster (K=3).}
\end{figure}

\begin{figure}[h]
	\vspace{-8pt}
	\subfloat[\label{fig:7a}]{%
		\includegraphics[scale = 0.55]{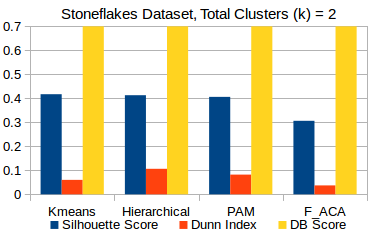}
	}
	\hfill
	\subfloat[\label{fig:7b}]{%
		\includegraphics[scale = 0.55]{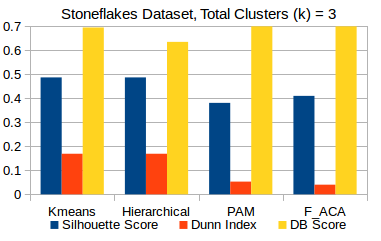}
	}    
	\caption{Stoneflakes Dataset, Performance of Silhouette, Dunn and DB Score: (a) Cluster (K = 2), (b) Cluster (K=3).}
\end{figure}

\begin{table}[h]
\centering
\begin{tabular}{ |p{0.5 cm}| p{3.0 cm}| p{2.2 cm}| p{2.0 cm}| p{2.0 cm}|p{1.4 cm}|}\hline
S.N. & Datasets & Qnt attributes & Qlt attributes & Total Objects & CA Size \\\hline\hline
1  & Iris   & 4  & 0 & 150 & 8 \\ \hline
2  & Buddy Move  & 5 & 0 & 403 & 10 \\\hline
3  & User Knowledge  & 6 & 0 & 249 & 12 \\\hline
4  & Seeds  & 7 & 0 & 210 & 14\\\hline
5  & Segmentation  & 4 & 0& 2000 & 8\\\hline
6  & Engineering Colleges  & 5 & 0 & 25 & 10\\\hline
7  & Stoneflakes  & 8 & 0 & 80 & 16\\\hline
\end{tabular}
\caption{Details of datasets which were used for experiment.}
\label{Table4}
\end{table}

\begin{table}[h]
\begin{center}
\begin{tabular}{ |p{2.5 cm}|p{2.5cm}|p{1.5cm}|p{2.0cm}|p{1.5cm}|p{1.5cm}|}
\hline
Datasets & Cluster Validation Indices & K-Means & hierarchical & PAM & fully ACA \\
\hline
\multirow{4}{4em}{Iris}
& Silhouette 	& 0.681& 0.686 & 0.684 & 0.401\\ 
& Dunn 	        & 0.230& 0.397 & 0.209 & 0.142\\ 
& DB Score   	& 0.404& 0.384 & 0.383 & 0.573\\
& Calinski H   	& 513.303& 501.924 & 500.112 & 500.601\\
\hline
\multirow{4}{4em}{Buddy Move} 
& Silhouette 	& 0.308& 0.302 & 0.397 & 0.291\\ 
& Dunn 	        & 0.049& 0.059 & 0.053 & 0.0801\\ 
& DB Score   	& 1.338& 1.211 & 0.913 & 0.622\\ 
& Calinski H   	& 119.973& 123.118 & 118.404 & 115.301\\
\hline
\multirow{4}{4em}{Stoneflakes} 
& Silhouette 	& 0.417 &0.413 & 0.406 & 0.306\\ 
& Dunn 	        & 0.060 &0.106 & 0.082 & 0.037\\ 
& DB Score   	& 1.010 &1.009 & 1.040 & 1.300\\ 
& Calinski H   	& 46.596&46.359 & 44.542 & 30.125\\
\hline
\multirow{4}{4em}{Seeds} 
& Silhouette 	& 0.495 &0.459 &0.490 &0.249\\ 
& Dunn 	        & 0.130 &0.141 &0.123 &0.071\\ 
& DB Score   	& 0.741 &0.649 &0.754 &0.916\\ 
& Calinski H   	& 310.501& 208.568 & 275.845 & 400.012\\
\hline
\multirow{4}{4em}{Segmentation} 
& Silhouette 	& 0.583 &0.568 &0.590 &0.351\\ 
& Dunn 	        & 0.001 &0.002 &0.003 &0.001\\ 
& DB Score   	& 0.619 &0.639 &0.656 &1.035\\ 
& Calinski H   	& 2777.304& 2753.427 & 2520.399 & 2201.001\\
\hline
\multirow{4}{4em}{Engg Colleges} 
& Silhouette 	& 0.351 &0.374 &0.279 &0.301\\ 
& Dunn 	        & 0.243 &0.354 &0.224 &0.158\\ 
& DB Score   	& 1.067 &1.037 &0.972 &2.250\\ 
& Calinski H   	& 19.051& 18.969 & 19.051 & 15.561\\
\hline
\multirow{4}{4em}{User Knowledge} 
& Silhouette 	& 0.206 &0.166 &0.442 &0.447\\ 
& Dunn 	        & 0.640 &0.615 &0.650 &0.650\\ 
& DB Score   	& 1.840 &2.142 &0.853 &0.853\\ 
& Calinski H   	& 109.996& 75.347 & 70.837 & 60.521\\
\hline
\end{tabular}
\end{center}
\caption{Clustering for cluster size K = 2. Table represents the cluster validation indices and scores of different datasets used for clustering with different clustering algorithms.}
\label{Table5}
\end{table}

\begin{table}[h]
\begin{center}
\begin{tabular}{ |p{2.5 cm}|p{2.5cm}|p{1.5cm}|p{2.0cm}|p{1.5cm}|p{1.5cm}|}
\hline
Datasets & Cluster Validation Indices & K-Means & hierarchical & PAM & fully ACA \\
\hline
\multirow{4}{4em}{Iris}
& Silhouette 	& 0.553 &0.354 &0.524 &0.523\\ 
& Dunn 	        & 0.717 &0.685 &0.190 &0.249\\ 
& DB Score   	& 0.662 &0.656 &0.786 &0.717\\ 
& Calinski H   	& 561.627& 556.841 & 520.313 & 367.739\\
\hline
\multirow{4}{4em}{Buddy Move} 
& Silhouette 	& 0.349 &0.207 &0.373 &0.168\\ 
& Dunn 	        & 0.054 &0.071 &0.047 &0.067\\ 
& DB Score   	& 1.039 &1.162 &1.102 &0.918\\ 
& Calinski H   	& 149.753& 103.809 & 144.449 & 101.991\\
\hline
\multirow{4}{4em}{Stoneflakes} 
& Silhouette 	& 0.487 &0.487 &0.3809 &0.410\\ 
& Dunn 	        & 0.169 &0.169 &0.0530 &0.040\\ 
& DB Score   	& 0.695 &0.635 &1.140  &1.350\\ 
& Calinski H   	& 97.069& 97.069 & 28.679 & 35.950\\
\hline
\multirow{4}{4em}{Seeds} 
& Silhouette 	& 0.492 &0.452 &0.492  &0.025\\ 
& Dunn 	        & 0.173 &0.154 &0.096  &0.067\\ 
& DB Score   	& 0.725 &0.756 &0.0722 &2.042\\ 
& Calinski H   	& 368.540& 303.465 & 369.546 & 450.121\\
\hline
\multirow{4}{4em}{Segmentation} 
& Silhouette 	& 0.511 &0.528 &0.582 &0.512\\ 
& Dunn 	        & 0.001 &0.002 &0.001 &0.001\\ 
& DB Score   	& 0.611 &0.550 &0.610 &0.604\\ 
& Calinski H   	& 3308.508& 2795.315 & 2932.353 & 1950.901\\
\hline
\multirow{4}{4em}{Engg Colleges} 
& Silhouette 	& 0.298 &0.301 &0.2621` &0.350\\ 
& Dunn 	        & 0.383 &0.343 &0.0461  &0.210\\ 
& DB Score   	& 1.060 &1.062 &1.150   &2.100\\ 
& Calinski H   	& 16.648& 16.648 & 12.729 & 12.651\\
\hline
\multirow{4}{4em}{User Knowledge} 
& Silhouette 	& 0.188 &0.158 &0.129 &0.436\\ 
& Dunn 	        & 0.742 &0.721 &0.689 &0.562\\ 
& DB Score   	& 1.595 &1.755 &2.230 &0.536\\ 
& Calinski H   	& 98.504& 82.057 & 93.853 & 65.556\\
\hline
\end{tabular}
\end{center}
\caption{Clustering for cluster size K = 3. Table represents the cluster validation indices and scores of different datasets used for clustering with different clustering algorithms.}
\label{Table6}
\end{table}
The performance of clustering algorithms varies significantly based on the dataset and the validation index used. Some general observations are:
\begin{itemize}
	\item[•] \textbf{KMeans:} Generally performed well across most indices but was outperformed by Hierarchical and PAM in some cases.
	\item[•] \textbf{Hierarchical:} Showed consistent performance, often leading in Silhouette and Dunn indices.
	\item[•] \textbf{PAM:} Performed well in Silhouette and DB scores, often outperforming other methods.
	\item[•] \textbf{Fully ACA:} Can be considered proposed method, with some exceptions in DB and Silhouette scores for specific datasets.

\end{itemize}
Selecting the appropriate clustering algorithm thus depends on the dataset characteristics and the priority of validation indices. Each algorithm has its strengths and may outperform others under different conditions.

\clearpage
\section{Conclusion}
\label{sec:S5}
We have studied and performed experiment on different types of real world datasets for clustering. We discussed the results and important observations in the section of results and discussion that we observed during the experiments. We found that the designed fully asynchronous CA based clustering algorithm performs well for clustering. It can be utilized for the clustering problems just like the other benchmark algorithms.
\chapter{Skew Asynchronous Cellular Automata }
\label{chap4}
\section{Introduction}
To break the assumption of global clock, cellular automata (CA) community has explored the notion of asynchronous cellular automata in the last two decades, see \cite{jca/Fates14}. Asynchronism is seen as an uncontrolled phenomenon where the cells are independent and are updated independently during the evolution of the system. To introduce this notion of independence (of cells), researchers have introduced different updating schemes \cite{BoureFC12,Bandini12,Alberto12,ROY2019600,CARONLORMIER2008522}, where the most studied schemes are fully asynchronous updating scheme \cite{FATES20061,Sethi2016} and $\alpha$-asynchronous updating scheme \cite{BoureFC12}. In fully asynchronous updating scheme \cite{FATES20061,Sethi2016}, a random cell is selected at each time step to update following uniform distribution. On the other hand, each cell has a given probability $\alpha$ to update and a probability $1 - \alpha$ not to update in $\alpha$-asynchronous updating scheme \cite{BoureFC12}. 

In the literature, some researchers have considered fully asynchronous updating scheme as `just' a sequential updating scheme \cite{Bandini12}. On the contrary, Fat\`{e}s \cite{jca/Fates14} has argued that fully asynchronous cellular automata (ACA) is the most `natural' updating scheme following the continuous nature of `real' time. In the same direction, a strong counter argument provides the justification that fully asynchronous updating scheme is a special case of {\em atomicity property} which is capable to model concurrent and distributed systems \cite{roydistributed}. In fact, in the early work, Zielonka et al. \cite{robco} have initially introduced ACA for modelling distributed systems using the atomicity property. 

Following the atomicity property \cite{roydistributed,Anindita12}, during the update of a cell's state, we say that the cell is enabled for update. While enabled, a cell first reads the states of its self and neighbours following the neighbourhood dependency, and then acts following the state transition function to update its state. This entire operation, i.e. reading of self and neighbour's states and update of the cell's own state, is considered as {\em atomic}. Therefore, during the enablement of a cell, its neighbours can not update their states. In other words, no two neighbouring cells can be enabled simultaneously. We call this cell update property as {\em atomicity property}. However, more than one cell at the same time step can be enabled following the atomicity property. For instance, at most half of the cells can be enabled following three-neighbourhood dependency in one-dimensional system. 

However, this assumption of atomicity property is not always applicable if we see cellular automata as a model of societal phenomenon. For example, in society, some neighbouring populations follow the same societal norms which violate the independence introduced by atomicity property. Moreover, for modelling natural systems, atomicity property is not a `good' choice as the perturbation (or noise) is applied to all cells of a small part of the system. With this motivation, this study observes the effect of the cellular system after breaking the atomicity property. In the first step towards this direction, we introduce the notion of {\em skew}-asynchronous updating scheme, where two neighbouring cells are allowed to be enabled simultaneously. Specifically, for a one-dimensional system, at each time step we randomly and uniformly select one cell and update the state of that selected cell and the state of it's right neighbouring cell. That is, two neighbouring cells are bound to update together (in other words, they are bound to follow the same societal norms), however, the notion of independence is still there in the overall system.

Particularly, in this paper, we observe the effect of breaking the atomicity property in elementary cellular automata (ECAs). Recently \cite{roy:hal-04456320}, we have completed the characterization of the dynamics, i.e. convergence, recurrence, and non-convergence non-recurrence, of ECAs following fully asynchronous updating scheme. In the light of this characterization, here, we classify $88$ minimal ECAs following qualitative and quantitative approaches under proposed skewed environment in comparison with fully asynchronous environment. This displays the overall picture after breaking atomicity property. As we will see, this  study is sufficiently rich to provide many kinds of worthy examples. Some ECAs show their absolute irritation in the absence of atomicity property, where some convergent ECAs under fully asynchronous scheme show opposite divergence dynamics in the absence of atomicity property, and vice versa. Moreover, Some ECAs show dependency on CA size for convergence following atomicity property \cite{Sethi901}, however, the same is not applicable for skewed environment. Obviously, some ECAs ignore the absence of atomicity property. Following this, in the second part of this study, we theorize the conditions for convergence towards all $0$ and all $1$ point attractors under proposed skewed environment. 

\section{Cellular automata, asynchronism and experimental setup}
\label{S1}
In this work, we consider elementary cellular automata (ECAs), i.e. one-dimensional three - neighbourhood (left, self, right) and two-state (\{$0,1$\}) cellular automata. Here, we consider periodic boundary condition, i.e. cells are arranged as a ring. The set of indices that represent each cell is denoted by $\mathcal{L} = \mathbb{Z}/n\mathbb{Z}$, here, the number of cell is $n$. The collection of all states at a given time is called a configuration. We denote $\varepsilon_n = $\{$0,1$\}$^\mathcal{L}$ as the set of configurations. Following the local transition function $f: \{0,1\}^3 \rightarrow \{0,1\}$, we define elementary cellular automata which indicates how a cell updates its state according to its own state and the state of its left and right neighbours. In Table~\ref{T1}, we express local transition function in a look-up table format where Rule Min Term (RMT) represents each argument by $r = 4 \times x + 2 \times y + z$, see Table~\ref{T1} (row $2$). Finally, the decimal equivalent of the eight outputs is called ``rule'', i.e. $f(1,1,1)\cdot 2^7 + f(1,1,0)\cdot 2^6 + \cdots + f(0,0,0)\cdot2^0$. We call an RMT active if it changes the state of a cell, i.e. $f(x,y,z) \neq y$, otherwise it is passive.  Note that, a configuration can also be written as a sequence of RMTs. Therefore, we have $2^8 = 256$ ECAs, out of which $88$ are minimal representative ECAs and the rests are their equivalent. Table~\ref{Table1} shows example ECA $134$ where RMT $6$ (resp. $7$) is active (resp. passive).

\begin{table}[!htbp]    
\centering
\scriptsize
\begin{tabular}{cccccccccc} \hline  
($x$,$y$,$z$)       &111&110&101&100&011&010&001&000& Rule \\    
(RMT)               &(7)&(6)&(5)&(4)&(3)&(2)&(1)&(0)& \\\hline  
$f$($x$,$y$,$z$)    & 1 & 0 & 0 & 0 & 0 & 1 & 1 & 0 & 134\\ \hline
\end{tabular}
\caption{Look-up table for ECA 134.\label{Table1}}  
\end{table}

Here, we update the system following fully asynchronous and proposed skew-asynchronous updating schemes. In fully asynchronous updating scheme, one cell is selected randomly and uniformly for update at each time step. Similarly, in the proposed {\em skew}-asynchronous scheme, one cell (say $i$) is selected randomly and uniformly, and the local rule is applied for update to cell $i$ and $i+1$ (right neighbouring cell of the selected cell) at each time step. Let ($U_t$)$_{t \in \mathbb{N}} \in \mathcal{L}^{\mathbb{N}}$ denotes the random sequence of selected cells for update where $\mathbb{N} = \{0,1,\cdots\}$ is the set of natural numbers. Evolution of the ECA under these asynchronous updating scheme from an initial configuration $x$ is represented by the stochastic process ($x^t$)$_{t \in \mathbb{N}}$, and defined recursively by: $x^0 = x$ and $x^{t+1} = F(x^t,U_t)$ with

(For fully asynchronous): $ x^{t+1}_i =
\begin{cases}
f (x_{i-1}^t,x_i^t, x_{i+1}^t)&  \text{ if }  i = U_t \\
x^t_i  & \text{ otherwise.}
\end{cases}$

(For skew-asynchronous): $ x^{t+1}_i =
\begin{cases}
f (x_{i-1}^t,x_i^t, x_{i+1}^t)&  \text{ if }  i = U_t \text{ or }  i = U_t + 1 \\
x^t_i  & \text{ otherwise.}
\end{cases}$

Next, a configuration $x \in \varepsilon_n$ is called a {\em point attractor} if we have $F(x,u) = x$ for all $u \in \mathcal{L}$, that is, when all the RMTs of $x$ are passive. We have also introduced the recurrent property of (fully) ACA by introducing the notion of {\em recurrent} and {\em transient} configuration, see \cite{ref_r4,ref_r16}. Here, a configuration $x \in \varepsilon_n$ is recurrent if for every configuration $y$ that is reachable (following a sequence of successor relations) from $x$, $x$ is also reachable from $y$; a non-recurrent configuration is transient (i.e. after some finite time, it is not possible to return back again).  To sum up. following three are the possible dynamics of the system: (i) For convergent system, starting from any initial configuration, the system converges to point attractor; (ii) For recurrent system, the system shows kind of reversible `eternal return' phenomenon where all configurations are recurrent; and (iii) For non-convergent non-recurrent system, it shows the dynamics of convergence towards {\em multi-length} attractor, where few configurations are recurrent and rest are transient.

Following this, in this first experiment, we follow the well-established \cite{ROY2019600} qualitative and quantitative experimental approaches: (a) Firstly, we need to observe the evolution of the system through space-time diagrams which can be able to provide an important qualitative visual comparison. We consider $n \in$ [$90,100$] and evolve the system for $1000$ time steps; and (b) Secondly, we need to calculate the density of a configuration $x$ which can be written as $d_x = x_1/n$ ($x_1$ is the number of $1$s in configuration $x$ and $n$ is the lattice size). We start with $n = 100$ considering initial density $d_{ini} = 0.5$; and evolve the system again for $1000$ steps; and calculate the density of the configuration in every step. The change in density during evolution of the system provides the formal quantitative comparison of two asynchronous updating schemes. 

\section{Dynamics of skew-asynchronous cellular automata}

This section classifies the dynamics of ECA under skew-asynchronous scheme in comparison with fully asynchronous dynamics. In Table~\ref{Table2}, \verb|Fully|$_P$ denotes the dynamics (property) of these $88$ minimal ECAs under fully asynchronous updating scheme where convergence, recurrence and non-convergence non-recurrence dynamics respectively denoted as \verb|C|, \verb|R| and \verb|NC-NR|. We have identified $50$ convergent ECAs under fully asynchronous environement \cite{Sethi2016} where starting from any initial configuration cellular system converges to point attractor. Moreover, some ECAs \cite{Sethi901} show convergence dynamics depending on the lattice size ($n \in 2\mathbb{N}$ or $n \in 3\mathbb{N}$) under atomicity property. In the (opposite) divergence dynamics, we have identified $18$ recurrent ECAs \cite{ref_r4,ref_r16} following atomicity property which can be able to capture reversible \cite{ref_r4,ref_r16} `eternal return' phenomenon \footnote{ECA $204$ belongs to both convergenct and recurrent class.}. We have marked rest divergent fully asynchronous ECAs as non-convergent non-recurrent where some configurations are recurrent and some are transient (in a different view, they depict convergence towards multi-length attractor), see \cite{roy:hal-04456320}.

\begin{table}[h]  
\centering
\scriptsize
\begin{tabular}{lll|lll|lll} \hline  
ECA & Fully$_{P}$ & Skew$_{C}$ & ECA & Fully$_{P}$ & Skew$_{C}$ & ECA & Fully$_{P}$ & Skew$_{C}$ \\ \hline 
0 & C & \checkmark & 1 & NC-NR & \checkmark & 2 & C & \checkmark \\
3 & NC-NR & \checkmark & 4 & C & \checkmark & 5 & C & \checkmark \\
6 ($n \in 2\mathbb{N}$) &  C & \checkmark & \textbf{6} ($n \in 2\mathbb{N}+1$) & NC-NR & $\times$ & 7 ($n \in 2\mathbb{N}$) & C & \checkmark \\
7 ($n \in 2\mathbb{N}+1$) &  NC-NR & \checkmark & 8  & C & \checkmark  & 9  & NC-NR & \checkmark \\
10 & C & \checkmark & 11 & NC-NR & \checkmark & 12 & C & \checkmark \\
13 & C & \checkmark & 14 ($n \in 2\mathbb{N}$) & C & \checkmark & 14 ($n \in 2\mathbb{N}+1$) & NC-NR & \checkmark \\
15 ($n \in 2\mathbb{N}$) & C & \checkmark & 15 ($n \in 2\mathbb{N}+1$) & NC-NR & \checkmark & 18 & C & \checkmark \\
19 & NC-NR & \checkmark & 22 ($n \in 2\mathbb{N}$) & C & \checkmark & \textbf{22} ($n \in 2\mathbb{N}+1$) &  NC-NR   & $\times$ \\
23 ($n \in 2\mathbb{N}$) & C & \checkmark & 23 ($n \in 2\mathbb{N}+1$) & NC-NR & \checkmark & 24  &  C  & \checkmark \\
25 & NC-NR & \checkmark & \textbf{26} & C & $\times$ & 27 & NC-NR & \checkmark \\
28 & NC-NR & \checkmark & 29 & NC-NR & \checkmark & 30 ($n \in 2\mathbb{N}$) & C & \checkmark \\
30 ($n \in 2\mathbb{N}+1$) &  NC-NR & \checkmark & 32  & C & \checkmark  & 33  & R & \checkmark \\
34 &  C & \checkmark & 35  & R & \checkmark  & 36  & C & \checkmark \\
37 ($n \in 3\mathbb{N}$) &  C & \checkmark & 37 ($n \notin 3\mathbb{N}$)  & NC-NR & \checkmark  & \textbf{38}  & R & $\times$ \\
40 &  C & \checkmark & 41  & R & \checkmark  & 42  & C & \checkmark \\
43 &  R & \checkmark & 44  & C & \checkmark  & 45 ($n \in 3\mathbb{N}$) & C & \checkmark \\
45 ($n \notin 3\mathbb{N}$)  & NC-NR  & \checkmark & 46  & R & \checkmark  & 50  & C & \checkmark \\
51 &  R & \checkmark & \textbf{54}  & R & $\times$   & 56  & C & \checkmark \\
57 &  R & \checkmark & \textbf{58}  & C & $\times$  & 60  & R & \checkmark \\
62 &  R & \checkmark & 72  & C & \checkmark  & 73  & NC-NR & \checkmark \\
74 &  C & \checkmark & 76  & C & \checkmark  & 77  & C & \checkmark \\
78 &  C & \checkmark & \textbf{90}  & C & $\times$ & 94  & C & \checkmark \\
104 &  C & \checkmark & 105 ($n \notin 4\mathbb{N}$)  & R & \checkmark & \textbf{105} ($n \in 4\mathbb{N}$)  & R & $\times$ \\
106  & C & \checkmark & 108 &  R & \checkmark & \textbf{122}  & C & $\times$ \\
128  & C & \checkmark & 129 &  NC-NR & \checkmark & 130  & C & \checkmark  \\
132  & C & \checkmark & \textbf{134} &  R & $\times$ & 136  & C & \checkmark  \\
137  & NC-NR & \checkmark & 138 &  C & \checkmark & 140  & C & \checkmark  \\
142  & R & \checkmark & 146 &  C & \checkmark & \textbf{150}  & R & $\times$  \\
152  & C & \checkmark & 154 &  C & \checkmark & 156  & R & \checkmark  \\
160  & C & \checkmark & 162 &  C & \checkmark & 164  & C & \checkmark  \\
168  & C & \checkmark & 170 &  C & \checkmark & 172 & C & \checkmark  \\
178  & C & \checkmark & 184 &  C & \checkmark & 200 & C & \checkmark  \\
204  & C,R & \checkmark & 232 &  C & \checkmark & & &   \\
\hline
\hline
\end{tabular} 
\caption{$88$ minimal ECAs with their dynamics under fully and skew-asynchronous environment. Here, C,~R,~ NC-NR respectively notes convergence, recurrence and non-convergence non-recurrence dynamics.\label{Table2}}
\end{table}

Next, we compare the dynamics of these $88$ minimal ECAs following skewed environment. Following are the most remarkable observations of this study:

\begin{itemize}
\item[(a)] ECAs \textbf{26},~\textbf{58},~\textbf{90} and \textbf{122} show convergence dynamics (mostly towards all $0$ point attractor) under fully asynchronous updating scheme. However, these ECAs depict divergence dynamics in the absence of atomicity property.

\item[(b)] In an opposite direction, ECAs \textbf{38},~\textbf{54},~\textbf{134}, and \textbf{150} depict recurrence (reversible `eternal return' phenomenon) dynamics following atomicity property. Surprisingly, these ECAs show convergence (towards all $0$ point attractor) under skewed environment.

\item[(c)] ECAs \textbf{6} and \textbf{22} show most interesting dependency on lattice size following atomicity property. The system converges to point attractor for  $n \in 2\mathbb{N}$, and, show multi length attractor of length $2n$ (i.e. \verb|NC-NR|) for $n \notin 2\mathbb{N}$ (see \cite{roy:hal-04456320}). However, under skewed environment, the cellular system shows convergence for any $n \in \mathbb{N}$. 

\begin{figure}[h!tbp]
\begin{center}
\begin{tabular}{ccccc}
ECA & Synchronous & Fully ACA & Skew-ACA &  \\ 
26 & \includegraphics[width=24mm]{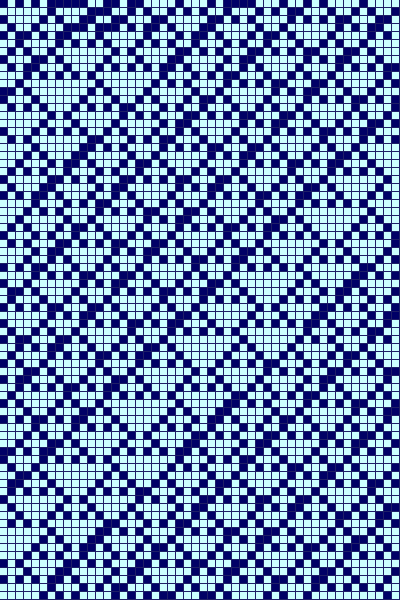} &  \includegraphics[width=24mm]{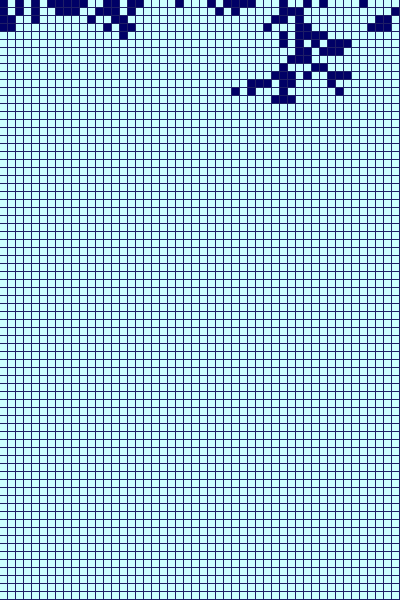} & \includegraphics[width=24mm]{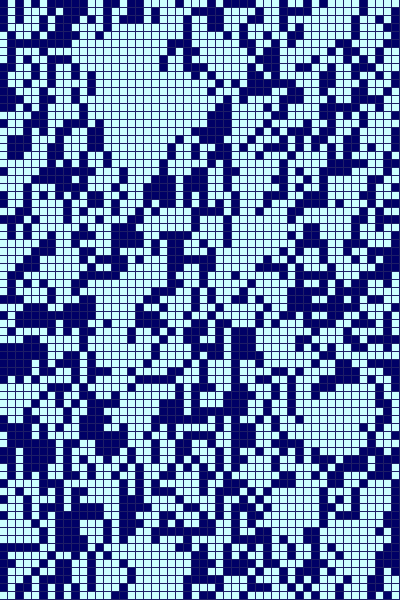} &  \includegraphics[width=52mm]{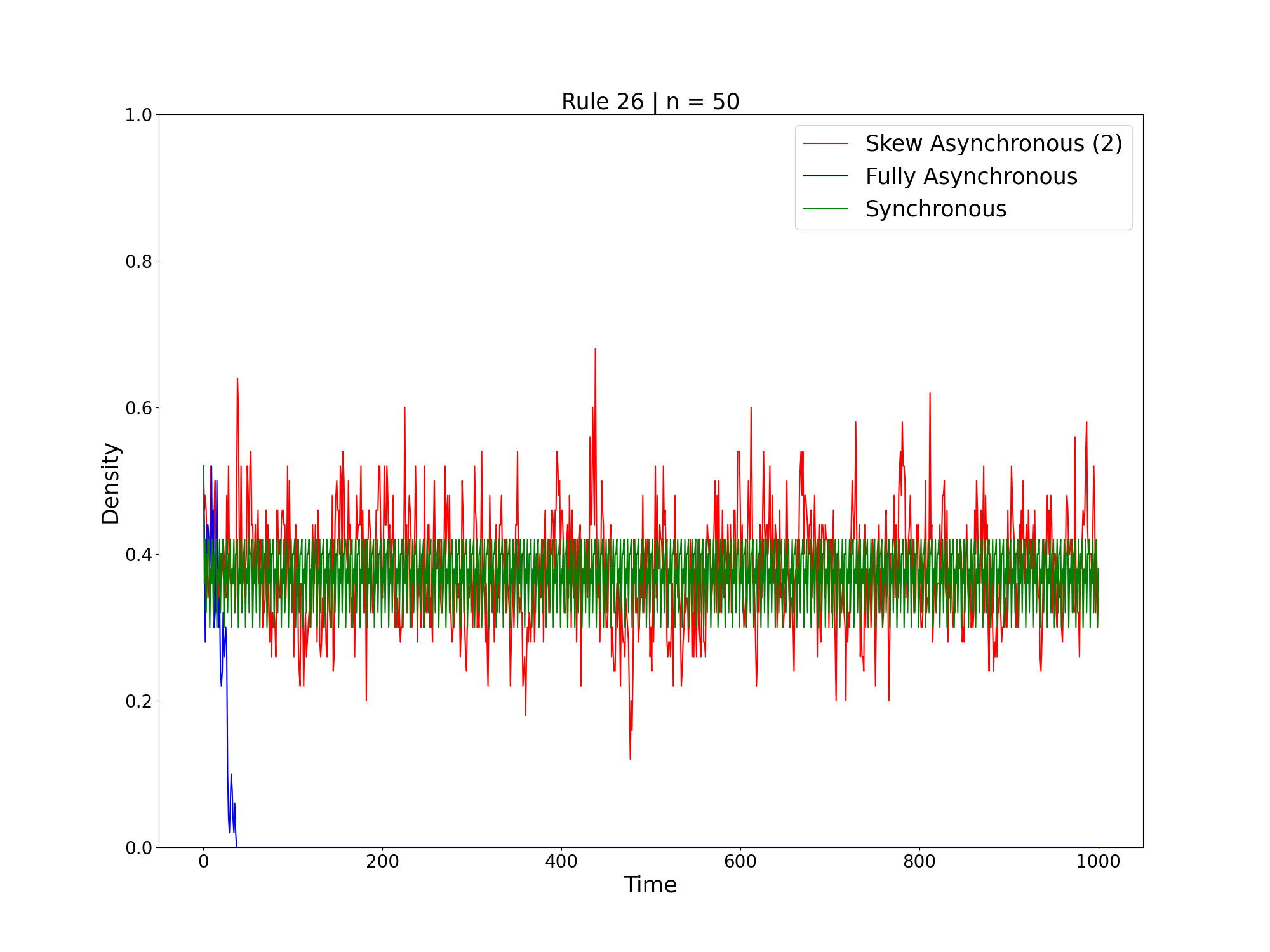} \\
50 & \includegraphics[width=24mm]{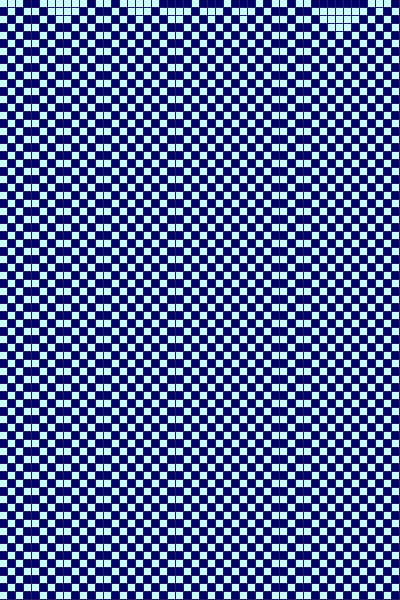} &  \includegraphics[width=24mm]{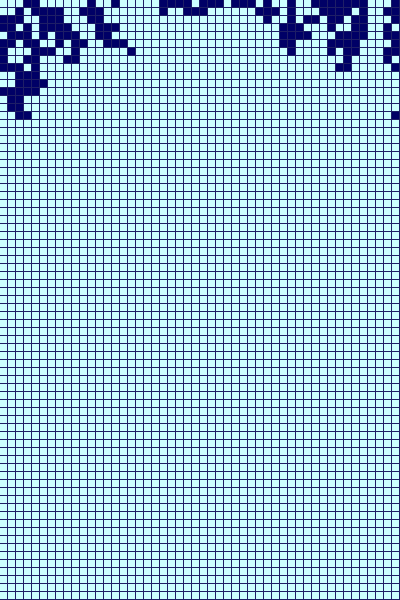} & \includegraphics[width=24mm]{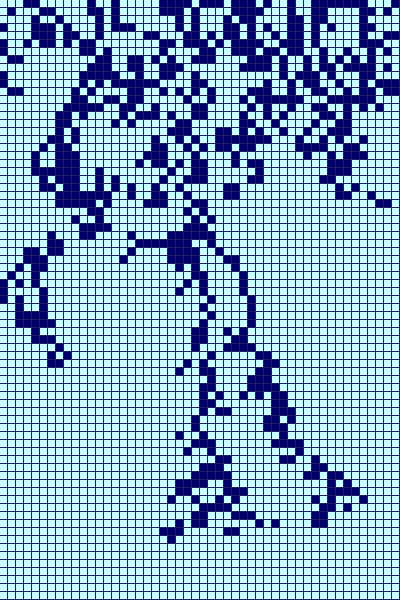} &  \includegraphics[width=52mm]{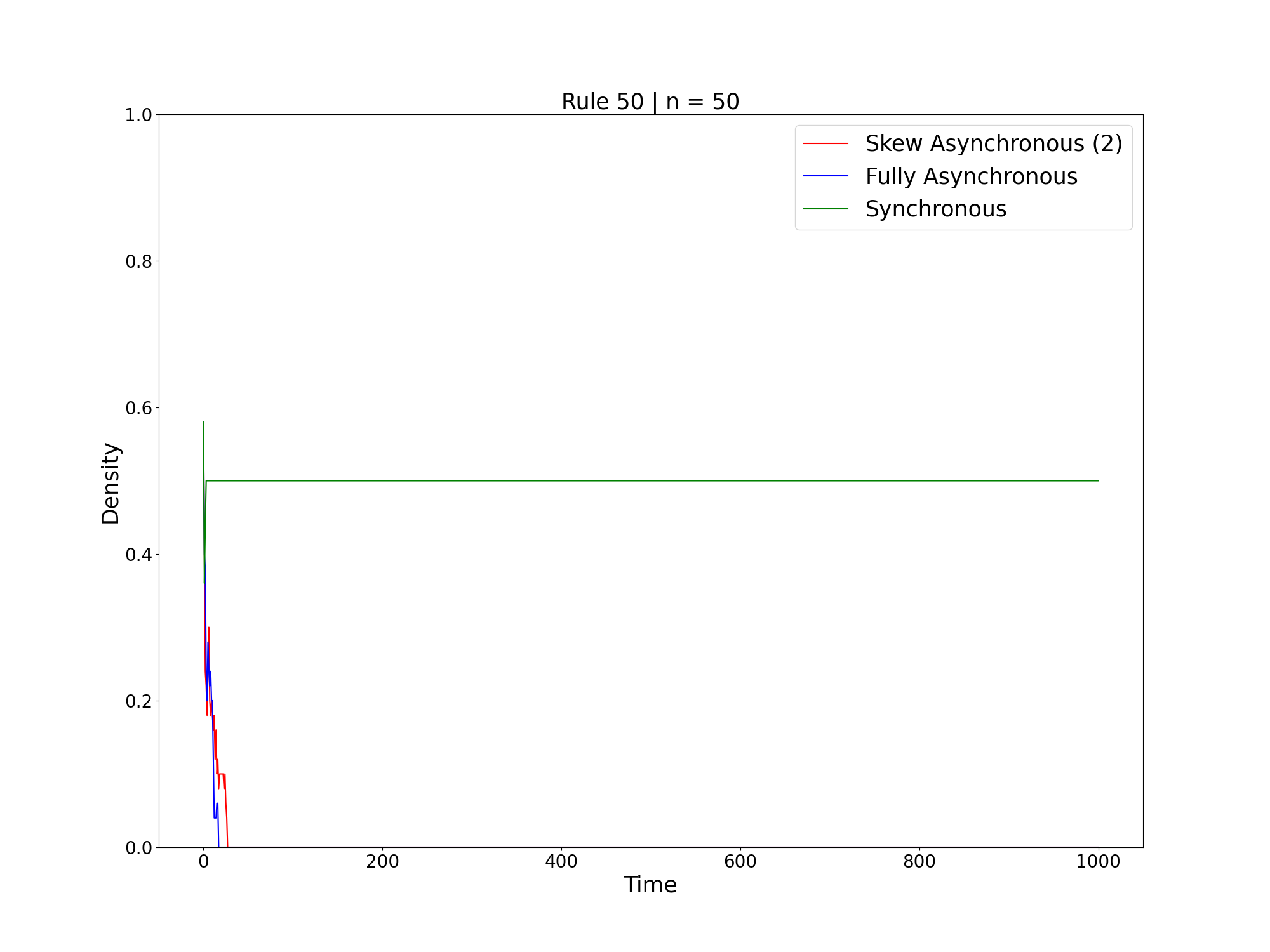} \\
58 & \includegraphics[width=24mm]{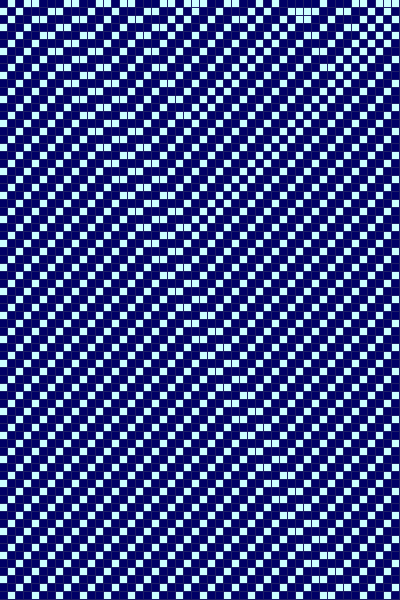} &  \includegraphics[width=24mm]{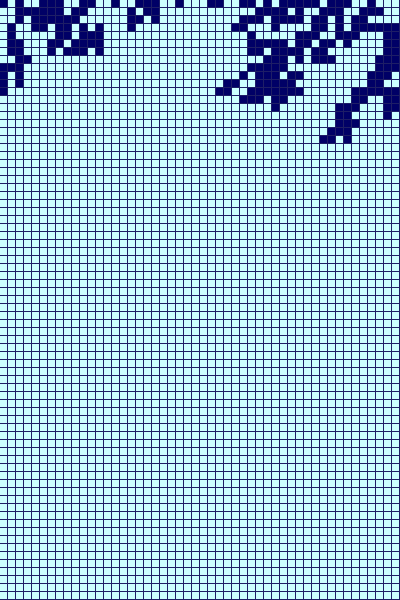} & \includegraphics[width=24mm]{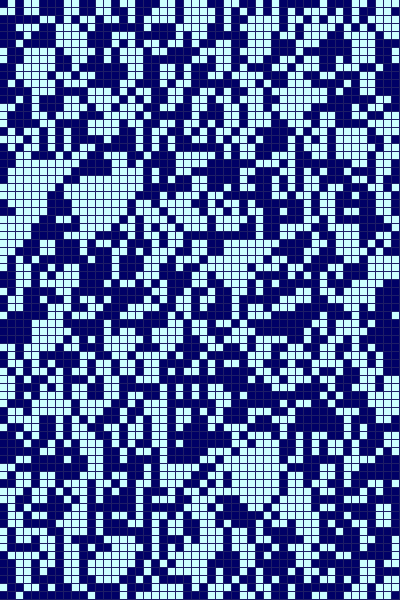} &  \includegraphics[width=52mm]{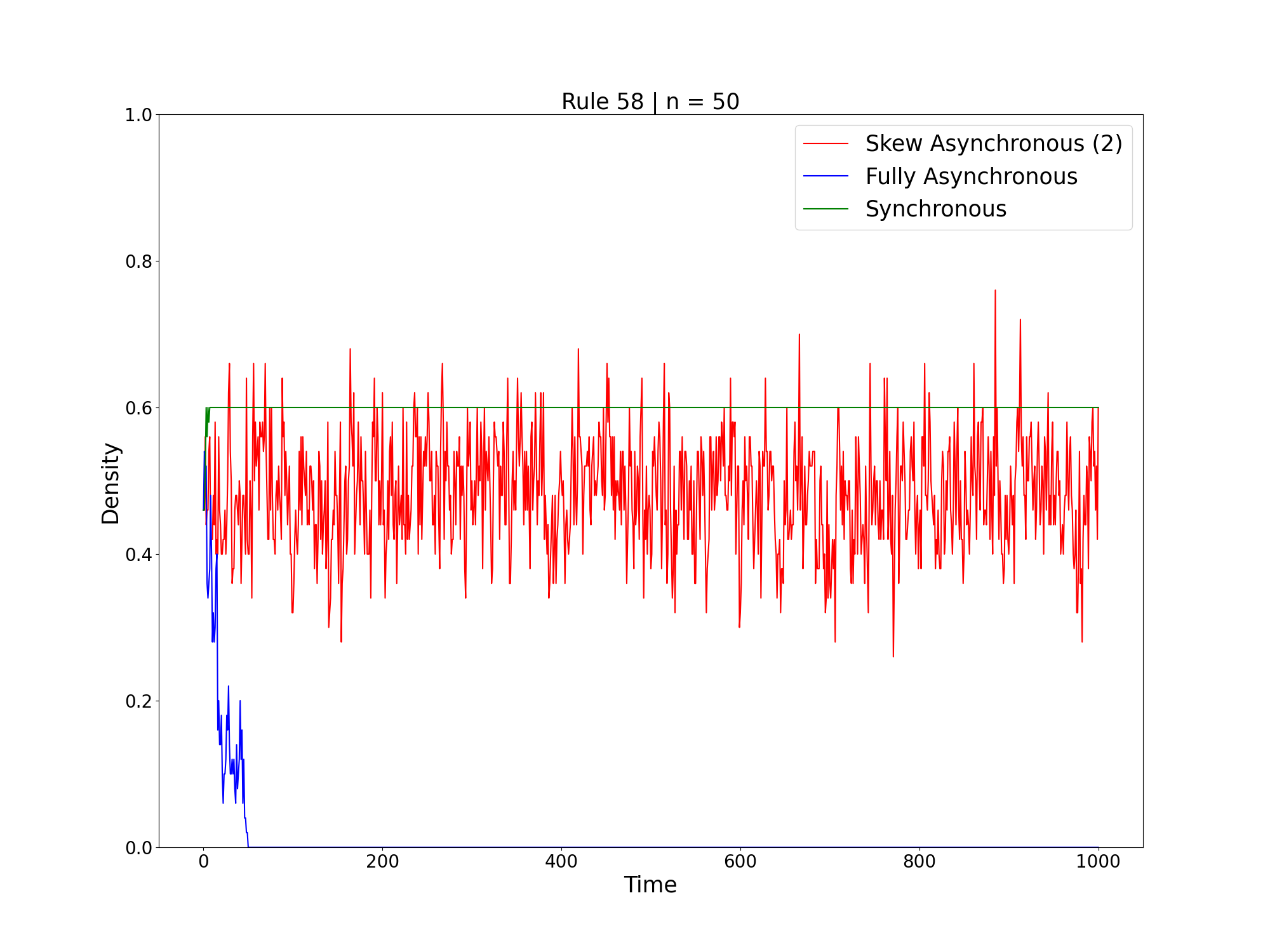} \\
\end{tabular}
\caption{\small{Dynamics of ECAs $26$,~$50$ and $58$ for - synchronous, fully asynchronous and skew-asynchronous systems. This figure also depicts the density-time plot. Here, $n = 50$ where the time goes form top to bottom. Each line represent the configuration of automata after $1,~n,~n/2$ updates respectively for changing updating schemes. Blue and white squares represent cells in state $1$ and $0$. This convention is kept in the rest of the text.}}
\label{Fig1}
\end{center}
\end{figure}

\item[(d)] Lastly, ECA \textbf{105} shows strange behaviour which shows convergence dynamics under skewed environment for lattice size divisible by $4$ ($n \in 4\mathbb{N}$). However, for $n \notin 4\mathbb{N}$, ECA $105$ shows divergence under skewed environment. Note that, this rule depicts recurrence behaviour under fully asynchronous updating for any lattice size $n \in \mathbb{N}$.
\end{itemize}

\begin{figure}[h!tbp]
\begin{center}
\begin{tabular}{ccccc}
ECA & Synchronous & Fully ACA & Skew-ACA &  \\ 
38 & \includegraphics[width=24mm]{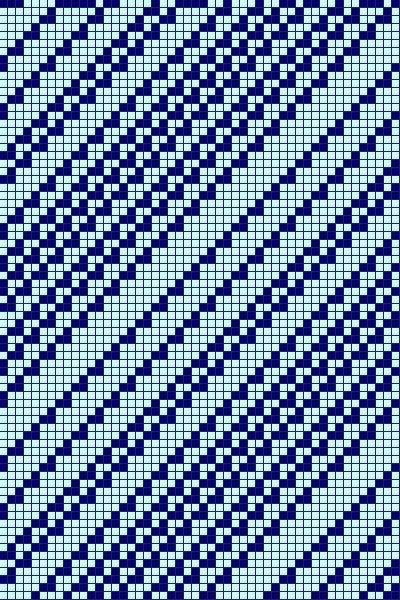} &  \includegraphics[width=24mm]{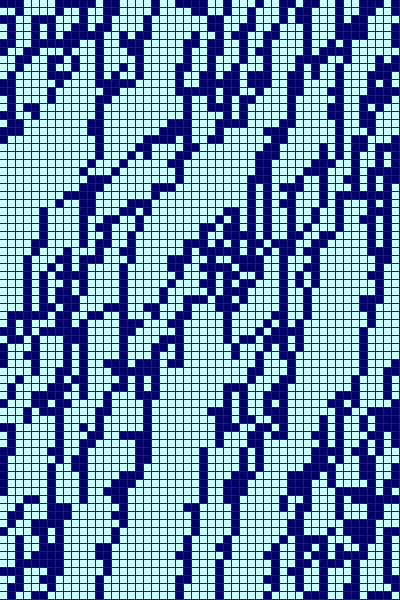} & \includegraphics[width=24mm]{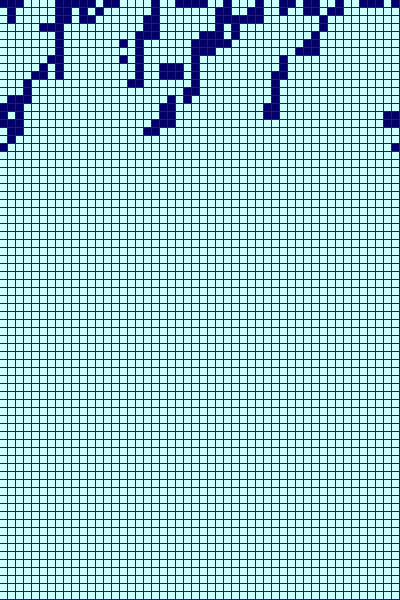} &  \includegraphics[width=52mm]{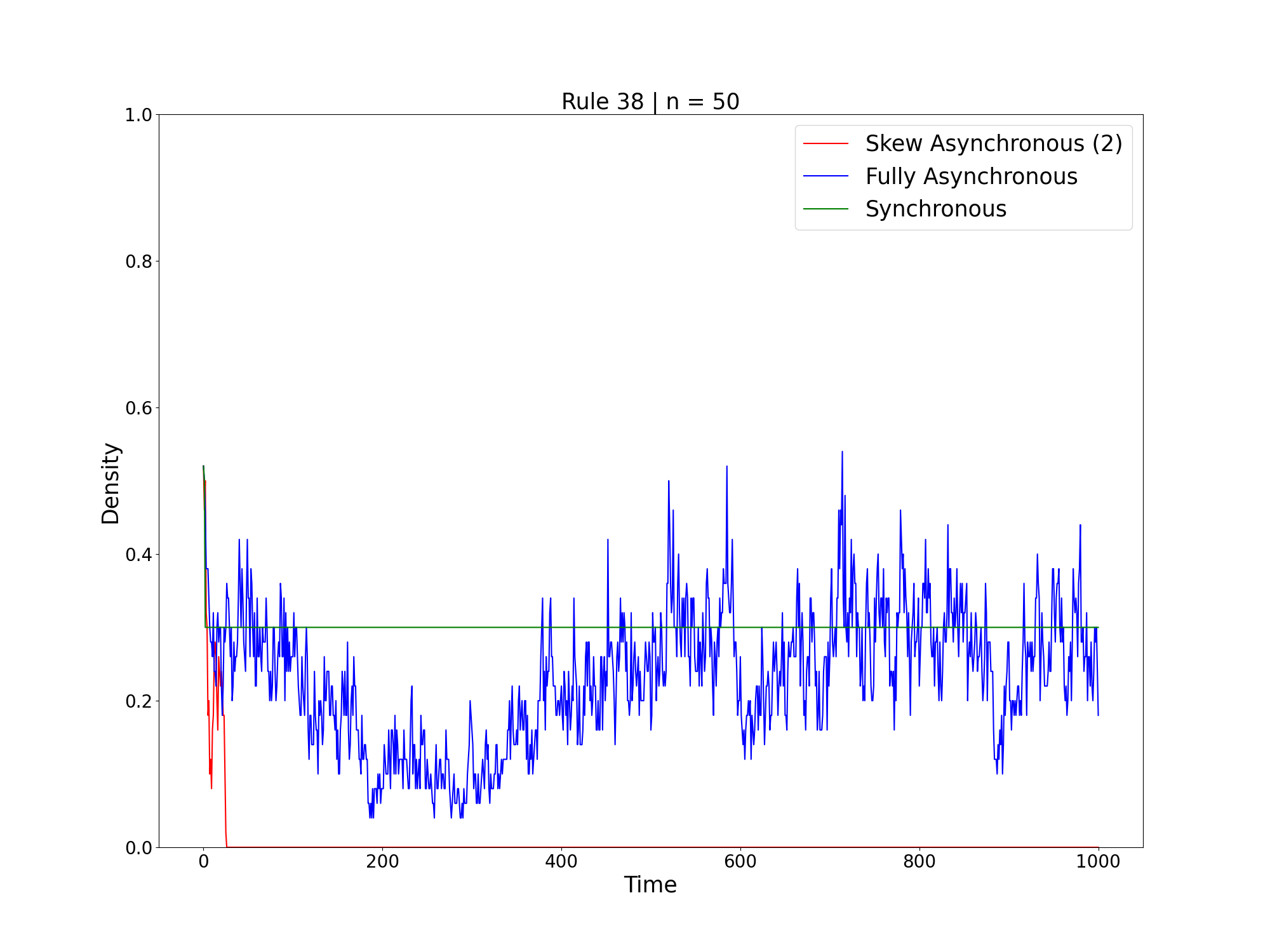} \\
54 & \includegraphics[width=24mm]{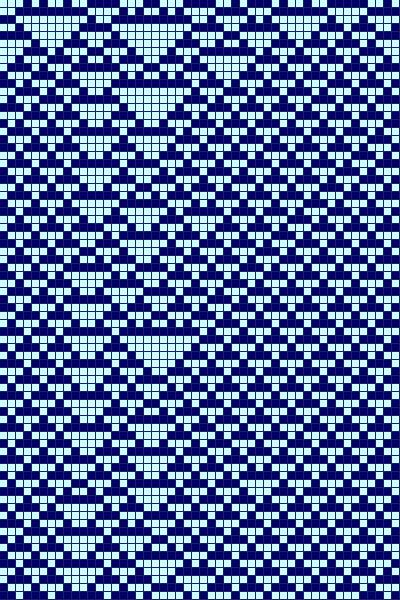} &  \includegraphics[width=24mm]{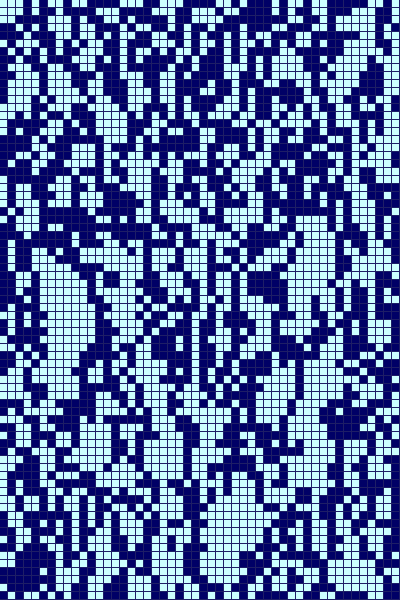} & \includegraphics[width=24mm]{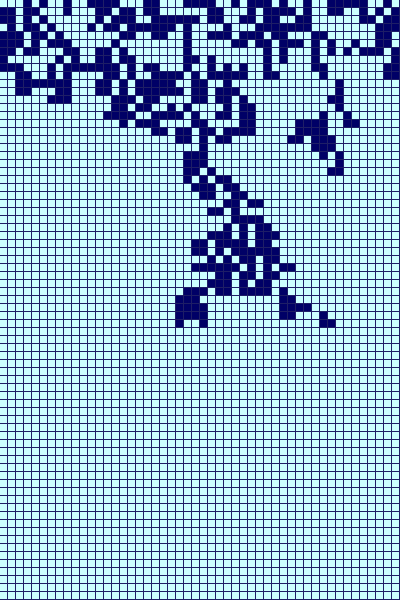} &  \includegraphics[width=52mm]{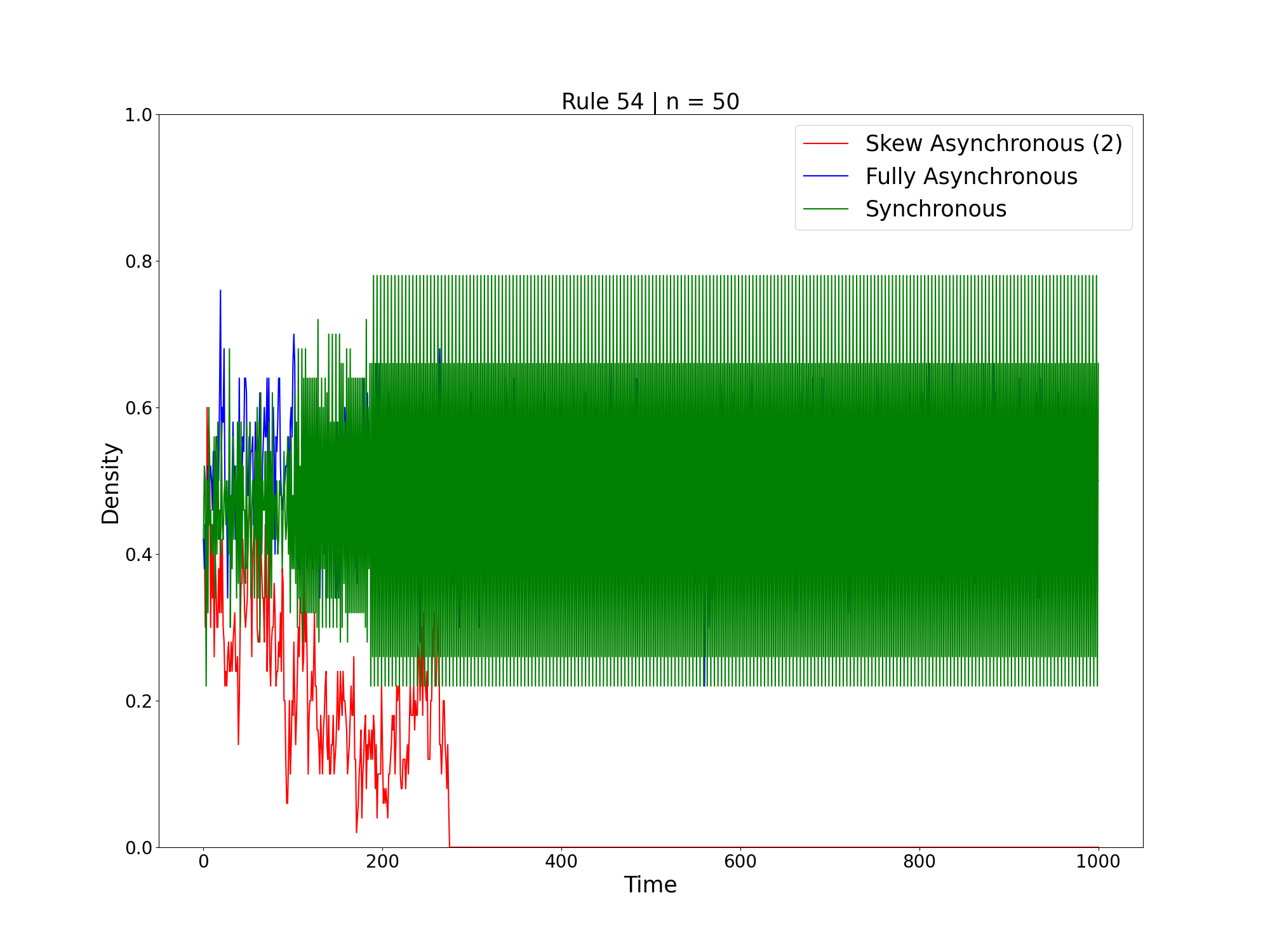} \\
62 & \includegraphics[width=24mm]{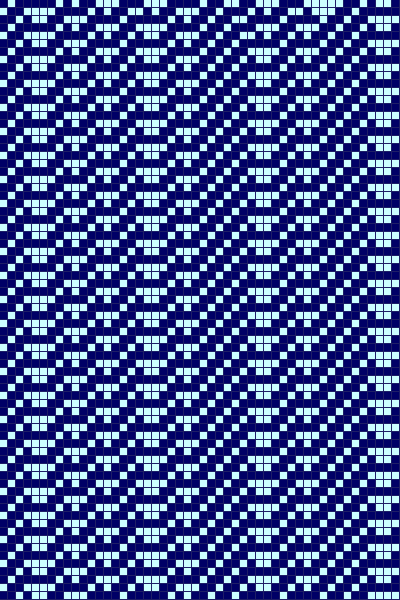} &  \includegraphics[width=24mm]{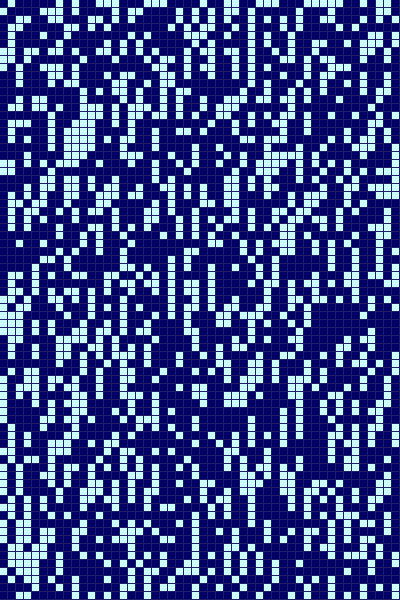} & \includegraphics[width=24mm]{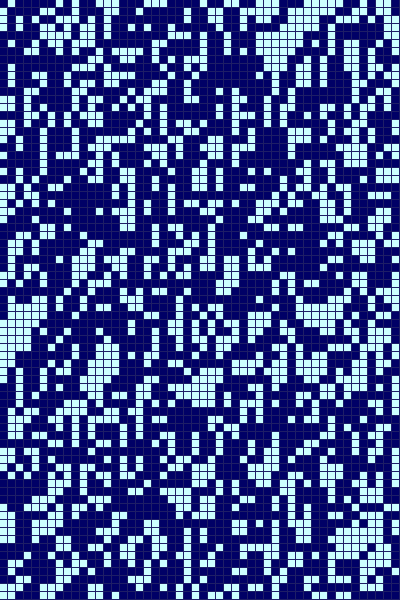} &  \includegraphics[width=52mm]{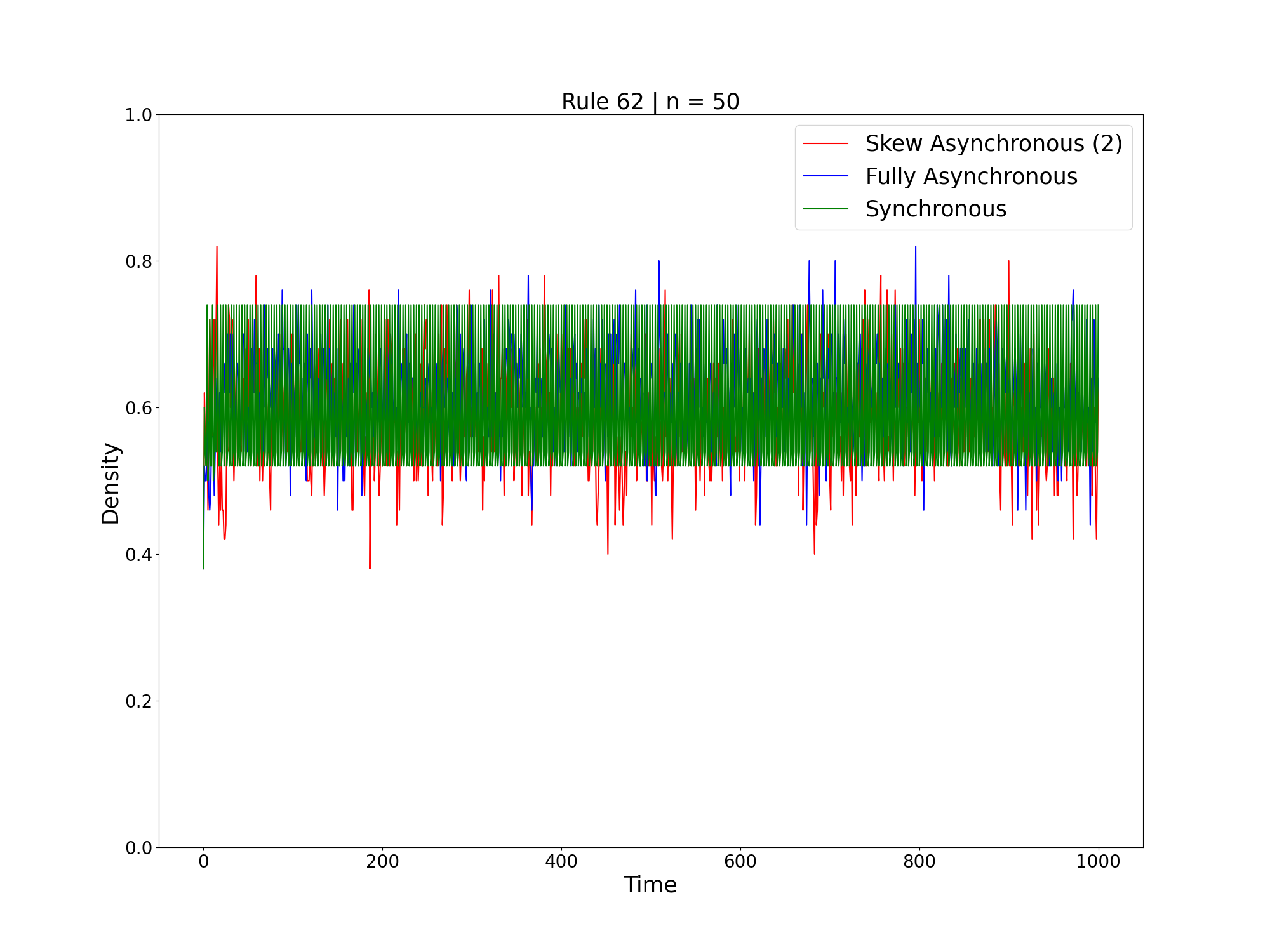} \\
\end{tabular}
\caption{Dynamics of ECAs $38$,~$54$ and $62$ for - synchronous, fully asynchronous and skew-asynchronous systems. This figure also depicts the density-time plot.}
\label{Fig2}
\end{center}
\end{figure}

To sum up, in Table~\ref{Table2}, \verb|Skew|$_C$ denotes the class of these $88$ ECAs under skewed environment where `$\checkmark$' depicts (overall) similar dynamics under both of the updating schemes, i.e. no effect of breaking atomicity property, and `$\times$' shows drastic difference in the dynamics of fully and skew-asynchronous systems. For evidence, Fig.~\ref{Fig1} depicts the dynamics of ECAs $26$ and $58$ where the system shows convergence towards all $0$ for fully asynchronous environment, and depicts divergence for skewed environment (\verb|Skew|$_C$ = $\times$). However, in Fig.~\ref{Fig1}, ECA $50$ shows same convergence (all $0$) dynamics for both of the updating schemes (\verb|Skew|$_C$ = $\checkmark$). Similarly, Fig.~\ref{Fig2} shows the dynamics of ECAs $38$ and $54$ which show divergence (specifically, recurrence) dynamics following atomicity property. On a contrary, these ECAs depict convergence towards all $0$ following skewed environment, i.e. \verb|Skew|$_C$ = $\times$. However, ECA $62$ shows divergence dynamics for both of the asynchronous updating scheme in Fig.~\ref{Fig2}; here, \verb|Skew|$_C$ = $\checkmark$. For quantitative evidence, see the density-time plot in Fig.~\ref{Fig1} and \ref{Fig2}. Fig.~\ref{Fig1} and \ref{Fig2} also depict the space-time diagrams under synchronous environment for reference. In this context, note that, Fat\`{e}s \cite{BoureFC12} has identified peculiar phase transition \footnote{There exist a critical value of $\alpha$ which distinguishes the behaviour of the system in two different phases - the system converges to all $0$ or all $1$ (passive phase); and the system oscillates around a fixed non-zero density (active phase).} behaviour for ECAs $26$,~$38$,~$58$,~$134$ (along with many others) for changing value of $\alpha$ following $\alpha$-asynchronous updating scheme. Here also, these ECAs ($26$,~$38$,~$58$,~$134$) show phase change (convergence to divergence) in the presence or absence of atomicity property. However, we are still open to identify atomicity property as (one of the) reason behind phase transition. 

\begin{figure}[h!tbp]
\begin{center}
\begin{tabular}{ccc}
ECA & Fully ACA & Skew-ACA \\  
6 ($n \in 2\mathbb{N}$) & \includegraphics[width=3.6cm]{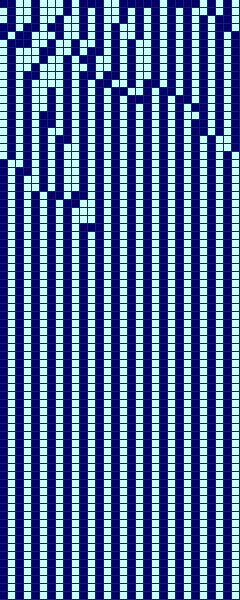} &  \includegraphics[width=3.6cm]{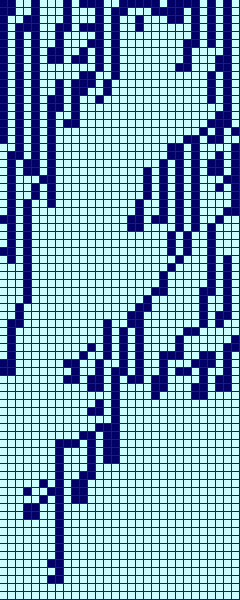} \\
 6 ($n \notin 2\mathbb{N}$) & \includegraphics[width=3.6cm]{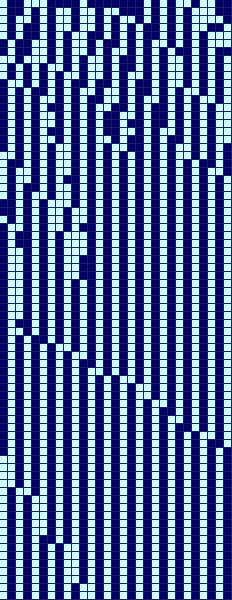} &  \includegraphics[width=3.6cm]{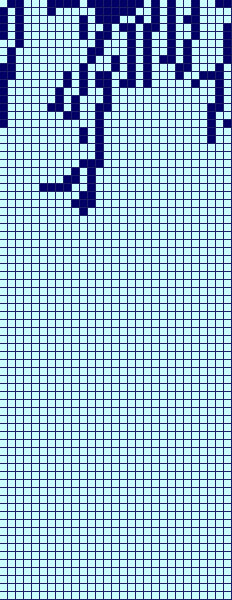}\\
 
\end{tabular}
\caption{Dynamics of ECAs $6$ under fully asynchronous and skew-asynchronous updating schemes for different lattice size.}
\label{Fig3}
\end{center}
\end{figure}

\begin{figure}[h!tbp]
\begin{center}
\begin{tabular}{ccc}
ECA & Fully ACA & Skew-ACA \\  
 22 ($n \in 2\mathbb{N}$) & \includegraphics[width=3.6cm]{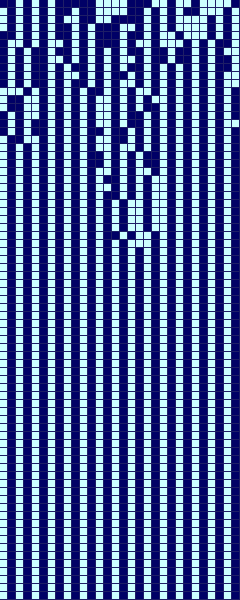} &  \includegraphics[width=3.6cm]{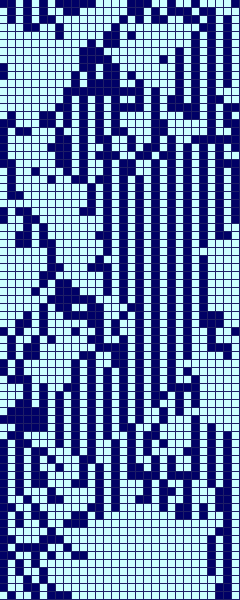} \\
 
 22 ($n \notin 2\mathbb{N}$) & \includegraphics[width=3.6cm]{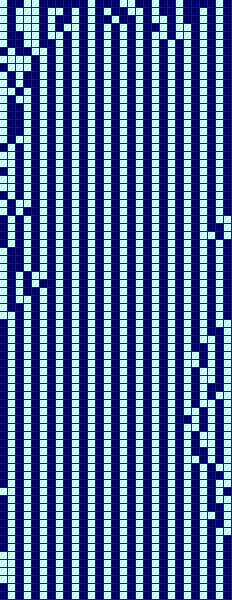} &  \includegraphics[width=3.6cm]{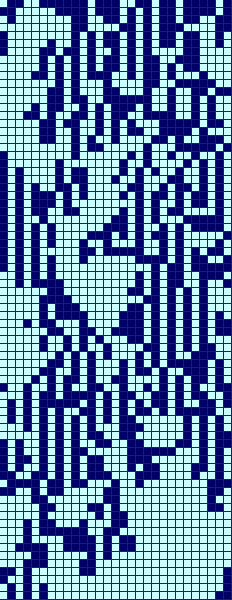}\\
 
\end{tabular}
\caption{Dynamics of ECAs $22$ under fully asynchronous and skew-asynchronous updating schemes for different lattice size.}
\label{Fig3}
\end{center}
\end{figure}

\begin{figure}[h!tbp]
\begin{center}
\begin{tabular}{ccc}
ECA & Fully ACA & Skew-ACA \\  
 105 ($n \in 4\mathbb{N}$) & \includegraphics[width=2.5cm]{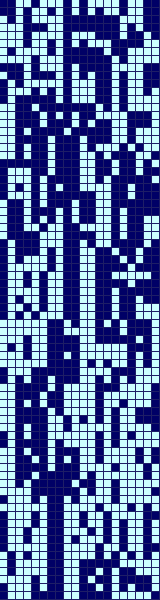} &  \includegraphics[width=2.5cm]{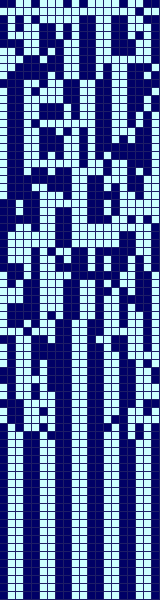}\\
 
 105 ($n \notin 4\mathbb{N}$) & \includegraphics[width=2.5cm]{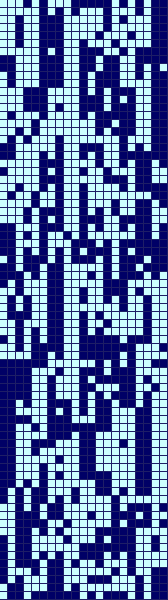} &  \includegraphics[width=2.5cm]{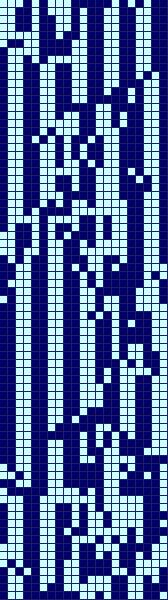}\\

\end{tabular}
\caption{Dynamics of ECAs $105$ under fully asynchronous and skew-asynchronous updating schemes for different lattice size.}
\label{Fig3}
\end{center}
\end{figure}

As a limitation, Table~\ref{Table2} only captures the extreme differences of these two asynchronous updating schemes, not microscopic details. However, the dynamics of many ECAs are not exactly same in microscopic view following these two updating schemes. For example, see ECAs $6$ and $22$, Fig~\ref{Fig3} shows their non-convergence non-recurrence dynamics ($n \notin 2\mathbb{N}$) for fully asynchronous update, however, these ECAs show convergence to all $0$ point attractor for skewed environment. Importantly, if we consider $n \in 2\mathbb{N}$, ECAs $6$ and $22$ show convergence for both of the updating schemes. However, there are differences: for fully asynchronous update, the only point attractor is ($001$)$^{n/3}$; for skewed environment, the system mostly converges to all $0$ point attractor, see Fig~\ref{Fig3}. Therefore, both updates show convergence ($n \in 2\mathbb{N}$), but there are differences following the microscopic view. Fig~\ref{Fig3} also shows ECA $105$ which converges to ($0011$)$^{n/4}$ for $n \in 4\mathbb{N}$ under skewed environment, however, depicts recurrence dynamics for $n \in 4\mathbb{N}$ under fully asynchronous update. However, we are still open about the status (recurrence or non-convergence non-recurrence) of ECA $105$ for $n \notin 4\mathbb{N}$ under skewed environment. To sum up, a proper theoretical claim about many microscopic facts are not possible with this experimental approach. However, the experimental results guide us towards the richness of this study. Following this, in the next section,  we theoretically explore part of the ECA rule space under skew-asynchronous updating scheme. 

\section{Convergence towards all 0 and all 1 point attractors under skewed environment}
It is quite obvious that an asynchronous CA is a Markov chain \cite{Sethi901}. In fact, we have also observed that convergent ACAs are absorbing Markov chain \cite{Sethi901}. Following this, we can write,

\begin{lemma}
\label{L1}
Convergent skew-ACAs are absorbing Markov chain. 
\end{lemma}

In this direction, we use the fact: in a absorbing Markov chain, the probability that the chain eventually enters to an absorbing state (and stays there forever) is $1$ \cite{Sethi901}. Next, we write the conditions for convergence under skewed environment. Here, in the following theorem, we denote two homogeneous configurations all $0$ and all $1$ as \textbf{0}, \textbf{1}, respectively. Moreover, we introduce the notion of $1$-region (resp. $0$-region) which corresponds to a maximal set of contiguous cells with state $1$ (resp. $0$) in a given configuration. 

\begin{theoremm}
\label{T1}
ECA $R$ converges to \textbf{0} (resp. \textbf{1}) point attractor under skew-asynchronous updating scheme if one of the following condition is satisfied:

\begin{itemize}
\item[1.] RMT $0$ (resp. RMT $7$) of $R$ is passive, RMT $2$ (resp. RMT $5$) of $R$ is active, and atleast one RMT from the pair of RMTs \{$1,4$\}  (resp. \{$3,6$\}) of $R$ is passive. 

\item[2.] RMT $0$ (resp. RMT $7$) of $R$ is passive, RMTs $3$ and $6$ (resp. RMTs $1$ and $4$) of $R$ are active. 
\end{itemize}
\end{theoremm}

\begin{proof}
\textbf{Case (1) Let us first consider that for a rule $R$, RMT $0$ is passive, RMT $2$ is active, and RMT $1$ or $4$ is passive:} Since, RMT $0$ of $R$ is passive, therefore \textbf{0} is clearly a point attractor for the CA. Next, we show that \textbf{0} is reachable from any configuration if it is the only point attractor. Otherwise, the CA reaches to other point attractor(s) for the configurations for which it does not reaches to \textbf{0}. 

Since skew-ACA forms an absorbing Markov chain (see Lemma~\ref{L1}), we need to show here that there exists an update pattern for which the CA converges to a point attractor. 

\underline{Let us first assume that RMT $7$ is active:} From any configuration $x$ different from \textbf{0}, we need to show disappearance of state $1$ to reach \textbf{0} configuration. Let us denote $l_1$ (resp. $l_0$) as the length of an arbitrary $1$-region (resp. $0$-region). 

If $l_1 > 3$, then we can directly apply active RMT $7$ for couple of two cells. Say, $i$ is the left most cell of $1$-region (note that, for $x =$ \textbf{1}, any cell can be able to play the role of cell $i$). Following this, $i+1, i+2$ cells move to state $0$. Then again, we apply active RMT $7$ in $i+4, i+5$ cells. Sequentially following this, cells $i+1, i+2, i+4, i+5, i+7, i+8, \cdots$ move to state $0$. On the other hand, cells $i, i+3, i+6, i+9, \cdots$ remain in state $1$. Following this, the original $1$-region of length $l_1$ is divided into many (say $m$) number of $1$-regions of length one. Here, RMTs $1,~2,~4$ are the member of RMT sequence for these (first $m$) $1$-regions of length one. For the $m+1^{th}$ $1$-region (rightmost) following are the possibilities:

\begin{itemize}
\item[] A: If $l_1 = 3m+1$, $m \in \mathbb{N}$, the rightmost $1$-region of size one. 

\item[] B: If $l_1 = 3m+2$, $m \in \mathbb{N}$, the rightmost $1$-region of size two. 

\item[] C: If $l_1 = 3m+3$, $m \in \mathbb{N}$, the rightmost $1$-region of size three. However, we can update pattern $11$ in this length three $1$-region. For active RMT $7$ and passive RMTs $3$ and $6$, it creates two $1$-regions of length one. For active RMT $7$ and active RMT $3$ or $6$, it creates single $1$-region of length one. Note that, the above argument is also applicable for independent $1$-regions of $l_1 = 3$. Fig.~\ref{Fig4} depicts this all possible situations.
\end{itemize}

\begin{figure}[h]
\begin{center}
\begin{tabular}{c}
\includegraphics[width=87mm]{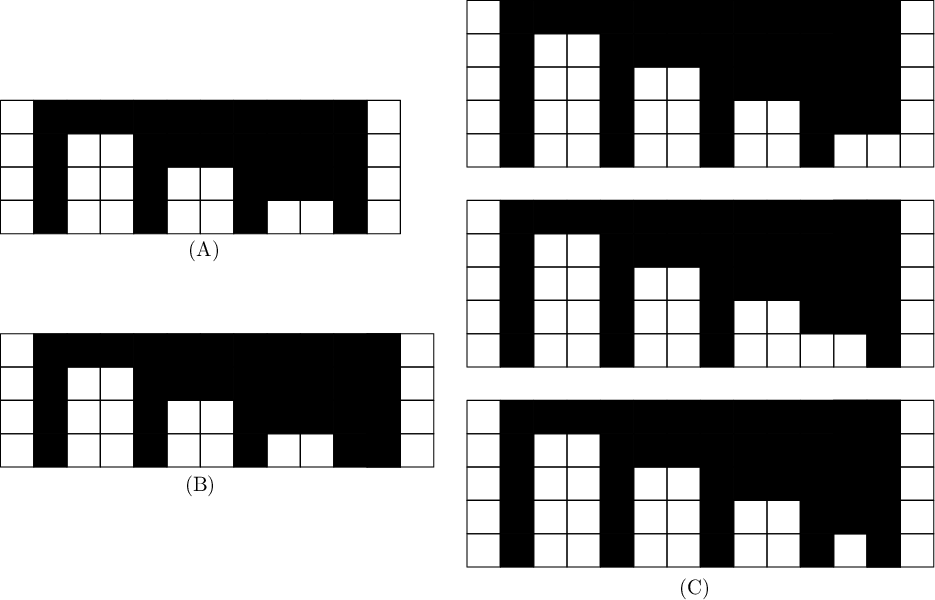}\\
\end{tabular}
\caption{Situations A, B, C where $1$-region of length $l_1 > 3$ is divided into many (say $m$) number of $1$-regions of length one. The $m+1^{th}$ $1$-region (rightmost) is of length one or two.}
\label{Fig4}
\end{center}
\end{figure}

Now, for the first $m$ $1$-regions with RMT sequence $1,~2,~4$, active RMT $2$ and passive RMT $1$ or $4$ applies, it moves the pattern $01/10$ to $00$. Hence, state $1$ disappears. 

For the rightmost ($m+1^{th}$) $1$-region of length one (situation A and C), either the RMT sequence is $1,~2,~4$ (right side $0$-region of length $l_0 > 1$) or $2,~5$ (right side $0$-region of length $l_0 = 1$). For RMT sequence $1,~2,~4$, state $1$ disappears directly following active RMT $2$ and passive RMT $1$ or $4$. For RMT sequence $2,~5$, if RMT $5$ is passive, we update pattern $10/01$, which vanishes state $1$, we can achieve the target. Otherwise, if RMT $5$ is active, we get $1$-region of length two (or more) by flipping the pattern $10$ (resp. $01$) to $01$ (resp. $10$). If again, the new $1$-region is with right side $0$-region of length $l_0 = 1$, we further follow the same, i.e. we increase (or decrease, if possible) the length of $1$-region, until and unless we get a right side $0$-region of length atleast two (which returns us back to the previous case). Note that, we definitely get a right side $0$-region of length atleast two following the early updates. During this process, if we again create $1$-region of length more than two, we return back to the initial problem. Note that, the above argument is also applicable for independent $1$-regions of $l_1 = 1$. The remaining situation, if the $1$-region of length two, the situation we discuss next.

For the rightmost ($m+1^{th}$) $1$-region of length two (situation B), either RMT sequence is $1,~3,~6,~4$ or $3,~6,~5$. For both of the situations, if RMT $3$ or $6$ or both are active, we update pattern $11$, either single or both $1$ disappear. If RMTs $1,~3,~6,~4$ all are passive, it is itself a point attractor. Similarly, if RMTs $3,~6,~5$ all are passive, it is again itself a point attractor. If the above situations are not true, then either RMT $1$ or $4$ is active. In that case, we increase the size of the $1$-region by updating pattern $01/10$ to $11$. Following this, we can be able to create $1$-region of size $l_1 = 3m+1$, $m \in \mathbb{N}$, which returns us back to the previous case. For RMT sequence $3,~6,~5$, the same argument of increasing length of $1$-region is applicable for active RMT $5$. Note that, the above argument is also applicable for independent $1$-regions of $l_1 = 2$.

Therefore, the above arguments include all possible situations for disappearance of $1$-region of length by from $l_1 \in \{1,\cdots,n\}$ considering active RMT $7$. If there are more such $1$-regions, one can use the above rationale for their disappearance which allows the CA to reach all \textbf{0}.

\underline{Next, we assume that RMT $7$ is passive:} Therefore, \textbf{1} configuration is clearly a point attractor. Now, for any configuration $x$ different from \textbf{0} and \textbf{1}, if $1$-region is of length $l_1 = 1$, it contains either RMT sequence $1,~2,~4$ or $2,~5$. For RMT sequence $1,~2,~4$, active RMT $2$ and passive $1$ or $4$ are applied, state $1$ disappears. For RMT sequence $2,~5$: for passive RMT $5$, we update pattern $10/01$, state $1$ disappears; for active RMT $5$, it flips the pattern $10/01$, and creates $1$-region of length $l_1 = 2$ or more. This situation, we discuss next. 

Next, if $x$ contains $1$-region of length $l_1 = 2$ or more, following are the situations: if one of the RMTs $3$ or $6$ is active, we update pattern $11$ (starting from left or right side of $1$-region), it makes sure that state $1$ disappears in each update, which returns us back to the previous case with RMTs $1,~2,~4$. If both RMTs $3$ and $6$ are passive, $x$ can not able to reach \textbf{0} configuration. If RMT $5$  and one of the RMTs $1$ or $4$ are active, we update pattern $10$ or $01$, state $0$ disappears, $x$ can be able to reach \textbf{1} configuration. On the other hand, if RMT $5$ is passive and one of the RMTs $1$ or $4$ is active, it reaches to point attractor where the RMT sequence is $5, 3, 7, 6$. Otherwise, the configuration itself a point attractor. 

Therefore, the above arguments include all possible situations for disappearance of $1$-region of length by from $l_1 \in \{1,\cdots,n-1\}$ considering passive RMT $7$. If there are more such $1$-regions, one can use the above rationale for their disappearance which allows the CA to reach all \textbf{0}.

While RMT $7$ is passive, RMT $5$ is active, and RMT $3$ or $6$ is passive, this property is the symmetric by the conjugation operation, i.e. exchange of $0$'s and $1$'s.

\textbf{Case (2) Next, let us consider for a rule $R$, RMT $0$ is passive, and RMTs $3$ and $6$ are active:} Since, RMT $0$ of $R$ is passive, therefore \textbf{0} is clearly a point attractor. Here, we follow the similar construction. 

\underline{Let us first assume RMT $7$ is active:} If length of $1$-region is even, we can be able to divide it into multiple $1$-region of size two using active RMTs $3,~6,~7$. If $l_1 = 2$, active RMTs $3$ and $6$ applies, both state $1$ disappear. If $l_1 > 2$ (even), cell $i, i+1, i+4, i+5, \cdots$ remain in state $1$, cell $i+2, i+3, i+6, i+7, \cdots$ move to state $0$ using active RMTs $6$ and $7$. If $l_1 > 1$ (odd), the region is divided into (say) $m+1$ number of $1$-regions, where the first $m$ $1$-region is of length two (return us back to the previous case), active RMTs $3$ and $6$ are applied, and the rightmost ($m+1^{th}$) $1$-region is of length one (will discuss later). The above argument is also applicable for \textbf{1} configuration where any cell can be able to play the role of cell $i$. 

\underline{Let us now assume RMT $7$ is passive:} If $l_1 > 1$, we can apply active RMT $3$ or $6$ and passive RMT $7$ by updating pattern $11$ (from left or right hand side of $1$-region), single state $1$ disappears in each step. Following this, we remain with $1$-region of length two, active RMTs $3$ and $6$ are applied, we can achieve the target. Here also, the remaining issue is with $1$-region of length one (will discuss next). Here, \textbf{1} configuration is a point attractor.

Next (\underline{independent of active/passive RMT $7$}), we discuss the remaining issue with $1$-region of length one. For $l_1 = 1$, the possible RMT sequences are $1,~2,~4$ and $2,~5$. For RMT sequence $1,~2,~4$: if RMT $2$ is active and RMT $1$ or $4$ is passive, we update pattern $10/01$, state $1$ disappears. For RMT $2$ passive and RMTs $1$ and $4$ passive, it is a point attractor. For RMT $2$ passive and RMT $1$ or $4$ active, we update pattern $10/01$ to get the pattern $11$, creates $1$-region of length $l_1 = 2$, returns back to the previous case. For RMT $2$ active and RMT $1$ or $4$ active, we update pattern $10/01$, the length of $1$-region remains unchanged, but it moves towards left or right direction. Now, if the new (after movement) $1$-region have a neighbouring $0$-region of length $l_0 > 2$, we update pattern $00$ using passive RMT $0$ and active RMT $1$ or $4$, creates $1$-region of length two, returns back to the previous case. Otherwise, if the new $1$-region have a neighbouring $0$-region of length $l_0 = 2$, we move the $1$-region to create neighbouring $0$-region of length $l_0 = 1$, RMT $5$ applies (discuss next).

For RMT sequence $2,~5$: if $2$ is active and $5$ is passive, state $1$ disappears. For $2$ and $5$ passive, it is a point attractor itself. For RMT $2$ active and RMT $5$ active, we update pattern $10/01$, it creates $1$-region of length two or more, returns back to the previous case. For $2$ passive and $5$ active, we update pattern $10/01$, it creates $1$-region of length three or more, returns back to the previous case. To conclude, the above arguments include all possible situations.

While RMT $7$ is passive, RMTs $1$ and $4$ are active, this property is the symmetric by the conjugation operation, i.e. exchange of $0$'s and $1$'s.
\end{proof}

Now, using Theorem~\ref{T1} following $34$, out of $88$, minimal ECAs converge to point attractor \textbf{0} and/or \textbf{1} for $n \in \mathbb{N}$:  0, 2, \textbf{6}, 8, 10, 18, \textbf{22}, 24, 32, 34, \textbf{38}, 40, 42, 50, \textbf{54}, 56, 74, 104, 106, 128, 130, \textbf{134}, 136, 138, 146, \textbf{150}, 152, 154, 160, 162, 168, 170, 178, 184. Note that, ECAs $38$,~$54$,~$134$ and $150$ (in bold) show recurrence under fully asynchronous updating scheme, however, these ECAs depict convergence to point attractor \textbf{0} following Theorem~\ref{T1} which validates our finite lattice size experimental findings. The same is true for ECAs $6$ and $22$ (in bold) considering odd lattice size. As an observation, ECAs $146$,~$150$,~$168$,~$170$,~$178$, and $184$ show convergence towards both of the point attractor \textbf{0} and \textbf{1}. Following this property of convergence towards both of the point attractors, these ECAs ($146$,~$150$,~$168$,~$170$,~$178$,~$184$) may establish themselves as a potential candidate for density classification problem \cite{fates13, fates:LIPIcs}, which is still open for us.

\section{Summary}
To sum up, we have explored the effect of breaking atomicity property in elementary cellular automata. Towards the first step in this direction, we have introduced the notion of skew-asynchronous updating scheme where two neighbouring cells are bound to update together. In the absence of atomicity property, ECAs have shown following rich variety of results: (i) A phase change (convergence $\rightarrow$ divergence) for ECAs $26$,~$58$,~$90$, and $122$; (ii) Another opposite phase change (divergence, i.e. recurrence $\rightarrow$ convergence) for ECAs $38$,~$54$,~$134$, and $150$; (iii) On the other hand, some remarkable behaviour emerged; phase change (divergence $\rightarrow$ convergence) depending on the lattice size ($n \in 2\mathbb{N}$) for ECAs $6$ and $22$; and phase change (divergence, i.e. recurrence $\rightarrow$ convergence) depending on the lattice size ($n \in 4\mathbb{N}$) for ECA $105$. Finally, we have identified the theoretical reasons behind convergence towards \textbf{0} and \textbf{1} point attractors which partially validates our finite lattice size experimental findings. However, we are still open about many questions which guides us towards following extensions:
\begin{itemize}
\item In this first report, we have only theorized convergences towards \textbf{0} and \textbf{1} point attractors. We are still open about convergence towards other point attractors considering the notion of primary RMT sets \cite{Sethi901}.

\item Moreover, we have only classified the system into convergence and divergence under skewed environment. The immediate question is for the dynamics of divergence systems; what can be said about recurrence and non-recurrence properties of the divergence systems?

\item Fat\`{e}s \cite{BoureFC12} has identified peculiar phase transition behaviour for ECAs $26$,~$38$,~$58$,~$134$ for changing value of $\alpha$. Here also, these ECAs ($26$,~$38$,~$58$,~$134$) show phase change (convergence to divergence) in the presence or absence of atomicity property. However, we are still open to identify atomicity property as (one of the) reason behind phase transition. Most importantly, in this direction, the hidden agenda is to understand probabilistic system ($\alpha$-asynchronism) following the dynamics of (kind of) deterministic environment (skewed environment). 
\end{itemize}
\chapter{Elementary Cellular Automata under $\alpha$ - Asynchronous Update Scheme}
\label{chap5}

\section{Introduction}
\label{sec:S1}

Cellular Automata (CAs) are discrete, abstract computational systems that have been studied extensively for their applications in modeling complex systems and processes~\cite{Sethi2016,Anindita12}. Asynchronous Cellular Automata (ACAs) introduce a variation where the updating of cells is not simultaneous, adding a layer of realism for certain applications. Alpha Asynchronous Cellular Automata ( $\alpha$ - ACAs) are a further specialization where the asynchrony is governed by a probability parameter  $\alpha$ ~\cite{Sethi2016,Anindita12,CARONLORMIER2008522}. In Chapter~\ref{chap2}, we have introduced asynchronous cellular automata in detail, in which no global clock exists.  The $\alpha$ - asynchronous cellular automata are a special class of asynchronism which are asynchronous by virtue of their nature of non - uniformity in update of states of cells. Though some of the reports were already published and available in literature. But in those works, no any explicit commentary were made related to the reversible and convergent nature of ECAs under $\alpha - $ asynchronism~\cite{ref_r3,ref_r4,ref_r6,Sethi2016,Anindita12,CARONLORMIER2008522}. In this class, the $\alpha$ is called asynchrony rate which can be as low as 0 and as high as 1. 

In previous chapter~\ref{chap4}, we have presented the important results produced under the application of skew asynchronous cellular automata, while here we study, examin, and analyze their dynamics under this very powerful class of asynchronous cellular automata. However, we have tried to theorize the logic working behind the reversibility of an ECA $R$ under the $\alpha$ - ACAs,as well as the convergent behaviour toward all - \textbf{0} or all - \textbf{1} of a given ECA $R$.

\section{Experimental Setup, $\alpha - $ACAs}
\label{exp_setup}
This experimental study were done under the following set-up:
As we already have discussed about the $\alpha$- ACAs in previous section and Chapter~\ref{chap2}, therefore here we only will talk about the experimental set-up and observations. In this experimental study, we have followed~\cite{ROY2019600} to conduct both qualitative and quantitative study. Following this, we utilized python language and anaconda jupyter notebook environment for performing computational experiment. We used lattice of 50 and 51 cells, 80, 1040, and 2000 iterations, and $\alpha \in \{0.1, 1.0\}$ for each CAs rule to produce the results. The all the space-time diagram~\ref{Fig1}, and~\ref{Fig2} were collected under this environment. (a) First of all, we experimented to collect quantitative data associated to each CAs rule for each $\alpha \in \{0.1, 1.0\}$. For this procedure, we used $n \in [4,5,6,7,8,9,10,11,12,..50]$ with iterations in the range of $2^{n}*100$. It resulted a clear quantitative view of dynamics of each CAs rule for each $\alpha$. We observed the experimental values and discussed the important findings in Section~\ref{sec:S2}. (b)Next, we performed the same experiment to collect the space time diagram for each minimal ECA, which can show the qualitative visuals of the dynamical behaviour under evolutions throughout the time displayed in Fig~\ref{Fig1}, and ~\ref{Fig2}. The black cells indicate bit 1 and white cell indicates bit 0 in the space-time diagram. The observations have been summarized into a tabulated form and depicted into the Table~\ref{table:tab1} and Table~\ref{table:tab2}. Algorithm~\ref{algo2} is implemented to collect the results for proper analysis of dynamical behaviours of ECAs. We explain this asynchronism in CAs with the help of the following example which have been explained in Section~\ref{subsec:alphaACA}.

\textbf{Example 2.1:} Suppose, we are given an $ECA - 90$, ring size $n = 5$, initial configuration $(1, 1, 0, 0, 1)$ and $\alpha$ = 0.5. Then, we are interested to see the next four configuration after evolved from  the given and subsequent evolved configurations.

Under the consideration of $\alpha -$ asynchronous CA update scheme, suppose that cell number $0$ and $1$ (i.e. two immediate neighbors) cells have $\alpha$ probability. So, therefore, these cells are select uniformly and randomly for update corresponding to the update rules based on the given $ECA - 90$. Then, using periodic boundary condition, the input patter based on the select cell number $0$ is $(1, 1, 1)$. Corresponding to this input pattern, the output pattern is $0$ for this cell. While, input pattern for cell number $1$ is $(1, 1, 0)$ and output pattern corresponding to this input pattern is $1$. Therefore, the first next configuration becomes: $(0, 1, 0, 0, 1)$. The other three next configurations have been summarized into the Table 2.2.
\newpage
\begin{table}[h]
	\centering
		\resizebox{0.8\textwidth}{!}{
		\begin{tabular}{|c|c|c|c|c|c|c|c|c|c|}
			\hline
{\bfseries Cell Number } & {\bfseries $0$} & {\bfseries $1$}& {\bfseries $2$} & {\bfseries $3$} & {\bfseries $4$} & {\bfseries Pr($\alpha$) - cells to be updated}\\\hline
Present State & 1 & 1 & 0 & 0 & 1 & \\\hline
Next State(i) & 0 & 1 & 0 & 0 & 1 & 0, 1\\
Next State(ii)& 0 & 1 & 1 & 1 & 1 & 0, 2, 3\\

Next State(iii)& 0 & 1 & 0 & 1 & 1 & 1\\

Next State(iv) & 0 & 0 & 0 & 1 & 1 & 0, 1, 2, 3, 4\\\hline
	\end{tabular}}
	\caption{The next four configurations after evolution of initial configuration (1, 1, 0, 0, 1) under the $\alpha$ fully asynchronous cellular automata and ECA - 90}\label{tab1}
\end{table}

\begin{algorithm}[h]
\caption{Algorithm to identify dynamical behaviour of ECAs under $\alpha - $ asynchronous update.}
\textbf{Step 1:} Read ECA rule, CA lattice size(n),alpha\\
\textbf{Step 2:} Execute the main function for total iterations up to 2$^n$*100.
$\space \hspace{1cm}$ def AlphaSyncCA(CA1,appliedRule,alpha)$:$

$\space \hspace{1.5cm}$    CA, size, temp = CA1, len(CA),[]

$\space \hspace{1.5cm}$    for state in CA$:$

$\space \hspace{2cm}$            temp.append(state) 
   
$\space \hspace{2cm}$    for j in range(int(size))$:$

$\space \hspace{2.5cm}$        if(alpha$\geq$random.random())$:$

$\space \hspace{3cm}$            if j==0$:$

$\space \hspace{3.5cm}$ Li,Ri = size - 1, 1

$\space \hspace{3cm}$            elif j==size-1$:$

$\space \hspace{3.5cm}$                Li,Ri = size - 2,0

$\space \hspace{3cm}$            else$:$

$\space \hspace{3.5cm}$                Li,Ri = j - 1,j+1

$\space \hspace{3cm}$            if CA[Li] == 1 and CA[j] == 1 and CA[Ri] == 1$:$

$\space \hspace{3.5cm}$                temp[j] = int(appliedRule[0]) 

$\space \hspace{3cm}$            elif CA[Li] == 1 and CA[j] == 1 and CA[Ri] == 0$:$

$\space \hspace{3.5cm}$                temp[j] = int(appliedRule[1]) 

$\space \hspace{3cm}$            elif CA[Li] == 1 and CA[j] == 0 and CA[Li] == 1$:$

$\space \hspace{3.5cm}$                temp[j] = int(appliedRule[2]) 

$\space \hspace{3cm}$            elif CA[Li] == 1 and CA[j] == 0 and CA[Ri] == 0$:$

$\space \hspace{3.5cm}$                temp[j] = int(appliedRule[3])
 
$\space \hspace{3cm}$            elif CA[Li] == 0 and CA[j] == 1 and CA[Ri] == 1$:$

$\space \hspace{3.5cm}$                temp[j] = int(appliedRule[4]) 

$\space \hspace{3cm}$            elif CA[Li] == 0 and CA[j] == 1 and CA[Ri] == 0$:$

$\space \hspace{3.5cm}$                temp[j] = int(appliedRule[5])
 
$\space \hspace{3cm}$            elif CA[Li] == 0 and CA[j] == 0 and CA[Ri] == 1$:$

$\space \hspace{3.5cm}$                temp[j] = int(appliedRule[6]) 

$\space \hspace{3cm}$            elif CA[Li] == 0 and CA[j] == 0 and CA[Ri] == 0$:$

$\space \hspace{3.5cm}$                temp[j] = int(appliedRule[7])

$\space \hspace{3cm}$        else$:$

$\space \hspace{3.5cm}$            temp[j] = CA[j]

$\space \hspace{1.5cm}$    return temp
\end{algorithm}

\begin{algorithm}[h]
$\space \hspace{1cm}$ def update(cur, rule, i, alpha)$:$

$\space \hspace{1.5cm}$     nxt=cur

$\space \hspace{1.5cm}$     rule1=rules(rule)

$\space \hspace{1.5cm}$     nxt=AlphaSyncCA(cur,rule1,alpha)

$\space \hspace{1.5cm}$     return nxt

$\space \hspace{1cm}$ def main(dimx, rule, att, alpha,itScale)$:$

$\space \hspace{1.5cm}$     countSameConf = 0

$\space \hspace{1.5cm}$     initConf   = np.zeros((1, dimx), dtype=int)

$\space \hspace{1.5cm}$     totlConfs = pow(2,dimx)

$\space \hspace{1.5cm}$     totlIteratons = totlConfs$*$itScale
        
$\space \hspace{1.5cm}$     for iConf in range(totlConfs)$:$

$\space \hspace{2.0cm}$        initConf = [int(x) for x in np.binaryrepr(iConf, dimx)]

$\space \hspace{2.0cm}$         nxtConf,counter = initConf,0

$\space \hspace{2.0cm}$         while counter$<$totlIteratons$:$

$\space \hspace{2.5cm}$             nxtConf = update(nxtConf, rule, counter, alpha)

$\space \hspace{2.5cm}$             if(nxtConf == initConf)$:$

$\space \hspace{3cm}$                 countSameConf = countSameConf + 1                

$\space \hspace{3cm}$                 break

$\space \hspace{2.5cm}$            counter+=1

$\space \hspace{2.0cm}$         if(countSameConf == totlConfs)$:$

$\space \hspace{2.5cm}$         att.write(Reversible)

$\space \hspace{2.0cm}$     else$:$

$\space \hspace{2.5cm}$         att.write(Not Reversible)      
\end{algorithm}
\label{algo2}

\section{Reversible Nature of ECAs}
\label{sec:S2}
\begin{figure}[h!tbp]
\begin{center}
\includegraphics[width=100mm, height=80mm]{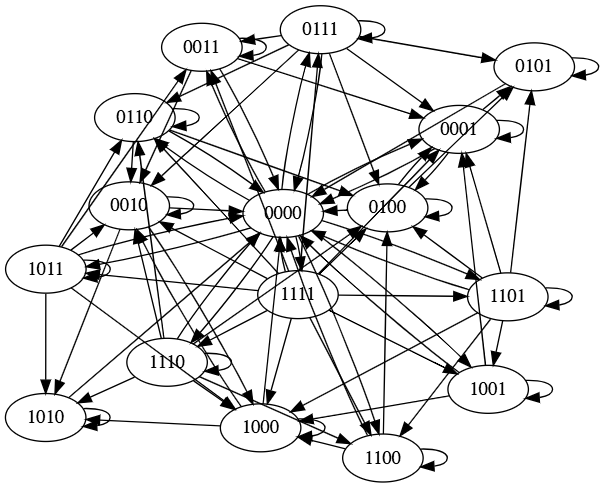} \\
\caption{\small{Transition diagram of ECA - 1 under $\alpha$ - ACA update scheme, where $n$ = 4 and $\alpha$ = 0.5.}}
\label{Fig1}
\end{center}
\end{figure}
 
\begin{table}[h!]
\centering
\begin{tabular}{c c c||c c c }\hline
\textbf{ECA} & $\alpha$ &\textbf{n}$\in \{4,5,\dots,51\}$ & \textbf{ECA} &  $\alpha$ & \textbf{n}$\in \{4,5,\dots,51\}$\\\hline
 1 &  0.2 - 0.9  & $n$                   & 57 & 0.1 - 0.9   &$n$ \\
 3 & ,,                 & ,,            & 105  & ,, &$n \neq  4i$\\
 7 & ,, & $n = 2i+1$& 142 & ,,          & $n$          \\
 9 & ,,  & ,,                           & 30  & 0.1 - 0.5   & $n = 2i+1$\\
 11 & ,, & ,,                           & 46  & ,,          & $n$\\
 19 & 0.1 - 0.9 & ,,                    & 60  & ,,          & ,,\\
 23 & ,, &$n = 2i+1$ & 62  & ,,          & ,,\\
 25 &,,   & $n$                         & 51  & 0.1 - 1.0   & ,,\\
 27 & ,,   & ,,                         & 204 & ,,          & ,, \\
 33 & ,,   & ,,                         & 15  &  1.0        & ,,\\
 35 & ,,   & ,,                        & 105 & ,,          & $n \neq 4i$ and $\neq 2i$  \\
 37 & ,,   & ,,                         & 170 & ,,          & ,,\\
 41 & ,,   & ,,                         & 106 & 0.1 - 0.2  & ,,\\
 43 & ,,   & ,,                         & 108 & ,,          & ,,\\
 45 & ,,   & ,,                        & 134 & ,,          & ,,\\\hline
\end{tabular}
\vspace{0.5em}
\caption{Reversible ECA R under the $\alpha$ - asynchronous update scheme.}
\label{table:tab1}
\end{table}

\begin{figure}[h!tbp]
\begin{center}
\begin{tabular}{cccccc}
ECA - $1$ & 
\includegraphics[width=25mm]{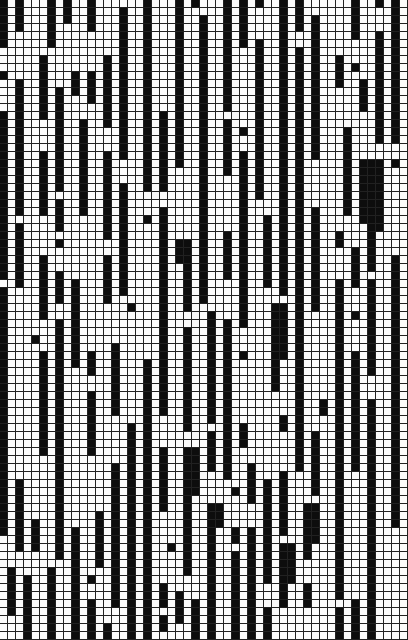} & 
\includegraphics[width=25mm]{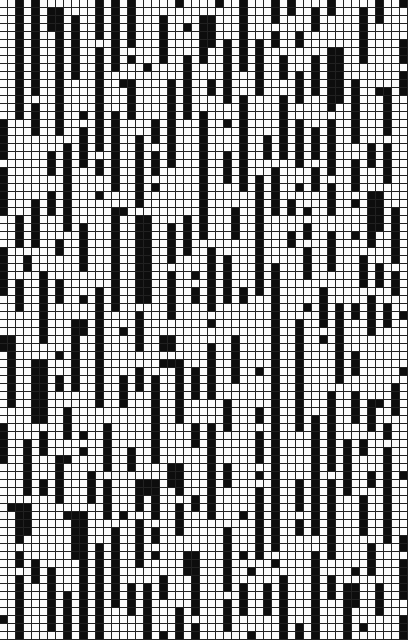} &  
\includegraphics[width=25mm]{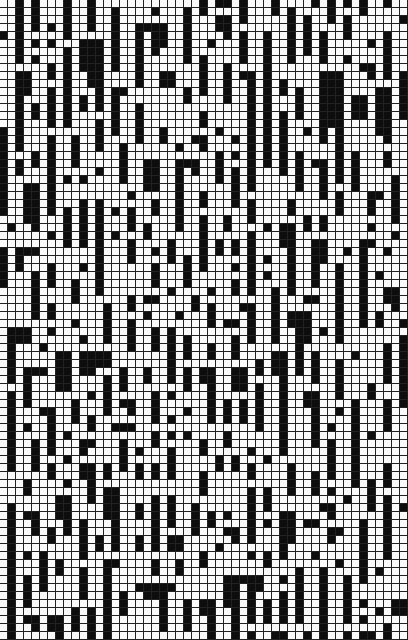} & 
\includegraphics[width=25mm]{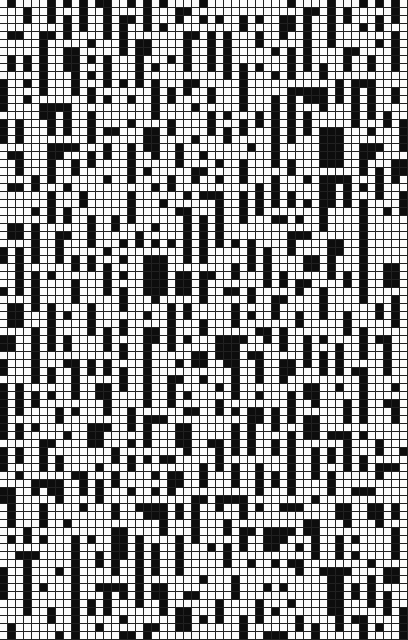} &  
\includegraphics[width=25mm]{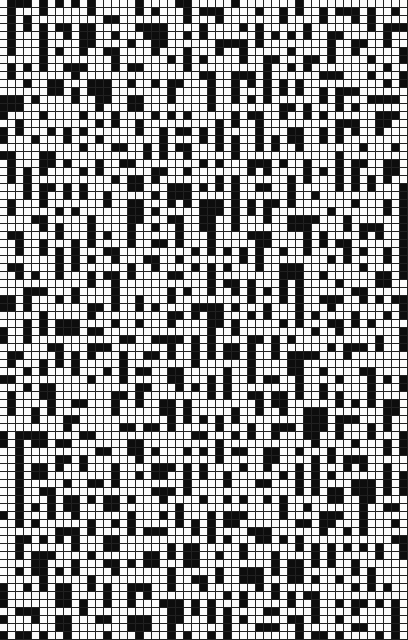} \\
  & $\alpha$ = 0.1 & $\alpha$ = 0.2 & $\alpha$ = 0.3 & $\alpha$ = 0.4 & $\alpha$ = 0.5\\
ECA - $1$  & 
\includegraphics[width=25mm]{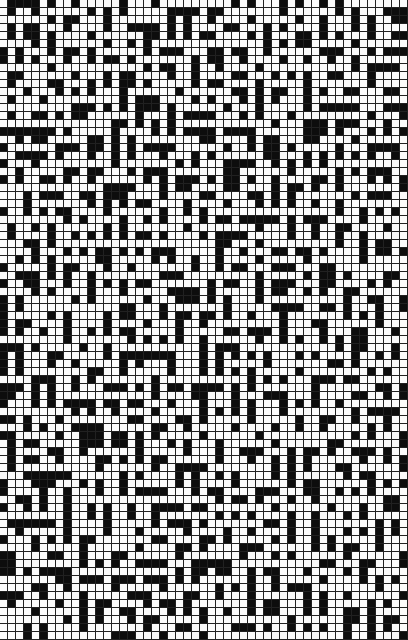} & 
\includegraphics[width=25mm]{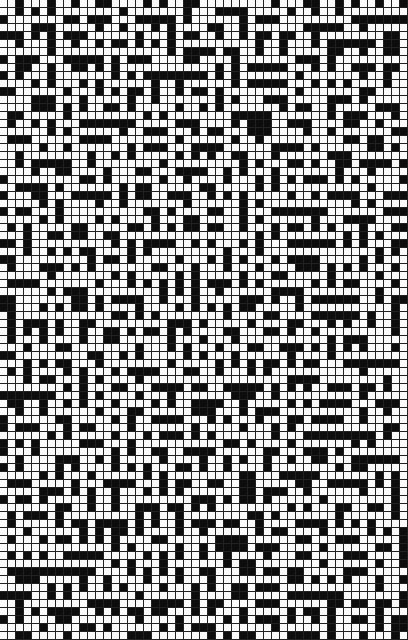} &  
\includegraphics[width=25mm]{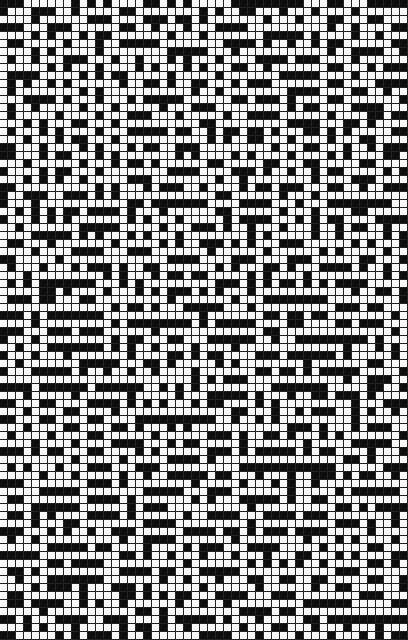} & 
\includegraphics[width=25mm]{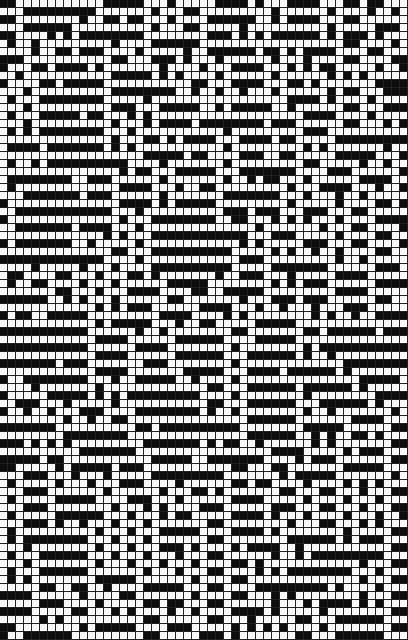} & 
\includegraphics[width=25mm]{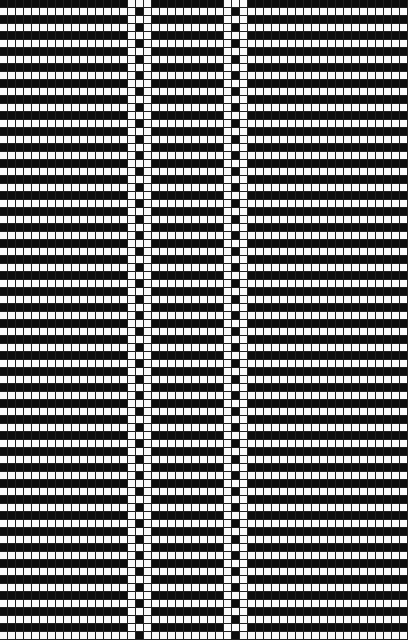} \\
  & $\alpha$ = 0.6 & $\alpha$ = 0.7 & $\alpha$ = 0.8 & $\alpha$ = 0.9 & $\alpha$ = 1.0\\
ECA - $3$ & \includegraphics[width=25mm]{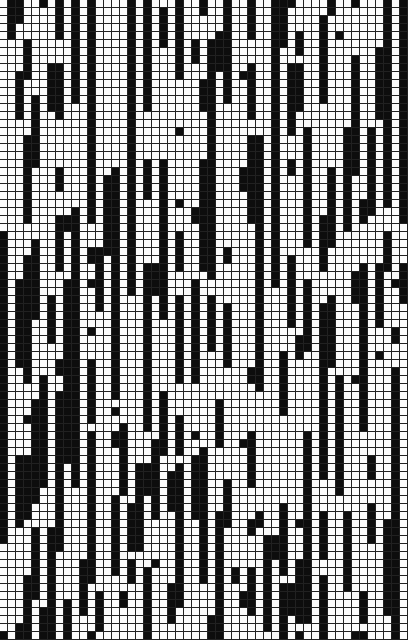} & 
\includegraphics[width=25mm]{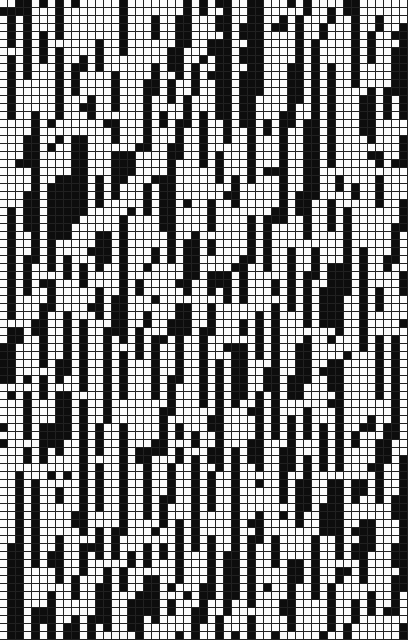} &  
\includegraphics[width=25mm]{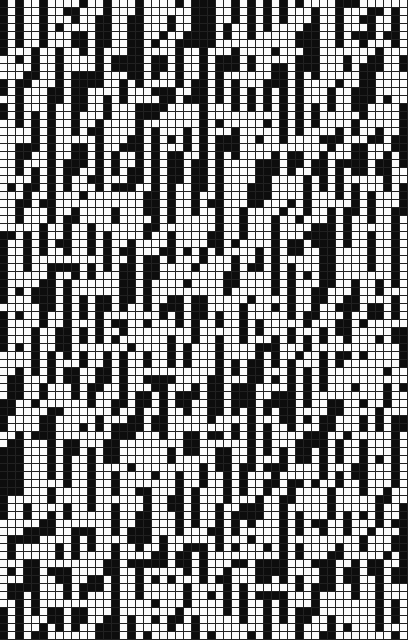} & 
\includegraphics[width=25mm]{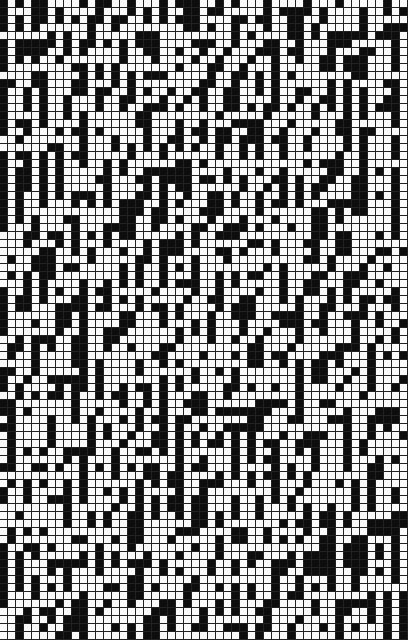} &  
\includegraphics[width=25mm]{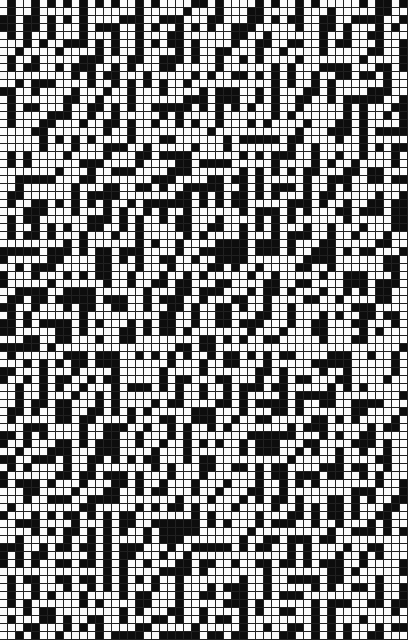} \\
  & $\alpha$ = 0.1 & $\alpha$ = 0.2 & $\alpha$ = 0.3 & $\alpha$ = 0.4 & $\alpha$ = 0.5\\

ECA - $3$& \includegraphics[width=25mm]{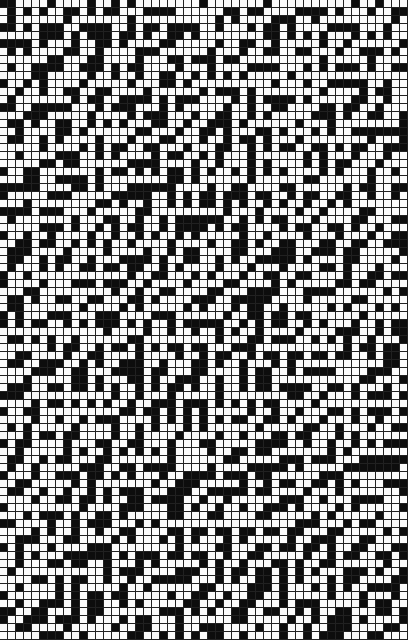} & 
\includegraphics[width=25mm]{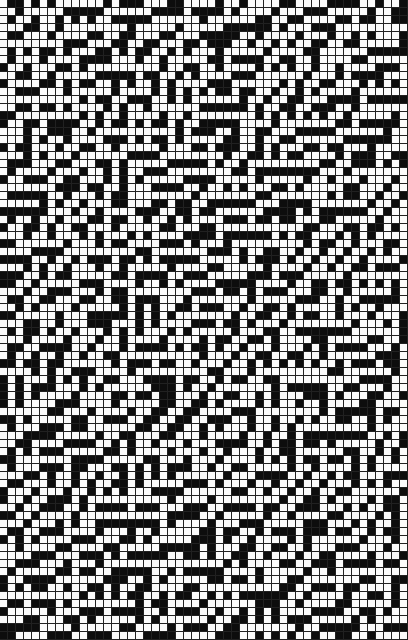} &  
\includegraphics[width=25mm]{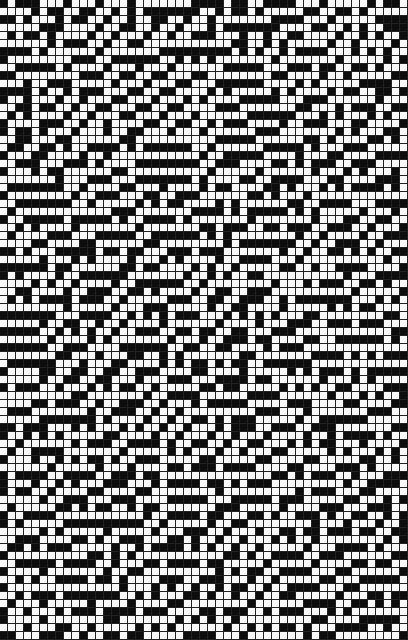} & 
\includegraphics[width=25mm]{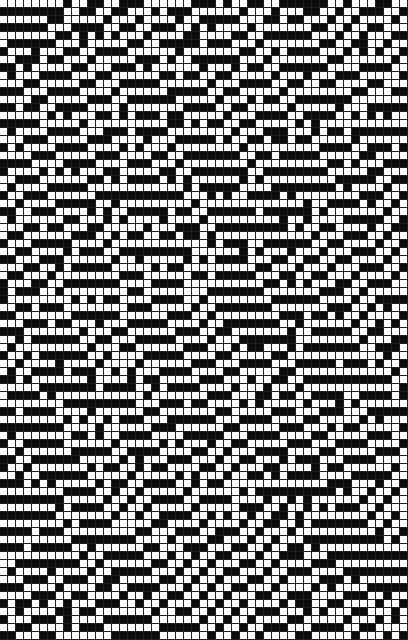} & 
\includegraphics[width=25mm]{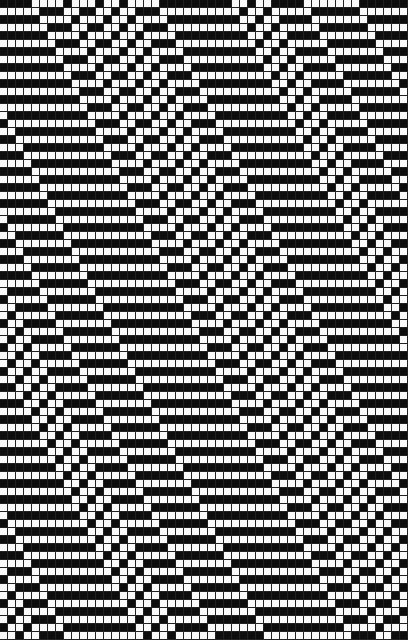} \\
  & $\alpha$ = 0.6 & $\alpha$ = 0.7 & $\alpha$ = 0.8 & $\alpha$ = 0.9 & $\alpha$ = 1.0\\
\end{tabular}
\caption{\small{Dynamics of ECAs $1$, and $3$  for $\alpha$- asynchronous systems. Here, $n = 51$ where the time goes form top to bottom. Each line represent the configuration of automata after $1,~n,~n/2$ updates respectively for changing updating schemes. Black and white squares represent cells in state $1$ and $0$. This convention is kept in the rest of the text.}}
\label{Fig2}
\end{center}
\end{figure}

\noindent \textbf{Observation 5.3.1} \textit{ An ECA R under the $\alpha$ - asynchronous update scheme is irreversible if anyone of the configurations all - $0$ and all - $1$ is a transient (acyclic) configuration.}. 
  
\noindent \textbf{Observation 5.3.2} \textit{An ECA under $\alpha$-ACA, where $\alpha \in $ $\{0.1, 0.9\}$, or $\{0.2, 0.9\}$ shows reversible dynamics if any one of the following is true corresponding to a given ECA rule R}.
\begin{itemize}
\item[1.] RMTs $0, 7$, at least one of $2, 5$ must be active such that any RMT of $1, 4, 3, 6$ can also be active, or.

\item[2.] RMT $6$, and either $0, 2$ or $1$ is active; and $5, 7$ are passive or.

\item[3.] RMTs $3, 4$ active and remaining others are passive, or.
\end{itemize}\

\noindent \textbf{Observation 5.3.3} \textit{Since ECAs show reversible dynamics for $n \in \{4,5,\dots,51\}$, therefore we can observe that they can be reversible for $n\in \mathbb{N}$ where $n \geq$ 4.}

These conditions highlight the intricate nature of reversibility in ECAs, influenced by both the rule applied and specific properties of the system size and density. Understanding these nuances is crucial for applications that require reversibility, such as cryptography and reversible computing.

\section{Convergent Nature of ECAs}
We have discussed findings and observations which focuses on the dynamics of ECA $R$ those are either reversible or convergent under the $\alpha \in [0.1, 0.9]$. However, we need to explore the logics related to the dynamics of other ECAs those behave sometimes as reversible and sometimes convergent for a specific range of $\alpha$.
\newpage
\begin{figure}[h!tbp]
\begin{center}
\includegraphics[width=14cm, height=12cm]{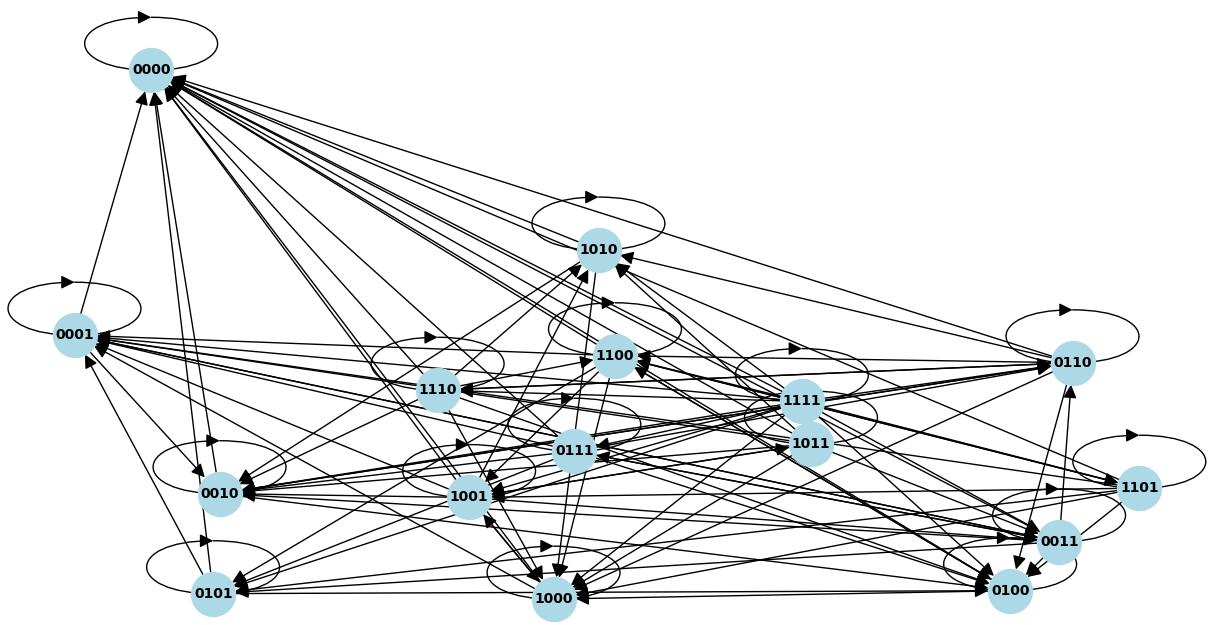} \\
\caption{\small{Transition diagram of ECA - 2 under $\alpha$ - ACA update scheme, where $n$ = 4 and $\alpha$ = 0.5.}}
\label{Fig1}
\end{center}
\end{figure}
\begin{table}[h!]
\centering
\begin{adjustbox}{max width=\textwidth}
\begin{tabular}{c c c c||c c c c }\hline
\textbf{ECA} & \textbf{C/NC/NR} & \textbf{n}$\in\{4,5,\dots,51 \}$ &  $\alpha$ & \textbf{ECA} &\textbf{C/NC/NR} & \textbf{n}$\in\{4,5,\dots,51 \}$& $\alpha$\\\hline
 0   & C & $n$ & 0.1 - 1.0                       & 2   & C & $n$ & 0.1 - 0.9\\
 4   & C & ,, & ,,                               & 6   & C & ,, & ,,\\
 5   & C & ,, & ,,                               & 10  & C & ,,& ,,\\
 7   & C & $n = 2i$& ,,     & 18  & C & ,,& ,,\\
 8   & C & $n$  & ,,                              & 24  & C & ,,& ,,\\
 12  & C & ,, & ,,                               & 30  & C & $n  = 5i$& ,,\\
 13  & C & ,, & ,,                               & 34  & C & $n$& ,,\\
 14  & C & ,, & ,,                               & 38  & C & ,,& ,,\\
 15  & C & ,, & ,,                               & 42  & C & ,,& ,,\\
 23  & C & $n = 2i$& ,,     & 50  & C & ,,& ,,\\
 28  & C & $n$ & ,,                               & 56  & C & ,,& ,,\\
 29  & C & ,,& ,,                               & 74  & C & ,,& ,,\\
 32  & C & ,, & ,,                               & 130 & C & ,,& ,,\\
 36  & C & ,, & ,,                               & 138 & C & ,,& ,,\\
 40  & C & ,, &,,                                & 152 & C & ,,& ,,\\
 44  & C & ,, &,,                                & 154 & C & ,,& ,,\\ 
 72  & C & ,, & ,,                               & 162 & C & ,,& ,,\\
 76  & C & ,, & ,,                               & 184 & C & ,,& ,,\\
 77  & C & ,, &,,                                & 26 & C & ,, & 0.1 - 0.5\\
 78  & C & ,, & ,,                               & 50 & C & ,, & ,,\\
 94  & C & ,, & ,,                               & 58 & C & ,, & ,,\\
 104 & C & ,, & ,,                               & 26 & NC,NR & ,,& 0.6 - 0.9\\
 128 & C & ,, & ,,                               & 50 & NC,NR & ,,& ,,\\
 132 & C & ,, & ,,                               & 58 & NC,NR & ,,& ,,\\
 136 & C & ,, & ,,                               & 106 & C & ,,& 0.3 - 0.9\\
 140 & C & ,, &,,                                & 108 & C & ,,& ,,\\
 156 & C & ,, & ,,                               & 46 & NC,NR & ,,& 0.6 - 1.0\\
 160 & C & ,, & ,,                               & 60 & NC,NR & ,,& ,,\\
 164 & C & ,, & ,,                               & 62 & NC,NR & ,,& ,,\\
 168 & C & ,, & ,,                               & 54 & NC,NR & ,,&0.1 - 1.0\\
 172 & C & ,, & ,,                               & 122 & NC,NR & ,,& ,,\\
 200 & C & ,, & ,,                               & 129 & NC,NR & ,,& ,,\\
 204 & C & ,, & ,,                               & 137 & NC,NR & ,,& ,,\\
 232 & C & ,, & ,,                               & 150 & NC,NR & ,, & ,,\\
 146 & C & ,, & 0.1 - 0.7                       & 146 & NC,NR & ,, & 0.8 - 1.0\\
 178 & C & ,, & 0.1 - 0.5                       & 178 & NC,NR & ,, & 0.6 - 1.0\\
 122&NC,NR &,, & 0.1 - 1.0 & 105&NC,NR & $n = 4i$ & 1.0\\\hline
\end{tabular}
\end{adjustbox}
\vspace{0.5em}
\caption{Convergent and non convergent - non reversible ECA R under the $\alpha$ - asynchronous update scheme. Here, C denotes convergent, NC denotes non convergent and NR denotes non reversible dynamics.}
\label{table:tab2}
\end{table}

\begin{figure}[h!tbp]
\begin{center}
\begin{tabular}{cccccc}
ECA - 2 & 
\includegraphics[width=25mm]{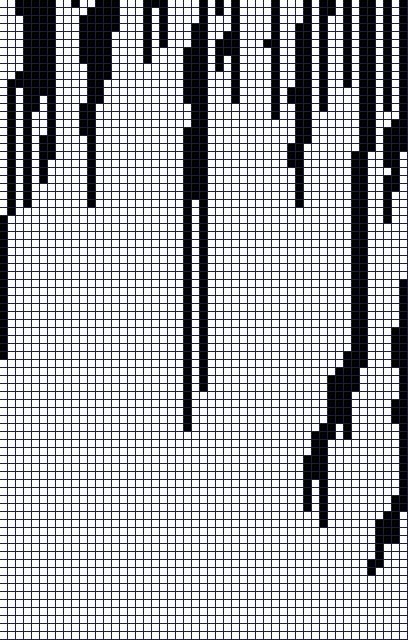} & 
\includegraphics[width=25mm]{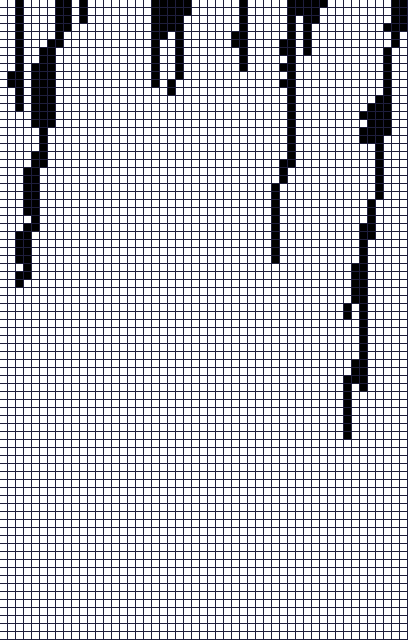} &  
\includegraphics[width=25mm]{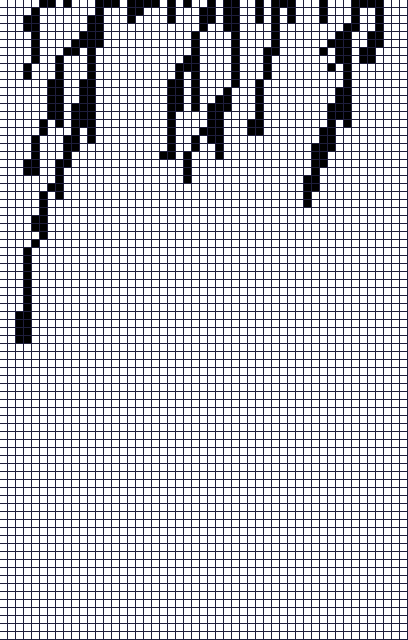} & 
\includegraphics[width=25mm]{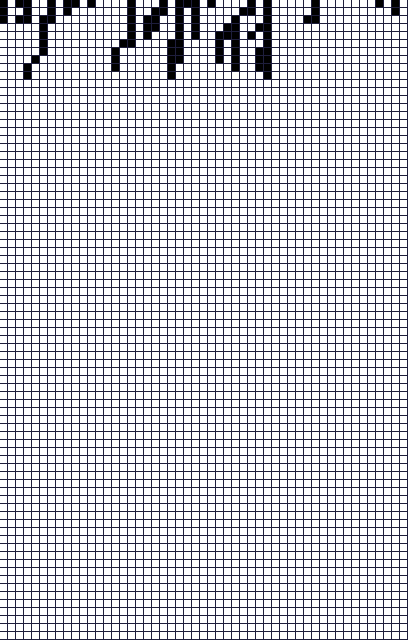} &  
\includegraphics[width=25mm]{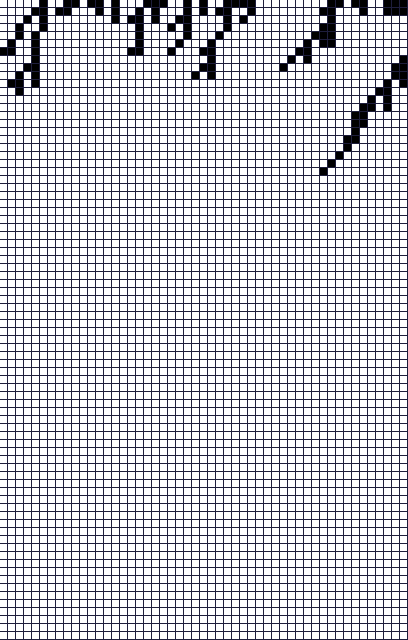} \\
  & $\alpha$ = 0.1 & $\alpha$ = 0.2 & $\alpha$ = 0.3 & $\alpha$ = 0.4 & $\alpha$ = 0.5\\
ECA - 2 & 
\includegraphics[width=25mm]{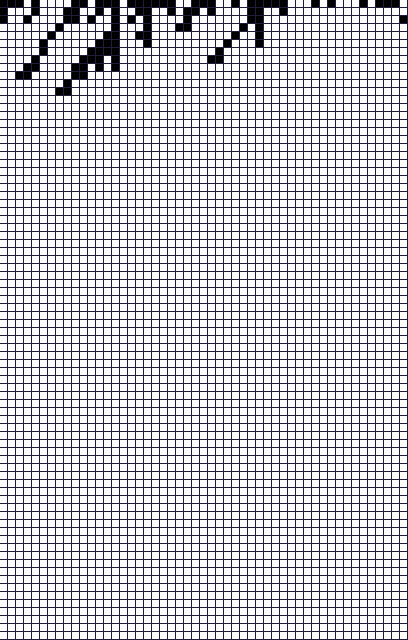} & 
\includegraphics[width=25mm]{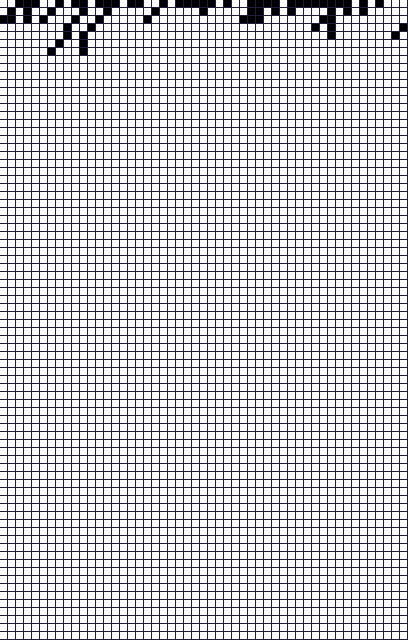} &  
\includegraphics[width=25mm]{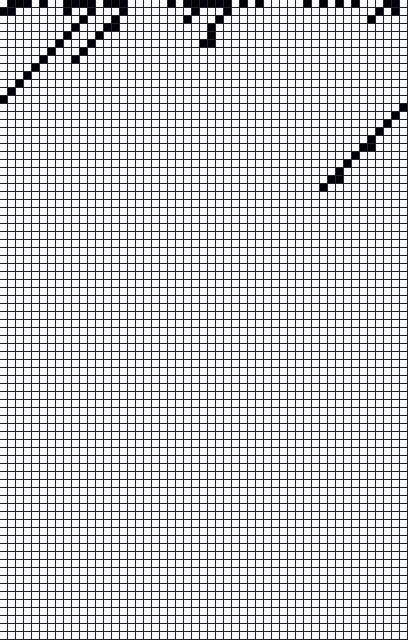} & 
\includegraphics[width=25mm]{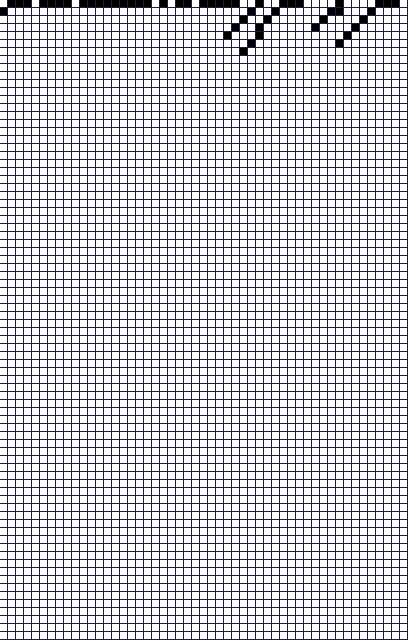} & 
\includegraphics[width=25mm]{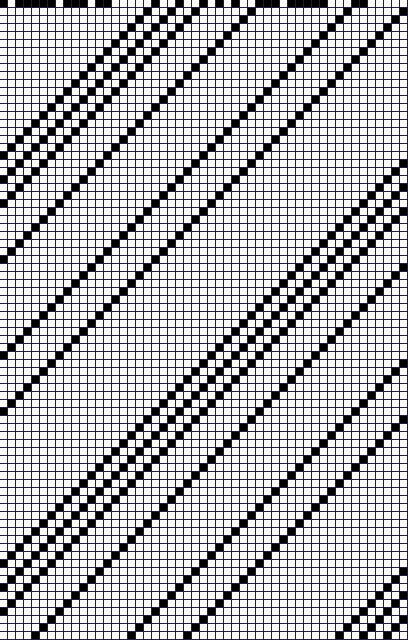} \\
 & $\alpha$ = 0.6 & $\alpha$ = 0.7 & $\alpha$ = 0.8 & $\alpha$ = 0.9 & $\alpha$ = 1.0\\
ECA - 56 & 
\includegraphics[width=25mm]{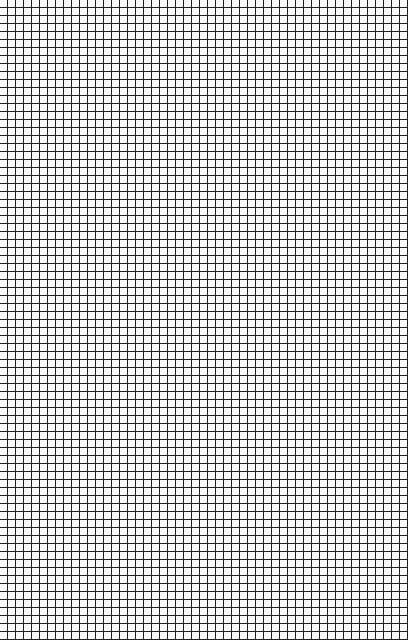} & 
\includegraphics[width=25mm]{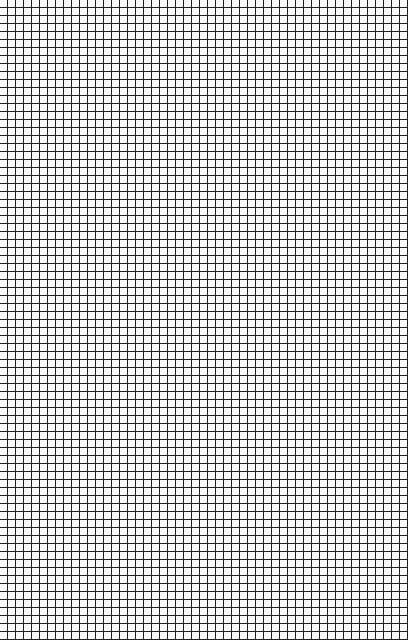} &  
\includegraphics[width=25mm]{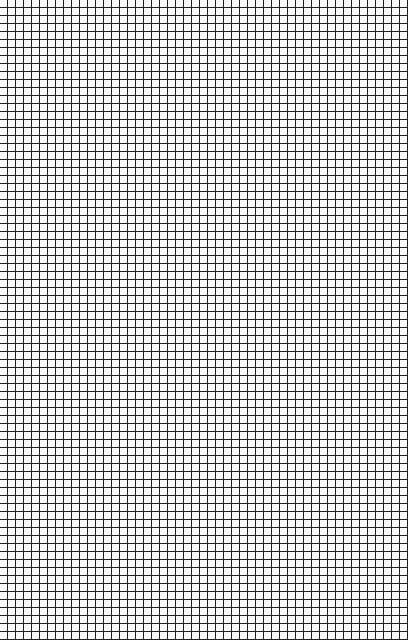} & 
\includegraphics[width=25mm]{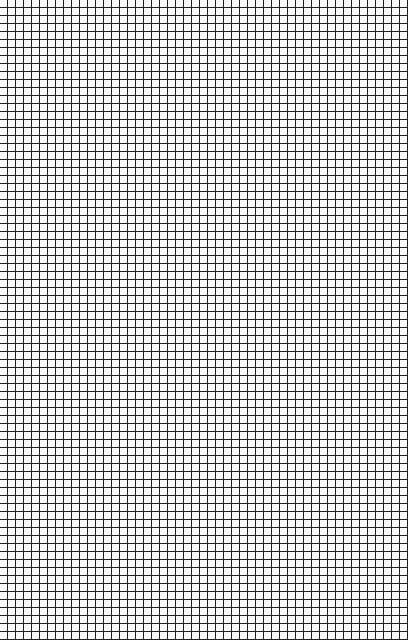} &  
\includegraphics[width=25mm]{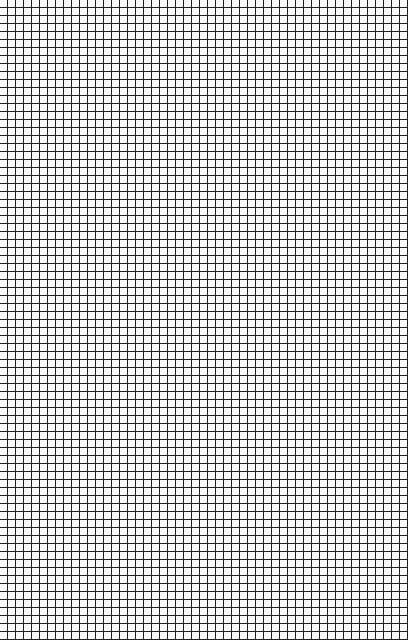} \\
  & $\alpha$ = 0.1 & $\alpha$ = 0.2 & $\alpha$ = 0.3 & $\alpha$ = 0.4 & $\alpha$ = 0.5\\
  ECA - 56 & 
\includegraphics[width=25mm]{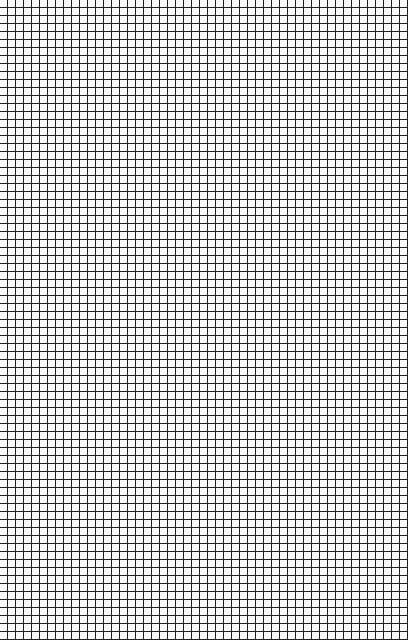} & 
\includegraphics[width=25mm]{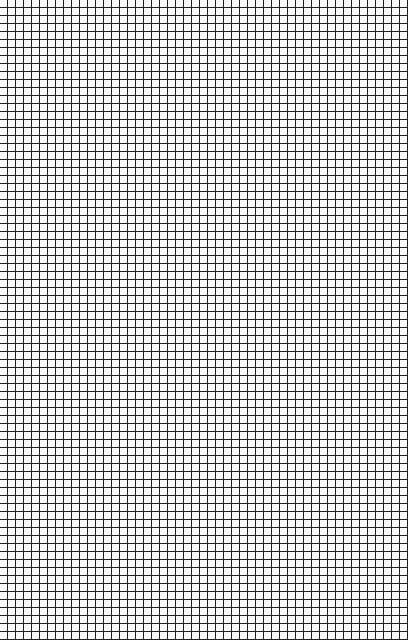} &  
\includegraphics[width=25mm]{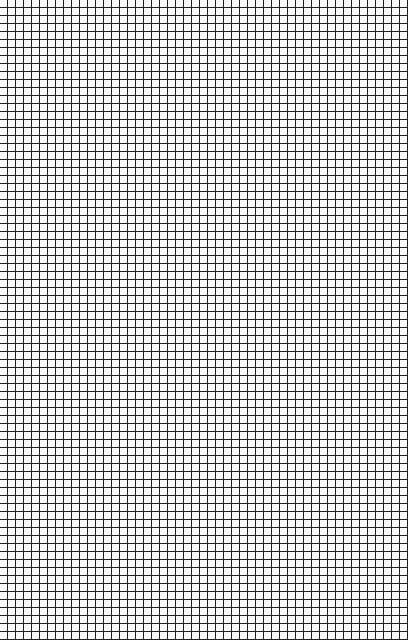} & 
\includegraphics[width=25mm]{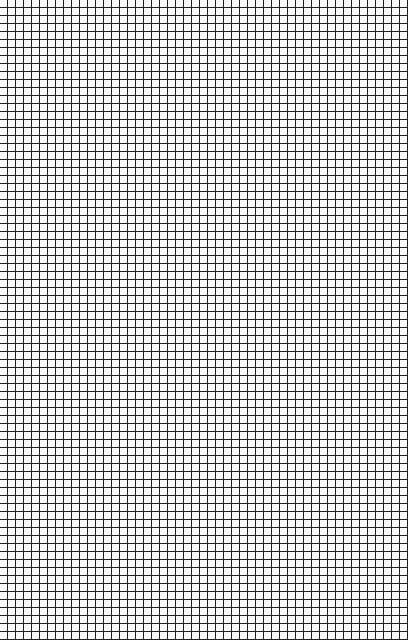} & 
\includegraphics[width=25mm]{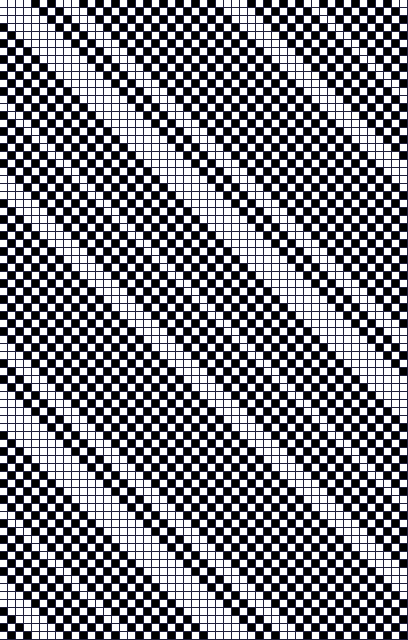} \\
 & $\alpha$ = 0.6 & $\alpha$ = 0.7 & $\alpha$ = 0.8 & $\alpha$ = 0.9 & $\alpha$ = 1.0\\
\end{tabular}
\caption{\small{Dynamics of ECAs $2$, and $56$  for $\alpha$- asynchronous systems. Here, $n = 51$ where the time goes form top to bottom. Here, ECA $2$ shows convergent behaviour on $80^{th}$ iteration, while the ECA $56$ shows it on $1040^{th}$ iteration. Each line represent the configuration of automata after $1,~n,~n/2$ updates respectively for changing updating schemes. Black and white squares represent cells in state $1$ and $0$. This convention is kept in the rest of the text. $\alpha$ = $1$ shows synchronous behaviour.}}
\label{Fig3}
\end{center}
\end{figure}

\noindent \textbf{Observation 5.4.1} \textit{Since ECAs show convergent dynamics for $n \in \{4,5,\dots,51\}$, therefore we can observe that they can be convergent for $n\in \mathbb{N}$ where $n \geq$ 4.}

\noindent \textbf{Observation 5.4.2} \textit{A rule R for $\alpha$-ACAs, where $\alpha \in \{0.1, 0.9\}$ , shows convergent dynamics to all - $0$ or all - $1$ if anyone case is true for a given rule R:}
\begin{itemize}
\item[1.] RMT $0$ is passive, one of the $2, 3$ is active such that:
\begin{itemize}
    \item[A.]$5$ and one of the $1, 4$ is passive, or.
    \item[B.]$5$ and one of the $4, 6$ is active, and $1$ is passive or.
\end{itemize}
\item[2.] RMTs $0, 7$ passive, $5$ and one of the $2,3,4,6$ are active, or.  

\item[3.] RMT $7$ is passive, $0, 5$ and one of $1, 4$ are active.
\end{itemize}

\section{Final Words}
In summary, we have explored $\alpha-$ACAs on minimal representative 88 ECAs. We studied the effect of $\alpha$ on the change into the dynamical behaviour of ECAs. In this study, we particularly focused on the (i) Reversibility and (ii) Convergent behaviour of ECAs for the $\alpha \in \{0.1, 0.9\}$. We observed that under this environment for $\alpha \in \{0.1, 0.9\}$, the ECAs $1$, $3$, $7$, $9$, $11$, $19$, $25$, and $27$ are reversible those are not reversible either under fully ACAs or skew ACAs. Though, we have theorized the logic working behind the reversibility and convergent behaviour of the ECAs those are either reversible or convergent to all-\textbf{0} or all-\textbf{1} under the above range of $\alpha$, we still need to establish theory for the logic working for these two dynamics for remaining ECAs. To sum up, we also have observed the two properties of the ECAs under the $\alpha$ - ACAs as pointed out below:
\begin{itemize}
\item[•] \textit{ An ECA R under the $\alpha$ - asynchronous update scheme is irreversible if anyone of the configurations all - $0$ and all - $1$ is a transient (acyclic) configuration.}. 
\item[•] An ECA $R \in \{ $0$, 2^{k} \}$, where $k\in\mathbb{N}$, will always converge to either all - \textbf{0} or all - \textbf{1}.
\end{itemize}
One can keenly observed these two propertied by performing experiment to establish theory and proof for these two properties.
\chapter{Conclusion}
\label{chap7}
This chapter aims to summarize the thesis' important contributions and to discuss the potential future directions for study in the subject of cellular automata.
\begin{itemize}
	\item[•] This thesis presents discussion about the asynchronous cellular automata with respect to the philosophy of asynchronism, type of the asynchronism and their applications in various dimensions of diverse fields.
	\item[•] This thesis presents a nice application of the fully asynchronous cellular automata to solve a machine learning problem that is clustering of a given data. The fully ACA based clustering algorithm has a potential as equivalent to as the other benchmark clustering techniques such as K-Means, Hierarchical clustering, and PAM. In future, its optimized form can be utilized for to cluster a data, where data can be homogeneous and heterogeneous.
	\item[•] The thesis presents a detailed study of the dynamical behaviour of the ECA under the various kinds of the asynchrous cellular auotmata update schemes. It also presents theoretical concepts with proper proofs, which explicitly tells about the logic working behind a particular dynamics under a given environment and ECA.
	\item[•] The thesis also presents a new kind of asynchronous cellular automata, which we called skewe asynchronocus cellular automata. In the skew cellular automata, exactly two consecutive and adjacent cells are updated simultaneously corresponding to update patterns associated to a given ECA. It first time explicitly tells about the effect of breaking atomicity over the reversibility and convergent dynamics.
	\item The thesis also present a partial study of the ECA under the $\alpha$-ACA update scheme. Fist time, we see that ECAs such as $1, 3, 7, 9, 11, 19, 25$, and $27$ are reversible in nature.
\end{itemize}




\cleardoublepage
\phantomsection
\addcontentsline{toc}{chapter}{Author's Publications}
\chapter*{Author's Publications}
\label{chap8}

\hspace{0.1in}\textbf{Accepted Paper}\\
\begin{itemize}
\item S.~Roy, V.~K.~Gautam and S.~Das, ``A note on Skew Asynchronous Cellular Automata'', ``Cellular Automata and Discrete Complex Systems'', \textit{AUTOMATA - 2024}, (2024).(Paper Accepted).
\end{itemize}

\vspace{0.4in}



\pagestyle{plain}
\bibliographystyle{IEEEtran}

\cleardoublepage
\phantomsection
\addcontentsline{toc}{chapter}{Bibliography}
\bibliography{Thesis}
\end{document}